\documentclass[12pt,a4paper]{article}
\usepackage[utf8]{inputenc}
\usepackage[english]{babel}
\usepackage{graphicx}
\usepackage{amsmath}
\usepackage{amsthm}
\usepackage{amsfonts}
\usepackage{amssymb}
\usepackage{geometry}
\usepackage[longnamesfirst,authoryear]{natbib}
\usepackage{afterpage}
\usepackage{bm}
\usepackage{here}
\usepackage[subrefformat=parens]{subcaption}
\usepackage{threeparttable}
\usepackage{standalone}
\usepackage{textcomp}
\usepackage{clipboard}

\usepackage{booktabs}
\usepackage{multicol}
\usepackage{multirow}
\usepackage{dcolumn}
\newcolumntype{.}{D{.}{.}{-1}}

\usepackage{tikz}
\usetikzlibrary{patterns}
\usetikzlibrary{patterns.meta}
\usepackage{pgfplots}
\pgfplotsset{compat=1.17}

\usepackage{lscape}

\usepackage{hyperref}
\hypersetup{colorlinks=true,linkcolor=darkblue,urlcolor=darkblue,citecolor=darkblue,urlbordercolor={1 1 1}}
\usepackage{amsmath,amsfonts,bbm,amsthm}

\makeatother
\usepackage{babel}
\usepackage{pdfpages}

\usepackage{appendix}

\usepackage{array}
\usepackage{booktabs}
\usepackage{tabularx}

\usepackage{xcolor}
\definecolor{darkblue}{rgb}{0,0,0.5}

\usepackage[T1]{fontenc}%
\usepackage{times}

\usepackage{tikz}
\usepackage{authblk}
\usepackage{indentfirst}

\usepackage{layouts}
\usepackage{amsthm}

\newtheorem{Theorem}{Theorem}
\newtheorem{Definition}{Definition}
\newtheorem{Lemma}{Lemma}
\newtheorem{Corollary}{Corollary}

\newtheorem{Assumption}{Assumption}
\newtheorem{Example}{Example}

\newtheorem{Remark}{Remark}

\usepackage{mathtools}

\tikzstyle{rednode} = [shape=rectangle, fill=red!50, line width=3]
\tikzstyle{bluenode} = [shape=rectangle, fill=blue!50, line width=3]
\tikzstyle{yellownode} = [shape=rectangle, fill=yellow!50, line width=3]

\tikzstyle{dotnode} = [dashed, pattern={Lines[angle=90,distance=3pt]}, pattern color=gray!150]
\tikzstyle{backslashnode} = [dashed, pattern={Lines[angle=45,distance=3pt]}, pattern color=gray!150]
\tikzstyle{slashnode} = [dashed, pattern={Lines[angle=-45,distance=3pt]}, pattern color=gray!150]

\usepackage{geometry}
\usepackage{enumerate}

\makeatletter
\renewcommand\AB@affilsepx{, \protect\Affilfont}
\makeatother

\usepackage{framed, xcolor, colortbl}
\usepackage{comment}
\colorlet{shadecolor}{gray!15}

\usepackage{multibib}
\newcites{SM}{References in the response letters:}

\begin{document}

\author{Takuya Ishihara\footnote{Tohoku University, Graduate School of Economics and Management. Email: takuya.ishihara.b7@tohoku.ac.jp} \and Toru Kitagawa\footnote{Brown University, Department of Economics. Email: toru\_kitagawa@brown.edu}}

\title{\vspace{-2cm}Evidence Aggregation for Treatment Choice{\Large \thanks{%
We would like to thank Tim Armstrong and participants of the 2018 Y-RISE meeting in the Cayman islands for inspiring discussion, which motivated us to initiate this project. We are grateful to Chuck Manski, Jonathan Roth, and J\"{o}rg Stoye for beneficial comments. We have also benefited from comments made by participants of the Cemmap econometrics lunch meeting, Kansai econometrics study group, and the annual meetings of the 2020 SEA, 2021 IAAE, and 2021 AMES. 
Financial support from
the ESRC through the ESRC Centre for Microdata Methods and Practice (CeMMAP)
(grant number RES-589-28-0001) and the ERC (grant number 715940) is gratefully acknowledged. Takuya Ishihara gratefully acknowledges financial support from JSPS through the Overseas Challenge Program for Young Researchers (number 201980057).}}}

\date{\today}
\maketitle

\vspace{-1cm} 

%REStud version
\begin{abstract}
Consider a planner who has limited knowledge of the policy's causal impact on a certain local population of interest due to a lack of data, but does have access to the publicized intervention studies performed for similar policies on different populations. How should the planner make use of and aggregate this existing evidence to make her policy decision? Following Manski (2020; Towards Credible Patient-Centered Meta-Analysis, \textit{Epidemiology}), we formulate the planner's problem as a statistical decision problem with a social welfare objective, and solve for an optimal aggregation rule under the minimax-regret criterion. We investigate the analytical properties, computational feasibility, and welfare regret performance of this rule. We apply the minimax regret decision rule to decide whether to enact an active labor market policy based on 14 randomized control trial studies.
\end{abstract}

% Longer version
%\begin{abstract}
%Consider a planner who has to decide whether or not to introduce a new policy to a certain local population. The planner has only limited knowledge of the policy's causal impact on this population due to a lack of data but does have access to the publicized results of intervention studies performed for similar policies on different populations. How should the planner make use of and aggregate this existing evidence to make her policy decision? Building upon the paradigm of `patient-centered meta-analysis' proposed by Manski (2020; Towards Credible Patient-Centered Meta-Analysis, \textit{Epidemiology}), we formulate the planner's problem as a statistical decision problem with a social welfare objective pertaining to the local population, and solve for an optimal aggregation rule under the minimax-regret criterion. We investigate the analytical properties, computational feasibility, and welfare regret performance of this rule. We also compare the minimax regret decision rule with plug-in decision rules based upon a hierarchical Bayes meta-regression or stylized mean-squared-error optimal prediction. We apply the minimax regret decision rule to two settings: whether to enact an active labor market policy given evidence from 14 randomized control trial studies; and whether to approve a drug (Remdesivir) for COVID-19 treatment using a meta-database of clinical trials.
%\end{abstract}


\pagebreak

%%%%%%%%%%%%%%%%%%%%%%%%%%%%%%%%%%%%%%%%%%%

\section{Introduction}\label{sec:intro}
An increasing number of policy-making authorities are  interested in making their policy decisions evidence-based. In evidence-based decision-making, it is crucial for a planner to acquire credible evidence of a policy's causal impact on the affected population. Obtaining credible evidence for better policy decision-making is, however, challenging in many contexts. For instance, although randomized control trials (RCTs) are considered to be ideal for obtaining evidence of the causal impact of a policy, conducting an RCT can be costly in terms of budget, time, or administrative resources. Moreover, ethical or legal constraints can prevent the use of an RCT in certain institutional environments. In contrast, observational data can be both more accessible and easier to collect, but the credibility of any resulting causal estimates is limited if the validity of these estimates relies on restrictive identifying assumptions. In  scenarios where the planner faces difficulties in collecting direct evidence, a practical alternative is to analyze the publicized results of intervention studies performed for similar policies on different populations. With this approach in mind, how should the planner make use of and aggregate existing evidence to reach her policy decision?   

Statistical methodologies to aggregate evidence from multiple studies have been considered in the literature of meta-analysis and research synthesis. See, for instance, \citet{HO85} and a recent handbook volume, \citet{CooperHandbook2019}.  Since the seminal works of \citet{Rubin81} and \citet{DL86}, a common approach to aggregation of evidence is the hierarchical Bayesian approach, in which the typical objective of analysis is to infer hyper-parameters indexing the \textit{population of studies}. This framework of meta-analysis is useful for ``summarizing what has been learned and quantifying how results differ across the studies beyond the sampling error'' (\citet{DL15}). Its use is, however, limited when it comes to the planner's policy choice because the output of meta-analysis mainly concerns the population of studies rather than the particular population that is of interest to the planner. This point is made in \citet{manski2020towards}:

\begin{quote}
\textit{``Clinicians need to assess risks and choose treatments for populations
of patients, not population of studies. To express this
distinction succinctly, I will say that clinicians should want meta-analysis
to be patient-centered rather than study-centered.''}
\end{quote}
  
We pursue this paradigm of `patient-centered meta-analysis' to develop a method to aggregate existing studies for the purpose of making an optimal treatment decision on the local population that is of interest to the planner (hereafter, the target population). Building on the framework of statistical treatment choice proposed by \citet{Manski2000, manski2004statistical}, we formulate the planner's problem as a statistical decision problem \`{a} la \citet{Wald50}. The basic formulation of the decision problem analyzed in this paper is as follows. Let $\tau_0$ be the average welfare effect of introducing a new policy to the target population. There is no data from which the planner can directly infer $\tau_0$, but she does have access to the results of existing intervention or observational studies that are indexed by $k=1,  2, \dots, K$, $K \geq 1$. Each study $k$ reports a point estimate $\hat{\tau}_k$  for the average welfare effect $\tau_k$ in the study population, and an associated estimate $\hat{\sigma}_k$ of the standard error. We allow the study population to be different from the target population, so that the average welfare effects can differ, i.e., $\tau_k \neq \tau_{k'}$ for $k \neq k'$, $0 \leq k, k' \leq K$. 
The risk-neutral planner's decision problem, which we solve in this paper, is  whether or not to adopt the new policy for the target population upon observing a meta-sample, $(\hat{\tau}_k, \hat{\sigma}_k)$, $k=1, \dots, K$. That is, the statistical treatment choice rule we consider in this paper is a function $\hat{\delta}$ that maps the meta-sample to the binary choice of whether to adopt the policy or not. \Copy{not_MSE}{Such planner's decision problem contrasts with the estimation problem of causal effect parameters with statistical risk functions such as the mean squared errors.}

Following \citet{manski2004statistical, Manski2007}, \citet{stoye2009minimax, stoye2012minimax}, and \citet{tetenov2012statistical}, we apply the minimax regret criterion of \citet{Savage51} to obtain a minimax-regret treatment choice rule for the planner. \Copy{why_minimax_regret}{There are a couple of reasons that motivate the use of minimax regret criterion in the current problem. First, the minimax regret criterion does not require an explicit specification for a prior for the causal parameters so that it can be implemented in the contexts where the researchers cannot form a credible prior distribution for them. Second, in constrast to the maximin criterion, the minimax regret possesses balanced consideration of the states and avoid overly pessimistic worst-case scenarios \citep{Savage51, manski2004statistical}.} We assume that the planner's objective function (social welfare function) is linear in $\tau_0$ and consider the class of \textit{non-randomized} statistical treatment choice rules that select the treatment based on the sign of linear aggregation of $(\hat{\tau}_k : k = 1, \dots, K)$:
\begin{equation*}
\left\{ \hat{\delta}_{\bm{w}} = 1\left\{ \sum_{k=1}^K w_k \hat{\tau}_k \geq 0 \right\}:\sum_{k=1}^K w_k=1\right\},
\end{equation*}
where $\bm{w}=(w_1, \dots, w_K)'$ is a vector of weights assigned to each estimate in the pool of studies, which does not depend on the data. 
\Copy{advantages_non-randomized_rules}{Restricting the feasible rules to non-randomized (non-fractional) linear ones can be attractive in the following aspects. First, in the real-world practice of treatment choice or drug approval decisions, non-randomized allocation of treatments are easier than randomized ones for policy authorities to administer. Second, non-fractional allocations are guaranteed to attain ex post parity within the target population since either everyone or no one in the target population is treated, while fractional rules do not attain the ex-post parity. Third, under a standard specification of the parameter space of $\tau_k$'s, our linear aggregation minimax regret rule coincides with the unconstrained globally optimal minimax regret rule shown in \cite{stoye2012minimax, yata2021optimal} and \cite{montiel2023decision} if the width of the identified set for $\tau_0$ is small relative to the standard errors of nearest $\hat{\tau}_k$. Otherwise, the unconstrained optimal rule becomes randomized and differs from our linear aggregation rule, while we show that the worst-case regret performance of ours is not substantially worse than the global optimal one.}
%Otherwise, restricting to the class of linear aggregation rules sacrifices the value of minimax regret since unconstrained minimax regret rules are known to be randomized depending on the size of identified set for $\tau_0$ and the variances of $\hat{\tau}_k$'s as shown by \cite{stoye2012minimax, yata2021optimal}, and \cite{montiel2023decision}.  

Assuming a Gaussian sampling distribution for $(\hat{\tau}_1, \dots, \hat{\tau}_K)$ with known variances and imposing certain symmetry and invariance conditions on the parameter space for $(\tau_0, \tau_1, \dots, \tau_K)$, we derive the aggregation weights   $\bm{w}_{\text{minimax}}$ leading to a minimax-regret treatment choice rule among the non-randomized rules. 
Analytical characterization and computation of the exact minimax regret rule often become challenging in the context of statistical treatment choice. 
Our approach to the planner's minimax regret aggregation rule, in contrast, overcomes these challenges by showing that some mild restrictions on the parameter space and the class of decision rules deliver analytically and computationally tractable minimax regret rules.

We assert that the perspective and tools of statistical decision theory are particularly appealing in the meta-analysis setting for the following reasons. First, if each study in the pool reports a consistent estimate using a sample of moderate to large size (e.g., the difference-in-means estimator for the average treatment effect) then, by its asymptotic normality, it is plausible to assume that $\hat{\tau}_k$ follows a Gaussian distribution centered at $\tau_k$. Hence, the standard and well-studied framework of Gaussian experiments fits well to the current meta-analysis setting. Second, it is common for the meta-sample to consist of only a small number of studies. In such instances, asymptotic analysis with $K \to \infty$ can be misleading, and deriving finite-$K$ optimal procedures, which statistical decision theory is particularly suitable for, is desirable.

As an alternative to the minimax regret treatment choice rule, one could consider using a \textit{plug-in} rule that chooses the treatment according to the sign of an estimate of $\tau_0$. The plug-in rule that uses a minimax mean squared error (MSE) optimal estimate of $\tau_0$ is an example. Minimax-MSE estimation for finite-dimensional Gaussian mean models is well-studied and the minimax-MSE weights $\bm{w}_{\text{MSE}}$ are simple to compute, although the resulting plug-in decision rule  does not generally possess decision theoretic optimality in terms of the planner's objective function. To quantify the welfare cost of $\hat{\delta}_{\bm{w}_{\text{MSE}}}$, we compare the worst-case regrets of $\hat{\delta}_{\bm{w}_{\text{minimax}}}$ and $\hat{\delta}_{\bm{w}_{\text{MSE}}}$, and show that the worst-case regret of $\hat{\delta}_{\bm{w}_{\text{MSE}}}$ is worse than the minimax regret up to a constant factor, independent of the number of studies $K$ and the parameter space.

Our framework can accommodate a vector of observable characteristics $x_k$, $k=0,1,\dots,K$, where $x_k$ includes the characteristics of the treatment and demographics of the population featured in study $k$. Under a linear functional form specification, $\tau_k = \beta_0 + x_k'\beta$, $\beta \in \mathcal{B}$, common to standard meta-regression analysis  (see, e.g., \citet{SJ89}), we discuss those restrictions on $\mathcal{B}$ under which we can apply our minimax regret decision rule. For minimax regret to be bounded, an important constraint is boundedness of $\mathcal{B}$, and the bounds of $\mathcal{B}$ have to be explicitly specified to obtain the minimax regret rule. In reality, the planner may not be able to come up with reasonable bounds for $\mathcal{B}$. To offer a practical solution to this difficulty, we consider a data-driven way to specify the parameter space based on confidence sets for $(\tau_1, \dots, \tau_K)$.       

We illustrate the use of our minimax regret treatment rule by analyzing whether an active labor market program should be adopted using the meta-database appearing in \citet{card2017works}. We consider a pool of 14 RCT studies of job training programs, covering 8 different countries (Argentina, Brazil, Colombia, Dominica, Jordan, Nicaragua, Turkey, and the United States). Based on the average treatment effect and standard error estimates in each of these studies, and the demographic characteristics of the studied populations, we calculate the minimax regret adoption decisions for several countries (Japan, the United Kingdom, and Peru) for which the corresponding experimental estimates are not available in the meta-database.\footnote{In \ref{sec:covid-19}, we apply our method to inform the drug approval decision for a COVID-19 medication called Remdesivir using the meta-database of randomized clinical trials for COVID-19 treatments provided by \citet{juul2020interventions}. We compute the minimax regret treatment choice for Remdesivir for different demographic groups.}

The remainder of the paper is organized as follows. The next subsection reviews the related literature. Section 2 formulates the minimax regret decision problem and shows the main analytical result of the paper. In Section 3, we compare the minimax regret with the maximum regret of the decision rule based on the minimax-MSE aggregation rule. Section 3 also discusses a data-driven construction of the parameter space. Section 4 performs numerical analysis to compare the minimax-regret aggregation rules with the minimax-MSE and meta-OLS rules. Section \ref{sec:empirical} illustrates the use of our methods through an application to an active labor market policy.

\subsection{Related Literature}\label{subsec:literature}
This paper contributes to the growing literature on statistical treatment choice and individualized treatment assignment rules initiated by \citet{Manski2000, manski2004statistical} and \citet{Dehejia2005}. Contributions to the current literature include, \citet{HiranoPorter2009}, \citet{stoye2009minimax, stoye2012minimax}, \citet{Chamberlain2011}, \citet{BhattacharyaDupas2012}, \citet{tetenov2012statistical}, \citet{Kasy2014b,Kasy2018}, \citet{kitagawa2018should, KT19}, \citet{KW20}, \citet{Russell20}, \citet{KST21}, \citet{MT17}, \citet{AW20}, \citet{Sakaguchi21}, and \citet{Viviano21}, among others.
The problem of individualized treatment assignment rules has also been an area of active research in the fields of medical statistics and machine learning; see, for instance, \citet{Zadrozny03}, \citet{BeygelzimerLangford09} \citet{QianMurphy2011}, \citet{Zhao2012JASA}, \citet{SJ15}, \citet{Kallus_2020}, to list but a few papers. The standard setting in the existing literature considers an optimal treatment assignment policy for the population from which a sample was drawn, rather than combining the pool of estimates from multiple studies performed on different populations.  

There is a growing literature on how to inform policy using multiple pieces of evidence or extrapolation from one or multiple reference populations. \citet{Dehejiaetal21} considers the use of (quasi-)experimental evidence to study the decision of whether to experiment or to extrapolate, and, if applicable, where to conduct a new experiment. \citet{Manski18} analyzes decision-making for personalized risk assessment under the ecological inference setting where (partial) identification of a long regression is obtained by combining information on a short regression and the joint distribution among the regressors. The meta-analysis setting considered in this paper differs from the ecological inference setting in terms of the object to identify and the type of information provided by the available studies. Focusing on conditional cash transfer programs, \citet{Gechteretal19} runs multiple program evaluation methods on data obtained from Mexico to inform treatment assignment policies for Morocco, and empirically compare the welfare performances of these policies. \citet{Hotzetal05} and \citet{Dehejiaetal21} analyze how to predict the effects of future programs from past experimental evaluations by adjusting for differences in the distributions of observable characteristics. \citet{AO19} propose a method to conduct sensitivity analysis and to approximate external validity bias when the trial and target populations differ in the distribution of unobservables. \citet{Gechter16} considers bounding causal effects in a target population by restricting the dependence between the treated and control outcomes.

Meta-analysis for research synthesis has been actively studied in statistics and the resulting literature is vast; see, e.g., \citet{Borenstein09textbook} for a textbook and \citet{CooperHandbook2019} for a handbook volume. 
In economics, existing applications of meta-analysis and meta-regression include \citet{CK95}, \citet{Dehejia03}, \citet{Bandieraetal17}, \citet{card2017works}, \citet{Meager19, Meager20}, \citet{Imai20}, and \citet{Vivalt20}. See \citet{Stanley01} for a review. 
The common framework of meta-analysis introduces the population of studies and draws inference for the parameters thereof. As we argue in the Introduction via the quote from \citet{manski2020towards}, the usefulness of the conventional framework of meta-analsyis is not obvious for informing the planner's policy decision. 
This paper follows and pushes forward the perspective of patient-centered meta-analysis. 
The methodological proposals in \citet{manski2020towards} concern predicting treatment effects for the target population by intersecting the population identified sets for $\tau_0$ formed by extrapolation from each study, rather than explicitly taking into account sampling uncertainty due to finite sample size when considering the treatment choice decision.   
 
In terms of the framework and analytical and computational challenges for obtaining minimax regret rules, this paper is most closely related to \citet{stoye2012minimax}. 
In one of his baseline settings, \citet{stoye2012minimax} considers Gaussian experiments for conditional average treatment effects with a scalar covariate $x \in \mathcal{X}$, and analyzes the properties of the minimax regret treatment rule under a restriction that the conditional average treatment effects depend on $x$ with bounded variation. 
Similar to \citet{stoye2012minimax}, our framework allows (study-specific) covariates to constrain the parameter space for $(\tau_0, \tau_1, \dots, \tau_K)$, but there are two aspects in which our framework differs.
First, the treatment assignment rules considered in \citet{stoye2012minimax} are functions $\delta: \mathcal{X} \to \{0 ,1\}$, while we concern ourselves with the treatment choice at a particular covariate value $x_0$ in $\mathcal{X}$ that corresponds to the covariate value of the target population. 
This reduction of the treatment choice rule from a function of $x$ to a point significantly simplifies the analysis and computation of the minimax-regret rule. 
Second, the conditions we impose on the parameter space for feasible computation of the minimax regret rule is general and includes the bounded variation restriction considered in \citet{stoye2012minimax} as a special case.

In another baseline setting of \citet{stoye2012minimax}, he derives a minimax regret treatment choice rule under partially identified welfare, where the identified set has the known width but unknown location, and a Gaussian signal is available for it. Our setting is more complex than his setting due to multiple Gaussian signals with both the location and width of the identified set unknown. Contemporaneously or after the initial version of the current paper was circulated, there have been some important advances in the literature. In a similar setting to this paper, \citet{yata2021optimal} obtained an analytical representation for an exact unconstrained minimax regret rule, which implements a fractional assignment when the identified set for ATE is wide relative to the variances of the bound estimates. \citet{montiel2023decision} shows that minimax regret rules are not unique and considers refining the set of minimax regret rules. An alternative approach to handle uncertainty in the bound estimates and ambiguity within the identified set is to introduce multiple priors as in \citet{giacomini2021} and performs a gamma minimax decision rule. See \citet{giacomini_etal2021} for gamma minimax decision rules for set-identified models. \citet{Christensen2022} studies this line of approach for treatment choice and shows its optimality properties in limit experiments. 

Viewing $\tau_0$ as the value of the regression equation at $x_0$ in a Gaussian regression model and considering standard estimation risk such as the mean squared errors for $\tau_0$, the problem is reduced to an interpolation or extrapolation exercise based on the Gaussian signals. As such, the minimax estimation and inference problem for $\tau_0$ is similar to the extrapolation issue in the regression discontinuity setting analyzed in \citet{KR18}. Recent contributions regarding estimation and inference for Gaussian sequence models are made by \citet{Johnstone17} and \citet{AK18, AK20}. These papers consider statistical losses for estimation and inference, but do not consider the welfare criterion for statistical treatment choice. 

%%%%%%%%%%%%%%%%%%%%%%%%%%%%%%%%%%%%%%%%%%%
\section{Minimax regret treatment rule}\label{sec:main}

%%%%%%%%%%%%%%%%%%%%%%%%%%%%%%%%%%%%%%%%%%%
\subsection{Setting}\label{subsec:setting}

Suppose we have access to the publicized results of $K$ studies indexed by $k=1, \dots, K$. Each of the studies estimates the causal effect of a particular binary policy or treatment. We allow the details of the policy (implementation protocol, dosage, program contents, etc.) to differ across the studies. For $k = 1, \cdots, K$, let $\hat{\tau}_k$ denote the estimate of the policy effect reported in study $k$ and $\sigma_k > 0$ denote the standard error of $\hat{\tau}_k$. For simplicity, we assume $\sigma_k$ is known, although, in practice, we can only construct a consistent estimator for $\sigma_k$. We solve for the finite sample minimax regret rule with known $(\sigma_k : k=1,\dots,K)$, recommending that in practice the rule is implemented with the true standard errors replaced by their consistent estimates.\footnote{Solving the decision problem with Gaussian signals with known variances and obtaining a feasible decision rule by plugging in consistent estimators for the variances are similar to the construction of an asymptotically optimal decision rule within the framework of Gaussian limit experiments. See \citet{HiranoPorter2020} for a recent review. \ref{subsec:unknown} considers the case of unknown $\sigma_k$ and studies how plugging in an estimator $\hat{\sigma}_k$ affects the welfare performance of the treatment choice rules.}

We assume
\begin{equation}
\hat{\tau}_k \ \sim \ \mathcal{N}(\tau_k,\sigma_k^2), \ \ \ k = 1, \dots, K, \label{tau_hat}
\end{equation}
where $\tau_k$ is the true policy effect of the population featured in study $k$. We allow $\tau_k$ to vary across the studies. The assumption that $\hat{\tau}_k$ follows a Gaussian sampling distribution is reasonable if the reported estimator $\hat{\tau}_k$ is consistent and asymptotically normal and each study has a moderate to large sample.

Throughout this paper, we consider a planner who, upon observing data $\mathbf{D} \equiv \{\hat{\tau}_k\}_{k=1}^K$, must determine whether or not to adopt the policy in the target population given that its policy effect $\tau_0$ is unknown. 

In our empirical application, we consider a treatment choice of whether an active labor market policy (job training, subsidized employment, and job search assistance programs) should be adopted or not for a target population based on $K=14$ RCT estimates of the average treatment effects $(\hat{\tau}_k: k=1, \dots, 14)$ obtained from 8 different countries (Argentina, Brazil, Colombia, Dominica, Jordan, Nicalagua, Turkey, and the United States). Each RCT estimate comes with a standard error estimate $\hat{\sigma}_k$. For the target population ($k=0$), we consider Japan, the United Kingdom, and Peru, for which an experimental estimate for the program's average effects are not available.

Following \cite{manski2004statistical, Manski2007}, \citet{stoye2009minimax, stoye2012minimax}, and \citet{tetenov2012statistical}, we focus on minimax regret criterion to solve this decision problem. To this end, we assume that the true parameters $\bm{\tau} \equiv (\tau_0, \tau_1, \dots , \tau_K)'$ are \textit{ex ante} known to belong to the parameter space $\mathcal{T}$. 
%As we show below, the parameter space can be dependent on observed covariates, $\{x_k\}_{k=0}^K$.

We impose the following restrictions on the parameter space $\mathcal{T}$:

%% Assumption 1 %%
\begin{Assumption} \label{parameter_space}
The parameter space $\mathcal{T}$ satisfies
\begin{enumerate}
\item Symmetry: $\bm{\tau} \in \mathcal{T} \ \Rightarrow \ -\bm{\tau} \in \mathcal{T},$ and
\item Invariance to common constant addition: $\bm{\tau} \in \mathcal{T} \ \Rightarrow \ \bm{\tau} + c \cdot (1, \dots , 1)' \in \mathcal{T}$ for any $c \in \mathbb{R}$.
\end{enumerate}
\end{Assumption}

The symmetry assumption rules out the imposition of a sign restriction on the causal effect parameters, i.e., $\tau_k \geq 0$ for some $k$. 
The condition of invariance to common constant addition (hereafter, shortened to invariance) implies that $\{(\tau_k - \tau_0) : \bm{\tau} \in \mathcal{T}, \tau_0 = t\}$ does not depend on $t$. 
We use this result to simplify derivation of a minimax regret treatment rule. 
It is worth noting that the invariance condition rules out the case in which the parameter space for some $\tau_k$, $k \in \{0,1,2, \dots, K\}$, is bounded. 
For instance, if the outcome is binary, the treatment effect on the outcome is bounded by $[-1,1]$ for all $\tau_k$, $k=0,1, \dots, K$.\footnote{In a similar setting to the one in this paper, \cite{ishihara2023bandwidth} studies the treatment choice problem with $\tau_k$, $k \in \{0,1,2, \dots, K\}$ being bounded due to binary outcomes.}
However, if the standard errors $(\sigma_k^2:k=1,2, \dots, K)$ and the variations among $(\tau_0, \tau_1, \dots, \tau_K)$  imposed on $\mathcal{T}$  (e.g., Lipschitz constants $C_{kl}$, $ 0 \leq k,l \leq K$, in Example \ref{ex:bounded variations} below)  are sufficiently small relative to the size of the supports, bounded support of $(\tau_0, \tau_1, \dots, \tau_K)$ is less of an issue because extreme values of $\bm{\tau}$ beyond its logical support are unlikely to correspond to a worst-case in terms of the regret. 

The following two examples satisfy the parameter space constraints of Assumption \ref{parameter_space}.

%% Example 1 %%
\begin{Example}[The space of $\bm{\tau}$ spanned by the meta-regressions] \label{Ex:meta_regression}

Suppose that for each study in the pool, $k=1, \dots, K$, we can construct a vector of study-specific observable characteristics $x_k \in \mathbb{R}^{d_x}$. For example, $x_k$ can contain the average characteristics of the individuals in the sample studied in the $k$-th study. It can also include the socioeconomic or demographic characteristics of the country that the sample was drawn from and the characteristics of the treatment studied. The target population has known covariate value $x_0 \in \mathbb{R}^{d_x}$, which shares the same interpretation and dimension as $x_k$, $k=1, \dots, K$. 

In meta-regression analysis, $\tau_k$ is often specified as
$$
\tau_k = \beta_0 + x_k' \beta.
$$
Accordingly, we assume that the parameter space can be written as
\begin{equation}
\mathcal{T}_{\mathrm{meta}} \equiv \left\{ \bm{\tau} = (\tau_0, \tau_1, \dots, \tau_K) : \ \text{$\tau_k = \beta_0 + x_k' \beta$, $\beta_0 \in \mathbb{R}$, and $\beta \in \mathcal{B}$ for $ 0 \leq k \leq K$.} \right\}, \label{T_meta} 
\end{equation}
where $\mathcal{B}$ is a compact subset of $\mathbb{R}^{d_x}$. As seen in Theorem \ref{thm:finite}, compactness of $\mathcal{B}$ implies that minimax regret is finite. If $\mathcal{B}$ satisfies $\beta \in \mathcal{B} \Rightarrow -\beta \in \mathcal{B}$, then $\mathcal{T}_{\mathrm{meta}}$ satisfies Assumption \ref{parameter_space}. We can allow study-specific intercepts without violating Assumption \ref{parameter_space} by viewing them as study-specific fixed dummy variables added to the covariate vector.
\end{Example}

%% Example 2 %%
\begin{Example}[The class of bounded variations] \label{ex:bounded variations}
Consider the following parameter space:
\begin{equation}
\mathcal{T}_C \equiv \left\{ \bm{\tau} : |\tau_k - \tau_l| \leq C_{kl}  \ \ \text{for $k, l = 0, 1, \cdots , K$.} \right\}, \label{T_C} 
\end{equation}
where $\{C_{kl} : k, l = 0, 1, \cdots , K \}$ is a set of known positive constants. Clearly, for any $\{C_{kl} : k, l = 0, 1, \cdots , K \}$, the parameter space $\mathcal{T}_C$ satisfies Assumption \ref{parameter_space}. Assumption 1 in \citet{stoye2012minimax} corresponds to the case where  $C_{kl}$
is a common constant for any $0 \leq k,l \leq K$.

\Copy{scale_invariance}{With study-specific covariates as introduced in Example \ref{Ex:meta_regression}, setting $C_{kl} = C \|x_k - x_l\|$, $C>0$, yields the class of Lipschitz vectors of $\bm{\tau}$. We can make $\mathcal{T}_C$ invariant in the scale of study-specific (cardinal) covariates by defining the metric with the standardized covariates.}

\Copy{study_heterogeneity}{In some cases, a single study may provide estimates of the average treatment effect for multiple subgroups. In such situations, one may wish to account for study-specific effects. By defining $C_{kl}$ as follows, we can consider a parameter space that incorporates study-specific effects:
\[
C_{kl} \ = \  C \|x_k - x_l\| + \gamma \cdot 1\{ \text{$\hat{\tau}_k$ and $\hat{\tau}_l$ are from different studies} \},
\]
where $\gamma > 0$ can be interpreted as an upper bound on the magnitude of the difference between study-specific effects.}
\end{Example}

In our empirical application on adoption of active job market policies, we use indicators for the gender (male, female, or both), categories for age group targeted by each experiment, as well as socioeconomic characteristics of the country (OECD or not, GDP growth rate, unemployment rate) in which the experiment was conducted. We set the parameter space as $\mathcal{T}_C$ defined in (\ref{T_C}) with $C_{kl} = C\|x_k - x_l \|$, and tune the Lipschitz constant $C>0$ by cross-validation or empirical Bayes (marginal likelihood maximization).

When Assumption \ref{parameter_space} does not hold in a given application, we can still formulate the optimization problem to derive the minimax regret treatment rule, although solving for this is accompanied by a substantial increase in computational complexity. In Remark \ref{rem:without_A1} below, we discuss a derivation of the minimax regret treatment rule that does not rely upon Assumption \ref{parameter_space}.

%%%%%%%%%%%%%%%%%%%%%%%%%%%%%%%%%%%%%%%%%%%
\subsection{Welfare and regret}\label{subsec:regret}

Given a non-randomized treatment choice action $\delta \in \{0,1\}$, define the welfare attained by $\delta$ as
\begin{equation}
W(\delta) \ \equiv \  (\tau_0-c_0) \cdot \delta + \mu_0, \label{welfare}
\end{equation}
where $c_0$ is the per-person cost of the policy and $\mu_0$ is the average outcome that would be realized in the absence of the policy. An optimal treatment choice action given knowledge of $\tau_0$ and $c_0$ is 
\begin{equation}
\delta^* \ \equiv \ \mathbf{1}\{ \tau_0 \geq c_0 \}. \label{optimal treatment rule}
\end{equation}

Let $\hat{\delta}(\mathbf{D}) \in \{0,1\}$ be a non-randomized statistical treatment rule that maps the meta-sample $\mathbf{D}$ to the binary decision of treatment choice in the target population. The welfare regret of $\hat{\delta}(\mathbf{D})$ is defined as
\begin{equation}
R(\bm{\tau} , \hat{\delta})  \equiv  E_{\bm{\tau}}\left[ W(\delta^*)-W(\hat{\delta}(\mathbf{D})) \right] =  (\tau_0-c_0) \left\{ \delta^* - E_{\bm{\tau}}[\hat{\delta}(\mathbf{D})] \right\}, \label{regret}
\end{equation}
where $E_{\bm{\tau}}(\cdot)$ is the expectation with respect to the sampling distribution of $\mathbf{D}$ given the parameters $\bm{\tau}$. 
\Copy{cost_replace}{Hereafter, we normalize the cost of treatment to $c_0=0$, i.e., we interpret $\tau_k$, $k=0, \dots, K$, as the average treatment effect net of the per-person treatment cost in the target population by replacing $\tau_k$ and $\hat{\tau}_k$ with $\tau_k-c_0$ and $\hat{\tau}_k - c_0$, respectively.}

The minimax regret criterion selects a statistical treatment rule that minimizes maximum regret:
\begin{equation}
\hat{\delta}_{\text{minimax}} \ \in \ \text{arg} \min_{\hat{\delta} \in \mathcal{D}} \max_{\bm{\tau} \in \mathcal{T}} R(\bm{\tau}, \hat{\delta}), \nonumber
\end{equation}
where $\mathcal{D}$ is a class of statistical treatment rules. We refer to $\hat{\delta}_{\text{minimax}}$ as a minimax regret rule. In the next subsection, we derive a minimax regret rule under the class of statistical treatment rules spanned by linear aggregation rules.

%%%%%%%%%%%%%%%%%%%%%%%%%%%%%%%%%%%%%%%%%%%
\subsection{The minimax regret rule among linear aggregation rules}\label{subsec:minimax}

To gain analytical and computational tractability, we focus on the class of linear aggregation rules.

%% Definition 1 %%
\begin{Definition}[Linear aggregation rules]\label{def:linear_rule}
The class of linear aggregation rules consists of non-randomized treatment choice rules, each of which chooses a treatment according to the sign of a linear aggregation of $(\hat{\tau}_1, \dots, \hat{\tau}_K)$:
\begin{equation}
\mathcal{D}_{\mathrm{lin}} \equiv \left\{ \hat{\delta}_{\bm{w}}(\mathbf{D}) = \mathbf{1}\left\{ \sum_{k=1}^K w_k \hat{\tau}_k \geq 0 \right\} : \sum_{k=1}^K w_k = 1 \right\}, \label{linear_aggregation_rule}
\end{equation}
where $\bm{w} = (w_1, \cdots, w_K)'$ does not depend on the data $\mathbf{D}$.
\end{Definition}

%This class rules out nonlinear treatment rules that plug in an aggregation of $(\hat{\tau}_1, \dots, \hat{\tau}_K)$ in the manner of James-Stein shrinkage or empirical Bayes. 
The class of linear aggregation rules contains many reasonable treatment rules. For example, plug-in rules $\mathbf{1} \left\{ \hat{\tau}_0(\mathbf{D}) \geq 0 \right\}$ based on linear estimators $\hat{\tau}_0(\mathbf{D})$ for $\tau_0$ belong to $\mathcal{D}_{\text{lin}}$. 
When study characteristics are included in the available covariates as in Example \ref{Ex:meta_regression}, $\mathcal{D}_{\text{lin}}$ includes those rules that plug in fitted values based on parametric linear regression or nonparametric kernel regression.
%Furthermore, as shown in Remark \ref{rem:Bayes} above, hierarchical Bayes decision rules under the linear meta-regression specification with Gaussian priors yields the linear aggregation rule.

%\begin{Remark} \label{rem:Bayes}
A popular approach in the meta-analysis literature is to perform hierarchical Bayesian inference that yields a posterior distribution for each parameter in $\bm{\tau}$ and its hyperparameters. Once uncertainty for $\bm{\tau}$ is summarized by the posterior distribution, we can show that the Bayes optimal decision rule is determined by the posterior mean of $\tau_0$. For any prior $\pi$, the Bayes optimal decision rule $\hat{\delta}_{\pi}$ is defined as
$$
\hat{\delta}_{\pi} \ \in \ \mathrm{arg} \min_{\hat{\delta}} \left\{ \int R(\bm{\tau},\hat{\delta})\,\mathrm{d}\pi(\bm{\tau}) \right\}.
$$
We observe that
\begin{eqnarray}
\int R(\bm{\tau},\hat{\delta}) \,\mathrm{d}\pi(\bm{\tau}) %&=& \int \tau_0 \left\{ \mathbf{1}\{\tau_0 \geq 0\} - E_{\tau}[\hat{\delta}(\mathbf{D})] \right\} \,\mathrm{d}\pi(\bm{\tau}) \nonumber \\
%&=& E_{\mathbf{D}}\left[ E_{\pi}\left( \tau_0 \left\{ \mathbf{1}\{\tau_0 \geq 0\} - \hat{\delta}(\mathbf{D}) \right\} \Big| \mathbf{D} \right) \right] \nonumber \\
&=& E_{\mathbf{D}}\left[ E_{\pi}\left( \tau_0 \mathbf{1}\{\tau_0 \geq 0\}| \mathbf{D} \right) - E_{\pi}\left( \tau_0 | \mathbf{D} \right) \hat{\delta}(\mathbf{D})  \right], \nonumber
\end{eqnarray}
where $E_{\mathbf{D}}(\cdot)$ denotes the expectation with respect to the marginal distribution of $\mathbf{D}$ and $E_{\pi}(\cdot | \mathbf{D})$ denotes the posterior mean. This implies that
\begin{equation}
    \hat{\delta}_{\pi}(\mathbf{D}) \ = \ \mathbf{1}\left\{ E_{\pi}\left( \tau_0 | \mathbf{D} \right) \geq 0 \right\}. \label{HB_rule}
\end{equation}

In Example 1, typical hierarchical Bayes approaches assume $\tilde{\beta} \equiv (\beta_0,\beta')' \sim \mathcal{N} \left( \bm{0},\bm{\Sigma}_{\tilde{\beta}} \right)$. Then, the posterior mean of $\tau_0 \equiv \beta_0 + x_0' \beta$ can be written as
\[
E[\tau_0 | \mathbf{D}] \ = \ \tilde{x}_0' \left(\bm{\Sigma}_{\tilde{\beta}}^{-1} + \bm{X}' \bm{\Sigma}^{-1} \bm{X} \right)^{-1} \bm{X}' \bm{\Sigma}^{-1} \hat{\bm{\tau}},
\]
where $\tilde{x}_k \equiv (1,x_k')'$, $\bm{X} \equiv (\tilde{x}_1, \dots, \tilde{x}_{K})'$, $\bm{\Sigma} \equiv diag \left( \sigma_1^2, \dots, \sigma_K^2 \right)$, and $\hat{\bm{\tau}} \equiv (\hat{\tau}_1, \dots, \hat{\tau}_K)'$. Hence, in this case, the Bayes optimal decision rule is a linear aggregation rule in the sense of Definition \ref{def:linear_rule}. The weights of the Bayes optimal rule depend on the hyperparameters of the Gaussian prior for $\tilde{\beta}$ and the matrix of study characteristics$, \bm{X}$. If the mean of the prior distribution is not zero, then the posterior mean of $\tau_0$ involves a constant. In such a case, the Bayes optimal decision rule belongs to the class of extended linear aggregation rules considered in Remark \ref{rem:construle} below. More generally, we consider the following hierarchical Bayesian model with the prior distribution,
\begin{equation}
\bm{\tau} \ = \ (\tau_0, \bm{\tau}_{-0})' \ \sim \ \mathcal{N} \left( \bm{0}, \bm{\Sigma}_{\bm{\tau}} \right), \label{prior_dist}
\end{equation}
where $\bm{\tau}_{-0} \equiv (\tau_1, \dots, \tau_K)'$ and
\[
\bm{\Sigma}_{\bm{\tau}} \ \equiv \ \left( \begin{array}{cc}
    \bm{\Sigma}_{\bm{\tau},11} & \bm{\Sigma}_{\bm{\tau},12}  \\
    \bm{\Sigma}_{\bm{\tau},21} & \bm{\Sigma}_{\bm{\tau},22}
\end{array} \right).
\]
Then, the posterior mean of $\tau_0$ is written as
\[
E[\tau_0 | \hat{\bm{\tau}}] \ = \ \bm{\Sigma}_{\bm{\tau}, 12} \bm{\Sigma}_{\bm{\tau}, 22}^{-1} \left(\bm{\Sigma}_{\bm{\tau},22}^{-1} + \bm{\Sigma}^{-1} \right)^{-1} \bm{\Sigma}^{-1} \hat{\bm{\tau}}.
\]
Therefore, when the prior distribution of $\bm{\tau}$ is normal with mean zero, the Bayes optimal decision rule is a linear aggregation rule.

In a typical empirical Bayes approach, one estimates the hyperparameters of the prior distribution and subsequently computes the posterior mean based on the plug-in prior (see, e.g., \citet{Robbins1951, Robbins1956}, \cite{efron1973stein}, and \cite{carlin1997bayes}). Accordingly, if the prior covariance matrix $\bm{\Sigma}_{\bm{\tau}}$ in (\ref{prior_dist}) is parameterized by hyperparameters and estimated by $\hat{\bm{\Sigma}}_{\bm{\tau}}$, the empirical Bayes decision rule $\hat{\delta}_{\mathrm{EB}}$ can be written as
\[
\hat{\delta}_{\mathrm{EB}}(\mathbf{D}) \ \equiv \ \bm{1} \left\{ \hat{\bm{\Sigma}}_{\bm{\tau}, 12} \hat{\bm{\Sigma}}_{\bm{\tau}, 22}^{-1} \left(\hat{\bm{\Sigma}}_{\bm{\tau},22}^{-1} + \bm{\Sigma}^{-1} \right)^{-1} \bm{\Sigma}^{-1} \hat{\bm{\tau}} \geq 0\right\}.
\]
%where
%\[
%\hat{\bm{\Sigma}}_{\bm{\tau}} \ \equiv \ \left( \begin{array}{cc}
%    \hat{\bm{\Sigma}}_{\bm{\tau},11} & \hat{\bm{\Sigma}}_{\bm{\tau},12}  \\
%    \hat{\bm{\Sigma}}_{\bm{\tau},21} & \hat{\bm{\Sigma}}_{\bm{\tau},22}
%\end{array} \right).
%\]
In contrast to the Bayes optimal decision rule under the hierarchical Bayesian model, the empirical Bayes decision rule is not a linear aggregation rule, since it involves the estimation of hyperparameters.
%\end{Remark}

We consider the minimax regret rule among $\mathcal{D}_{\mathrm{lin}}$ whose corresponding weight vector solves
\[
\bm{w}_{\text{minimax}}\ \in \ \text{arg} \min_{\bm{w}} \max_{\bm{\tau} \in \mathcal{T}} R(\bm{\tau}, \hat{\delta}_{\bm{w}}).
\]
To develop a computation method for $\bm{w}_{\text{minimax}}$, note from (\ref{tau_hat}) that
$$
\sum_{k=1}^K w_k \hat{\tau}_k \ \sim \ \mathcal{N} \left( \sum_{k=1}^K w_k \tau_k, \sum_{k=1}^K w_k^2 \sigma_k^2 \right).
$$
Hence, from (\ref{regret}), the regret of $\hat{\delta}_{\bm{w}}$ can be written as
\begin{eqnarray}
R(\bm{\tau}, \hat{\delta}_{\bm{w}}) &=& \tau_0 \cdot \left[ \mathbf{1}\{\tau_0 \geq 0\} - \Phi\left( \frac{\sum_{k=1}^K w_k \tau_k}{\sqrt{\sum_{k=1}^K w_k^2 \sigma_k^2}}\right) \right] \nonumber\\ 
&=& | \tau_0 | \cdot \Phi \left( - sgn(\tau_0) \cdot \frac{\sum_{k=1}^K w_k \tau_k}{\sqrt{\sum_{k=1}^K w_k^2 \sigma_k^2}}\right), \label{linear regret} 
\end{eqnarray}
where the first equality follows from the normality of $\sum_{k=1}^K w_k \hat{\tau}_k$ and the second equality follows from $1-\Phi(a) = \Phi(-a)$. %We then obtain $\bm{w}_{\text{minimax}}$ by minimizing the maximum regret $\max_{\bm{\tau} \in \mathcal{T}} R(\bm{\tau}, \hat{\delta}_{\bm{w}})$. 

Noting the symmetry of $\mathcal{T}$ from Assumption 1, we obtain
\begin{eqnarray}
\max_{\bm{\tau} \in \mathcal{T}} R(\bm{\tau},\hat{\delta}_{\bm{w}}) %&=& \max_{\bm{\tau} \in \mathcal{T}} \left\{ | \tau_0 | \cdot \Phi \left( - sgn(\tau_0) \cdot \frac{\sum_{k=1}^K w_k \tau_k}{\sqrt{\sum_{k=1}^K w_k^2 \sigma_k^2}}\right) \right\} \nonumber \\
&=& \max_{t \geq 0} \max_{\bm{\tau} \in \mathcal{T}, \tau_0 = t} \left\{ t \cdot \Phi \left( -  \frac{\sum_{k=1}^K w_k \tau_k}{\sqrt{\sum_{k=1}^K w_k^2 \sigma_k^2}}\right) \right\} \nonumber \\
&=& \max_{t \geq 0} \max_{\bm{\tau} \in \mathcal{T}, \tau_0 = t} \left\{ t \cdot \Phi \left( -  \frac{t+\sum_{k=1}^K w_k (\tau_k-t)}{\sqrt{\sum_{k=1}^K w_k^2 \sigma_k^2}}\right) \right\}, \nonumber 
\end{eqnarray}
where the first equality follows by symmetry of $\mathcal{T}$ and the second equality follows from $\sum_{k=1}^K w_k = 1$. Letting $s(\bm{w}) \equiv \sqrt{\sum_{k=1}^K w_k^2 \sigma_k^2}$ denote the standard deviation of $\sum_{k=1}^K w_k \hat{\tau}_k$, we have
\begin{eqnarray}
\max_{\bm{\tau} \in \mathcal{T}} R(\bm{\tau},\hat{\delta}_{\bm{w}}) %&=& \max_{t \geq 0} \max_{\bm{\tau} \in \mathcal{T}, \tau_0 = t} \left\{ t \cdot \Phi \left( - s^{-1}(\bm{w})\left( t+\sum_{k=1}^K w_k (\tau_k-t) \right) \right) \right\} \nonumber \\
&=& \max_{t \geq 0} \left\{ t \cdot \Phi \left( - s^{-1}(\bm{w})\left( t+\min_{\bm{\tau} \in \mathcal{T}, \tau_0 = t} \left\{ \sum_{k=1}^K w_k (\tau_k-\tau_0) \right\} \right) \right) \right\}. \nonumber
\end{eqnarray}

By Assumption 1, the term $\min_{\bm{\tau} \in \mathcal{T}, \tau_0 = t} \left\{ \sum_{k=1}^K w_k (\tau_k-\tau_0) \right\}$ does not depend on $t$. Hence, we obtain
\begin{equation}
b(\bm{w}) \ \equiv \ \max_{\bm{\tau} \in \mathcal{T}, \tau_0 = 0} \left\{ \sum_{k=1}^K w_k \tau_k \right\} \ = \ - \min_{\bm{\tau} \in \mathcal{T}, \tau_0 = t} \left\{ \sum_{k=1}^K w_k (\tau_k-\tau_0) \right\}. \nonumber
\end{equation}
Viewing $\sum_{k=1}^K w_k \hat{\tau}_k$ as an estimator for $\tau_0$, $b(\bm{w})$ and $s(\bm{w})$ can be interpreted as the maximum bias and the standard deviation of a linear estimator $\sum_{k=1}^K w_k \hat{\tau}_k$, respectively. Using these terms, we can express the maximum regret as
\begin{eqnarray}
\max_{\bm{\tau} \in \mathcal{T}} R(\bm{\tau},\hat{\delta}_{\bm{w}}) %&=& \max_{t \geq 0} \left\{ t \cdot \Phi \left( - \frac{t}{s(\bm{w})} + \frac{b(\bm{w})}{s(\bm{w})} \right) \right\} \nonumber \\
%&=& s(\bm{w}) \cdot \max_{t \geq 0} \left\{ t \cdot \Phi \left( - t + \frac{b(\bm{w})}{s(\bm{w})} \right) \right\} \nonumber \\
&=& s(\bm{w}) \cdot \eta \left( \frac{b(\bm{w})}{s(\bm{w})} \right), \nonumber
\end{eqnarray}
where $\eta(a) \equiv \max_{t \geq 0}\left\{ t \cdot \Phi(-t+a) \right\}$. Hence, we obtain the following theorem:

\bigskip

%% Theorem 1 %%
\begin{Theorem}\label{thm:minimax}
Suppose that the parameter space $\mathcal{T}$ satisfies Assumption 1. Then, the minimax regret rule among $\mathcal{D}_{\mathrm{lin}}$ is obtained via the following optimization:
\begin{equation}
\bm{w}_{\mathrm{minimax}} \ \in \ \mathrm{arg} \min_{\bm{w}} \left\{ s(\bm{w}) \cdot \eta \left( \frac{b(\bm{w})}{s(\bm{w})} \right) \right\}. \label{Theorem 1}
\end{equation}
\end{Theorem}

In view of Theorem \ref{thm:minimax}, we can compute $\bm{w}_{\text{minimax}}$ using the following algorithm: \vspace{0.05in}
\begin{itemize}
\item[\textbf{1.}] Fix $\bm{w}$ such that $\sum_{k=1}^K w_k = 1$. Calculate $s(\bm{w})$ and $b(\bm{w})$, and obtain the maximum regret $s(\bm{w}) \cdot \eta(b(\bm{w})/s(\bm{w}))$.%, where we approximate $\eta(\cdot)$ by a piecewise linear function.
\item[\textbf{2.}] Minimize the maximum regret $s(\bm{w}) \cdot \eta(b(\bm{w})/s(\bm{w}))$ subject to $\sum_{k=1}^K w_k = 1$. \vspace{0.05in}
\end{itemize}

\bigskip

In step 1, we compute $b(\bm{w})$ by solving the optimization in $\bm{\tau}$. As shown below, there are  many examples in which we can calculate $b(\bm{w})$ using linear programming. 
If so, $b(\bm{w})$ can be solved quickly and reliably even when $K$ is large. Furthermore, because $t \mapsto t \cdot \Phi(-t+a)$ is a smooth unimodal function, $\eta(a)$ is easy to compute.\footnote{To reduce computational cost, the function $\eta(\cdot)$ is approximated by a piecewise linear function in practice.} 

Figure \ref{fig:eta} displays the shape of $\eta(a)$ as a function of $a$. From the proof of Lemma \ref{lemma:2} below, we find that $\eta(a)$ is strictly increasing and convex. In numerical simulations given in Section 4 below, we compute the optimization of step 2 using the R package ``Rsolnp''. We find that this optimization step is quick and stable even when $K$ exceeds 100.

%%% Figure 1 %%%
\begin{figure}[h] 
\centering
\includegraphics[width=9cm]{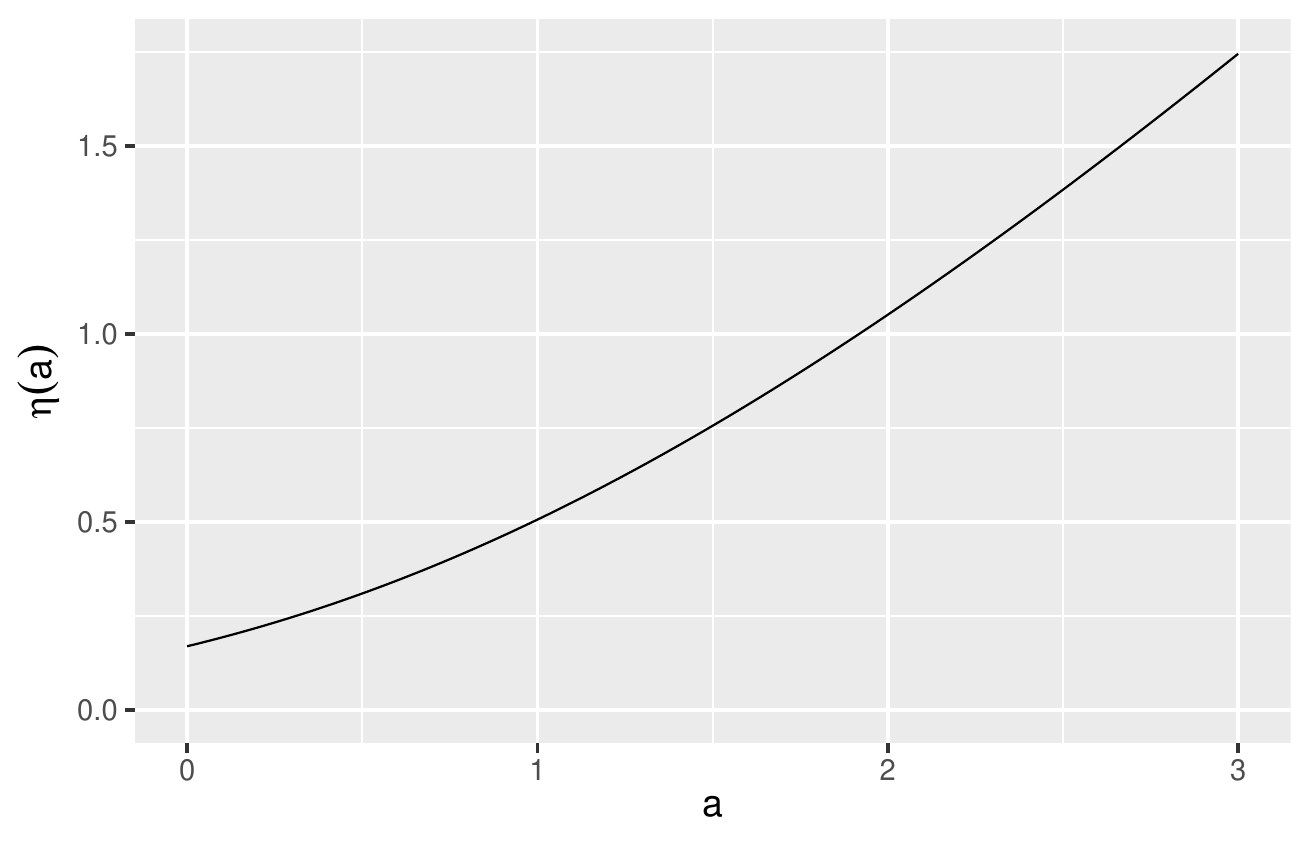}
\caption{Functional form of $\eta(a)$. The function $\eta(a)$ is strictly increasing and convex and $\eta(0)$ is approximately equal to 0.17.} \label{fig:eta}
\end{figure}

%\begin{Remark}\label{rem:LP}
There are some important cases where we can calculate $b(\bm{w})$ using linear programming. For example, consider the parameter space $\mathcal{T}_{\text{meta}}$ of Example 1, where we have
$$
b(\bm{w}) \ = \ \max_{\beta \in \mathcal{B}} \left\{ \sum_{k=1}^K w_k(x_k-x_0)'\beta \right\}. \nonumber 
$$
Hence, if $\mathcal{B}$ is a polyhedron, we can calculate $b(\bm{w})$ using linear programming. 

If the parameter space is $\mathcal{T}_C$ as in Example 2, we have
$$
b(\bm{w}) \ = \ \max_{|\tau_k - \tau_l| \leq C_{kl}} \left\{ \sum_{k=1}^K w_k (\tau_k - \tau_0) \right\},
$$
where this maximization can again be solved using linear programming. As an alternative approach to compute the minimax regret rule in the current setting, Proposition 1 in \cite{montiel2023decision} presents a procedure based on a solution to a fixed point problem.
%\end{Remark}

Even if the parameter space $\mathcal{T}$ satisfies Assumption 1, the minimax regret can be unbounded. For example, $\mathcal{T} = \mathbb{R}^{K+1}$ satisfies Assumption 1 but the maximum regret is unbounded with $b(\bm{w}) = +\infty$ for any $\bm{w}$. This is because $\lim_{a \rightarrow \infty} \eta(a) = + \infty$. 

To have the maximum regret bounded, we need to impose a restriction that  the difference between $\tau_k$ and $\tau_0$ is bounded for some $k$.

%% Assumption 2 %%
\begin{Assumption}
There exists $M < \infty$ such that $\bm{\tau} \in \mathcal{T}$ implies $|\tau_k - \tau_0| < M$ for some $k$.
\end{Assumption}

%% Theorem 2 %%
\begin{Theorem}\label{thm:finite}
Suppose that the parameter space $\mathcal{T}$ satisfies Assumptions 1 and 2. Then, the minimax regret is finite.
\end{Theorem}

Assumption 2 means that there exists some study in the pool that provides some (partially) identifying information about $\tau_0$. This condition holds for the parameter space $\mathcal{T}_{\mathrm{meta}}$ of Example 1 if $\mathcal{B}$ is compact. Similarly, $\mathcal{T}_C$ satisfies Assumption 2. Theorem \ref{thm:finite} then applies to these cases and guarantees that the minimax regret is bounded.

%% Remark 1 %%
\begin{Remark}\label{rem:construle}
We can add a constant term to equation (\ref{linear_aggregation_rule}), i.e., consider the following treatment rule:
\begin{eqnarray}
\hat{\delta}_{v,\bm{w}}(\mathbf{D}) \ = \ \mathbf{1} \left\{ v + \sum_{k=1}^K w_k \hat{\tau}_k \geq 0 \right\}, \label{linear constant rule}
\end{eqnarray}
where $v \in \mathbb{R}$. Then, the maximum regret of $\hat{\delta}_{v,\bm{w}}$ can be written as
$$
s(\bm{w}) \eta \left( s^{-1}(\bm{w})\{b(\bm{w})+|v|\} \right).
$$
Since $\eta(a)$ is strictly increasing in $a$, for $v \neq 0$, the maximum regret of $\hat{\delta}_{v,\bm{w}}$ is greater than that of $\hat{\delta}_{\bm{w}}$. Hence, the minimax regret rule sets $v = 0$, so that it is not necessary to consider treatment rules with an intercept as in (\ref{linear constant rule}).
\end{Remark}

%% Remark 2 %%
\begin{Remark} \label{rem:without_A1}
Let  $\mathcal{S} \equiv \{\tau_0: \bm{\tau} \in \mathcal{T} \}$. If Assumption 1 is relaxed, 
%the optimization problem for maximum regret can be expressed as
%\begin{eqnarray}
%\max_{\bm{\tau} \in \mathcal{T}} R(\bm{\tau},\hat{\delta}_{\bm{w}}) %&=& \max_{\bm{\tau} \in \mathcal{T}} \left\{ | \tau_0 | \cdot \Phi \left( - sgn(\tau_0) \cdot \frac{\sum_{k=1}^K w_k \tau_k}{s(\bm{w})}\right) \right\} \nonumber \\
%&=& \max_{\bm{\tau} \in \mathcal{T}} \left\{ | \tau_0 | \cdot \Phi \left( - sgn(\tau_0) \cdot \frac{\tau_0 + \sum_{k=1}^K w_k (\tau_k-\tau_0)}{s(\bm{w})}\right) \right\} \nonumber \\
%&=& \max_{\bm{\tau} \in \mathcal{T}} \left\{ | \tau_0 | \cdot \Phi \left( \cdot \frac{- |\tau_0| - sgn(\tau_0) \cdot \sum_{k=1}^K w_k (\tau_k-\tau_0)}{s(\bm{w})}\right) \right\} \nonumber \\
%&=& \max_{t \in \mathcal{S}} \left\{ | t | \cdot \Phi \left( \cdot \frac{- |t| - \min_{\tau \in \mathcal{T}, \tau_0 = t} \left\{ sgn(t) \cdot \sum_{k=1}^K w_k (\tau_k-t) \right\} }{s(\bm{w})}\right) \right\}, \nonumber
%\end{eqnarray}
%where $\mathcal{S} \equiv \{\tau_0: \bm{\tau} \in \mathcal{T} \} $. Hence, 
the maximum regret can be written as
\begin{equation}
\max_{\bm{\tau} \in \mathcal{T}} R(\bm{\tau},\hat{\delta}_{\bm{w}}) \ = \ \max_{t \in \mathcal{S}} \left\{ | t | \cdot \Phi \left( \cdot \frac{- |t| + \tilde{b}(t,\bm{w}) }{s(\bm{w})}\right) \right\}, \label{maximum_regret_without_Assumption_1}
\end{equation}
where
$$
\tilde{b}(t,\bm{w}) \equiv - \min_{\tau \in \mathcal{T}, \tau_0 = t} \left\{ sgn(t) \cdot \sum_{k=1}^K w_k (\tau_k-t) \right\}.
$$
The weights of the minimax regret rule therefore solve
$$
\bm{w}_{\mathrm{minimax}} \ \in \ \mathrm{arg}\min_{\bm{w}} \max_{t \in \mathcal{S}} \left\{ |t| \cdot \Phi \left( - \frac{|t| - \tilde{b}(t,\bm{w})}{s(\bm{w})} \right) \right\}.
$$
Since the parameter space $\mathcal{T}$ does not satisfy the second condition in Assumption \ref{parameter_space}, the maximum bias $\tilde{b}(t,\bm{w})$ may depend on $t$. \Copy{computational_difficulty_without_A1}{When Assumption \ref{parameter_space} holds, the maximum regret can be computed without solving the optimization problem once the functions $\eta(a)$ and $b(\bm{w})$ are obtained. On the other hand, when Assumption \ref{parameter_space} does not hold, computing the maximum regret requires maximization on $t$ in every time we evaluate candidate $\bm{w}$. This complicates computation for $\bm{w}_{\mathrm{minimax}}$.}

In some cases, it is natural to consider a bounded parameter space. For example, when the outcome variable is bounded, the policy effect $\tau_k$ must be also bounded. If the policy effects are known to be nonnegative, the cost-adjusted policy effects are bounded below by some constant. In \ref{subsec:bounded}, we show that computing the minimax regret rule becomes simpler under certain restrictions even when the parameter space is bounded.
\end{Remark}

%% Remark 3 %%
\begin{Remark} \label{rem:randomized}
Our results can be extended to the following randomized statistical treatment rules:
\begin{equation}
\hat{\delta}_{\sigma_Z, \bm{w}} (\mathbf{D}, Z) \ = \ \mathbf{1} \left\{ \sum_{k=1}^K w_k \hat{\tau}_k + \sigma_Z Z \geq 0 \right\}, \label{randomized_rule}
\end{equation}
where $\sigma_Z \geq 0$ and $Z \sim N(0,1)$ is independent of the data $\mathbf{D}$. Then, conditional on the data, this randomized statistical treatment rule (\ref{randomized_rule}) becomes 1 with probability $\Phi_{\sigma_Z}\left( \sum_{k=1}^K w_k \hat{\tau}_k \right)$, where $\Phi_{\sigma_Z}$ is the distribution function of $N(0,\sigma_Z^2)$. This rule becomes the linear aggregation rule when $\sigma_Z=0$. Because $\sum_{k=1}^K w_k \hat{\tau}_k + \sigma_Z Z$ is normally distributed with mean $\sum_{k=1}^K w_k \tau_k$ and variance $\sum_{k=1}^K w_k^2 \sigma_k^2 + \sigma_Z^2$, we obtain
\begin{equation*}
R(\bm{\tau}, \hat{\delta}_{\sigma_Z, \bm{w}}) \ = \ | \tau_0 | \cdot \Phi \left( - sgn(\tau_0) \cdot \frac{\sum_{k=1}^K w_k \tau_k}{\sqrt{\sum_{k=1}^K w_k^2 \sigma_k^2 + \sigma_Z^2}}\right).
\end{equation*}
From the same discussion as in Theorem \ref{thm:minimax}, we obtain the following corollary.

\begin{Corollary}\label{cor:randomized}
Suppose that the parameter space $\mathcal{T}$ satisfies Assumption 1. Then, we obtain
\begin{equation*}
\max_{\bm{\tau} \in \mathcal{T}} R(\bm{\tau}, \hat{\delta}_{\sigma_Z, \bm{w}}) \ = \ s(\sigma_Z,\bm{w}) \cdot \eta \left( \frac{b(\bm{w})}{s(\sigma_Z,\bm{w})} \right), 
\end{equation*}
where $s(\sigma_Z,\bm{w}) \equiv \sqrt{\sum_{k=1}^K w_k^2 \sigma_k^2 + \sigma_Z^2}$. 
\end{Corollary}

From Corollary \ref{cor:randomized}, optimization for finding a minimax regret rule among this class of randomized rules is similar to and a straightforward extension of the optimization for non-randomized rules shown in Theorem \ref{thm:minimax}. Note that the class of randomized rules we are optimizing over contains the minimax regret rule shown in  \citet{yata2021optimal}. \Copy{Rem5_discussion_Yata}{\citet{yata2021optimal} shows that the minimax regret rule takes the form $\hat{\delta}(\mathbf{D}) = \Phi_{\sigma_Z}\left( \sum_{k=1}^K w_k \hat{\tau}_k \right)$ with $\sigma_Z$ dependent on the variances $(\sigma_k^2, k=1, \dots, K)$ and a specification of the parameter space.}
\end{Remark}

%% Remark 4 %%
\begin{Remark}\label{rem:consistency}
If we know $(\tau_1, \dots, \tau_K)$, i.e., $\sigma_k = 0$ for all $1 \leq k \leq K$, we can obtain the (true) identified set of $\tau_0$ based on the constraints on the parameter space $\mathcal{T}$. For instance, we construct the identified set of $\tau_0$ by intersecting multiple bounds for $\tau_0$, each of which is constructed by extrapolating from $\tau_k$ as in \citet{manski2020towards}. We can then obtain a minimax regret treatment rule given the true identified set of $\tau_0$ without any sampling uncertainty. We denote by $\delta^*_{IS}$ such a (non-randomized) minimax regret rule. In \ref{subsec:consistency}, we show that the limiting version of $\hat{\delta}_{\bm{w}_{\text{minimax}}}$, with the sample size of each study diverges $\sigma_k \rightarrow 0$ for all $k=1, \dots, K$, does not agree with $\delta^*_{IS}$ in some settings.
\end{Remark}

We conclude this section by summarizing implementation of our minimax regret rule.
\begin{enumerate}
    \item Collect $k=1, \dots, K$ studies on the policy of interest and extract the estimated policy effects $\hat{\tau}_k$, their standard errors $\sigma_k$, and the corresponding study characteristics $x_k$.
    \item With $\tau_0$ be the policy effect of interest for the target population, specify the parameter space $\mathcal{T}$ for $\bm{\tau} = (\tau_0, \tau_1, \ldots, \tau_K)'$, e.g., select Lipschitz constants in Example \ref{ex:bounded variations}.  
    \item Solve the optimization problem in Theorem \ref{thm:minimax} to obtain $\bm{w}_{\mathrm{minimax}}$.
    \item Decide whether to adopt the policy or not based on whether $\sum_{k=1}^K w_{\mathrm{minimax},k} \hat{\tau}_k$ exceeds zero or not.
\end{enumerate}

%%%%%%%%%%%%%%%%%%%%%%%%%%%%%%%%%%%%%%%%%%%
\section{Extensions and discussion}\label{sec:extensions}

%%%%%%%%%%%%%%%%%%%%%%%%%%%%%%%%%%%%%%%%%%%
\subsection{Comparison with a minimax MSE rule}\label{subsec:MSE}

In this section, we compare $\bm{w}_{\text{minimax}}$ with other ways of forming the weights. First, we consider a minimax linear estimator of $\tau_0$ in terms of the mean squared errors (MSE). It is well known that the maximum MSE of $\sum_{k=1}^K w_k \hat{\tau}_k$ can be decomposed into the variance and the squared maximum bias:
\begin{equation}
b^2(\bm{w}) + s^2(\bm{w}). \nonumber
\end{equation}
Hence, the weights of the minimax MSE estimator are
\begin{equation}
\bm{w}_{\text{MSE}} \ \in \ \text{arg} \min_{\bm{w}} \left\{ b^2(\bm{w}) + s^2(\bm{w}) \right\}. \label{MSE weight}
\end{equation}
We refer to $\hat{\delta}_{\bm{w}_{\text{MSE}}}$ as the minimax MSE rule.\footnote{ \label{ftn:Donoho}The weights of minimax MSE in the current setting correspond to the minimax optimal weights of \cite{Donoho94} and \cite{AK20} for estimation of $\tau_0$. See also Appendix D of \cite{Ishihara_etal25}.}

To compare the minimax regret and MSE rules, we  focus on the analytical properties of $\eta(a)$. \cite{tetenov2012statistical} shows that $\eta(a)$ is a continuous, strictly increasing function and $\eta(0) \simeq 0.17$. Furthermore, in the proof of lemmas in \ref{sec:proof}, we show that $\eta(a)$ is convex. This allows us to derive upper and lower bounds for $\eta(a)$, which are then used in the following theorem to bound the maximum regret.

%Furthermore, in the proof of the following lemmas, we show that $\eta(a)$ is concave. We accordingly obtain the following upper and lower bounds on $\eta(a)$:

\if0
%% Lemma 1 %%
\begin{Lemma}
For any $v, a \geq 0$, we have
\begin{equation}
v\cdot\Phi(-v) + \Phi(-v)\cdot a \ \ \leq \ \ \eta(a) \ \ \leq \ \ \eta(0) + a. \label{Lemma 1}
\end{equation}
\end{Lemma}

%% Lemma 2 %%
\begin{Lemma} \label{lemma:2}
For any $a \geq 0$, we have
\begin{equation}
\eta(0) \sqrt{1+a^2} \ \ \leq \ \ \eta(a) \ \ \leq \ \ \sqrt{1+a^2}. \label{eq:Lemma2}
\end{equation}
\end{Lemma}
\fi

%Relying on Theorem \ref{thm:minimax} and Lemmas 1--2, the next theorem bounds the maximum regret.

\bigskip

%% Theorem 3 %%
\begin{Theorem}\label{thm:regret_bound}
Suppose that the parameter space $\mathcal{T}$ satisfies Assumptions 1 and 2. Then, for any $\bm{w}$ and $v \geq 0$, we obtain
$$
\Phi(-v)\cdot s(\bm{w}) + v \cdot \Phi(-v)\cdot b(\bm{w}) \ \leq \ \max_{\bm{\tau} \in \mathcal{T}} R(\bm{\tau},\hat{\delta}_{\bm{w}}) \ \leq \ \eta(0)\cdot s(\bm{w}) + b(\bm{w}).
$$
In addition, we obtain
\begin{equation}
\eta(0) \cdot \sqrt{b^2(\bm{w}) + s^2(\bm{w})} \ \leq \ \max_{\bm{\tau} \in \mathcal{T}} R(\bm{\tau},\hat{\delta}_{\bm{w}}) \ \leq \ \sqrt{b^2(\bm{w}) + s^2(\bm{w})}. \nonumber \label{Theorem 3}
\end{equation}
\end{Theorem}

\bigskip

Theorem \ref{thm:regret_bound} provides lower and upper bounds on the maximum regret. These bounds show that the maximum regret is bounded from above and from below by $b(\bm{w}) + s(\bm{w})$ and $\sqrt{b^2(\bm{w}) + s^2(\bm{w})}$ up to some proportional factors, independently of the number of studies, $K$, and the dimension of $x_k$, $d_x$. The second set of inequalities imply that the minimax regret is equivalent to the minimax RMSE (root-MSE) up to a constant factor. In other words, minimax RMSE enables us to bound the minimax regret. 

%Furthermore, Theorem \ref{thm:regret_bound} and Lemma \ref{lemma:2} lead to the following comparison of the maximum regret between the minimax regret rule  $\hat{\delta}_{\bm{w}_{\text{minimax}}}$ and the minimax MSE rule $\hat{\delta}_{\bm{w}_{\text{MSE}}}$.

Furthermore, Theorem \ref{thm:regret_bound} leads to the following comparison of the maximum regret between the minimax regret rule  $\hat{\delta}_{\bm{w}_{\text{minimax}}}$ and the minimax MSE rule $\hat{\delta}_{\bm{w}_{\text{MSE}}}$.

\bigskip

%% Theorem 4 %%
\begin{Theorem}\label{thm:MSE}
Suppose that the parameter space $\mathcal{T}$ satisfies Assumptions 1 and 2. Then, we have
\begin{equation}
1 \ \leq \ \cfrac{\max_{\bm{\tau} \in \mathcal{T}} R(\bm{\tau},\hat{\delta}_{\bm{w}_{\mathrm{MSE}}})}{\max_{\bm{\tau} \in \mathcal{T}} R(\bm{\tau},\hat{\delta}_{\bm{w}_{\mathrm{minimax}}})} \ \leq \frac{1}{\eta(0)} \ \simeq \ 5.88. \label{Theorem 4}
\end{equation}
\end{Theorem}

\bigskip

Theorem \ref{thm:MSE} shows that the maximum regret of the minimax MSE rule can be bounded by the minimax regret up to a constant factor, independently of $K$ and $d_x$. Numerical simulations in Section 4 suggest that the maximum regret of $\hat{\delta}_{\bm{w}_{\text{MSE}}}$ can be about 40 percent greater than the minimax regret.

One useful observation is that the minimax regret criterion places greater emphasis on the bias than on the variance compared with the minimax MSE criterion. To see this point, consider the directional derivatives of the maximum regret and MSE. We fix $\bm{\theta} = (\theta_1, \cdots , \theta_K)'$ with $\sum_{k=1}^K \theta_k = 1$ and assume that $b(\bm{w})$ is directionally differentiable in the direction of $\bm{\theta}-\bm{w}$.\footnote{\label{footnote:directional_differentiability}\Copy{directionally_differentiability_footnote}{For example, if $\mathcal{T} = \{ \bm{\tau} : |\tau_k - \tau_l| \leq C \|x_k-x_l\| \ \text{for any $k$, $l$} \}$ and all elements of $\bm{w}$ and $\bm{\theta}$ are nonnegative, then we obtain $$ b\left( (1-h) \cdot \bm{w} + h \cdot \bm{\theta} \right) \ = \ C \sum_{k=1}^K \| x_k - x_0 \| \cdot \{(1-h) w_k + h \theta_k\} \ = \ (1-h) \cdot b(\bm{w}) + h \cdot b(\bm{\theta}).$$ This implies that $b'_{\bm{\theta}}(\bm{w}) = b(\bm{\theta}) - b(\bm{w})$ in this case. Furthermore, Appendix D of \cite{Ishihara_etal25} shows that all elements of $\bm{w}_{\text{MSE}}$ must be nonnegative in this setting. Therefore, if all elements of $\bm{\theta}$ are nonnegative, $b(\bm{w})$ is directionally differentiable in the direction of $\bm{\theta}-\bm{w}$ when $\bm{w}=\bm{w}_{\text{MSE}}$.}} We define
\begin{eqnarray}
b_{\bm{\theta}}'(\bm{w}) &\equiv & \lim_{h \downarrow 0} \left( b\left((1-h)\cdot \bm{w}+h\cdot \bm{\theta} \right) - b(\bm{w}) \right)/h , \nonumber \\
s_{\bm{\theta}}'(\bm{w}) &\equiv & \lim_{h \downarrow 0} \left( s\left((1-h)\cdot \bm{w}+h\cdot \bm{\theta} \right) - s(\bm{w}) \right)/h , \nonumber \\
Q_{\bm{\theta}}(\bm{w}) & \equiv & \left( \max_{\bm{\tau} \in \mathcal{T}} R(\bm{\tau},\hat{\delta}_{(1-h) \cdot \bm{w}+h \cdot \bm{\theta}}) - \max_{\bm{\tau} \in \mathcal{T}} R(\bm{\tau},\hat{\delta}_{\bm{w}}) \right)/h, \nonumber
\end{eqnarray}
where $Q_{\bm{\theta}}(\bm{w})$ is the directional derivative of the maximum regret. Then, as shown in \ref{subsec:derivation_MSE}, we obtain 
\[ 
b_{\bm{\theta}}'(\bm{w}_{\mathrm{MSE}}) < 0 \ \  \text{and} \ \  s_{\bm{\theta}}'(\bm{w}_{\mathrm{MSE}}) > 0 \ \ \Rightarrow \ \ Q_{\bm{\theta}}\left( \bm{w}_{\mathrm{MSE}} \right) < 0. 
\]
That is, at the minimax MSE weights $\bm{w} = \bm{w}_{\mathrm{MSE}}$, locally perturbing the weight vector in the direction that reduces the bias and increases the variance improves the welfare regret. This implies that the minimax regret criterion places greater emphasis on the bias than on the variance compared with the minimax MSE criterion. In the numerical analysis of Section 4, we plot $\bm{w}_{\mathrm{minimax}}$ and $\bm{w}_{\mathrm{MSE}}$ to illustrate the difference in their bias-variance balancing properties.

%%%%%%%%%%%%%%%%%%%%%%%%%%%%%%%%%%%%%%%%%%%
\subsection{Comparison with a hierarchical Bayes rule (simple case: $K=2$)}\label{subsec:K=2}
To illustrate the minimax regret rule among $\mathcal{D}_{\mathrm{lin}}$ and compare it with Bayes rules, consider the simple case where $K=2$ and $\mathcal{T} = \{ \bm{\tau}: |\tau_k - \tau_l| \leq C \ \text{for $k,l=0,1,2$ and $k \neq l$} \}$. 
Because $w_1 + w_2 = 1$, we can express $(w_1,w_2) = (w,1-w)$ and the maximum regret of $\hat{\delta}_{\bm{w}}$ as $s(w) \cdot \eta \left( b(w)/s(w) \right)$, where $b(w) = \max_{\bm{\tau} \in \mathcal{T}, \tau_0 = 0} \left\{ w \tau_1 + (1-w) \tau_2 \right\}$ and $s(w) = \sqrt{w^2 \sigma_1^2 + (1-w)^2 \sigma_2^2}$. Hence, the minimax regret weight $w^{\ast}$ is obtained by minimizing $s(w) \cdot \eta \left( b(w)/s(w) \right)$.

To compare the minimax regret rule with Bayes rules, consider the following hierarchical Bayes model:
\begin{eqnarray*}
& & \tau_k = \tau_0 + \epsilon_k, \ \ k = 1,2, \\
& & \tau_0 \sim N(0, \sigma_{\tau}^2), \ \ \epsilon_k \sim N(0, \sigma_{\epsilon}^2),
\end{eqnarray*}
where $\tau_0$, $\epsilon_1$, and $\epsilon_2$ are mutually independent. As discussed in Section \ref{subsec:minimax}, the Bayes optimal rule takes the form $1\{w_{\text{HB}} \cdot \hat{\tau}_1 + (1-w_{\text{HB}}) \cdot \hat{\tau}_2 \geq 0\}$. We compare the minimax regret weight $w^{\ast}$ with hierarchical Bayes weight $w_{\text{HB}}$.

As shown in \ref{subsec:derivation_K=2}, the minimax weight $w^{\ast}$ satisfies the following conditions:
\begin{eqnarray*}
\text{if} \ \sqrt{2/\pi} \cdot C < \sqrt{\frac{\sigma_1^2 \sigma_2^2}{\sigma_1^2 + \sigma_2^2}} & \Rightarrow & w^{\ast} = \frac{\sigma_2^2}{\sigma_1^2 + \sigma_2^2}, \\
\text{if} \ \sqrt{\frac{\sigma_1^2 \sigma_2^2}{\sigma_1^2 + \sigma_2^2}} \leq \sqrt{2/\pi} \cdot C \leq \max\{\sigma_1, \sigma_2\} & \Rightarrow & s(w^{\ast}) = \sqrt{2/\pi} \cdot C, \\
\text{if} \ \sqrt{2/\pi} \cdot C >  \max\{\sigma_1, \sigma_2\} & \Rightarrow & w^{\ast} \leq 0 \ \text{or} \ w^{\ast} \geq 1,
\end{eqnarray*}
Furthermore, the hierarchical Bayes weight becomes 
$$
w_{\text{HB}} \ = \ \frac{ \sigma_2^2 + \sigma_{\epsilon}^2}{\sigma_1^2 + \sigma_2^2 + 2 \sigma_{\epsilon}^2}.
$$

If the dispersion of parameters $C$ is small compared to the standard deviations $\sigma_1$ and $\sigma_2$, the minimax weight attains the smallest variance, that is, $w^{\ast} = \frac{\sigma_2^2}{\sigma_1^2 + \sigma_2^2}$. On the other hand, if $C$ is large compared to $\sigma_1$ and $\sigma_2$, then the minimax regret criterion may favor rules with variance $s(w)$ larger than the minimum.\footnote{If we consider the randomized statistical treatment rules in (\ref{randomized_rule}), for $C \geq \sqrt{\pi /2} \cdot \sqrt{\frac{\sigma_1^2 \sigma_2^2}{\sigma_1^2 + \sigma_2^2}}$ the maximum regret is minimized when $w \in [0,1]$ and
\begin{equation}
s(\sigma_Z,w) \ \equiv \ \sqrt{w^2\sigma_1^2+(1-w)^2\sigma_2^2 + \sigma_Z^2} \ = \ \sqrt{2/\pi} \cdot C. \label{minimax_K=2_randomized}
\end{equation}
Hence, if $\sqrt{2/\pi} \cdot C >  \max\{\sigma_1, \sigma_2\}$, then $\sigma_Z$ must be positive to satisfy (\ref{minimax_K=2_randomized}). This implies that when the dispersion of parameters is large, the minimax regret criterion may favor randomized treatment rules over non-randomized treatment rules, and this observation is consistent with \cite{stoye2012minimax} and \cite{yata2021optimal}.} The minimax regret weight $w^{\ast}$ and the hierarchical Bayes weight $w_{\text{HB}}$ satisfy $|w^{\ast}-1/2| \geq |w_{\text{HB}}-1/2|$ when $C$ is small. This implies that $w_{\text{HB}}$ shrinks $w^{\ast}$ toward $1/2$ when $C$ is small, and the degree of shrinkage is increasing in $\sigma_{\epsilon}^2$. If $\sigma_1 \neq \sigma_2$, the weights $w^{\ast}$ and $w_{\text{HB}}$ agree only in an extreme scenario of $\sigma_{\epsilon}^2 = 0$.

%%%%%%%%%%%%%%%%%%%%%%%%%%%%%%%%%%%%%%%%%%%
\subsection{Comparison with a global minimax regret rule}\label{subsec:global}
In this section, we compare the proposed minimax regret linear aggregation rule with a global minimax regret rule that solves
\[
\min_{\hat{\delta}} \max_{\bm{\tau} \in \mathcal{T}} R(\bm{\tau}, \hat{\delta})
\]
where $\min_{\hat{\delta}}$ denotes minimization over any treatment rules including stochastic assignments. \cite{montiel2023decision} and \cite{yata2021optimal} show that the class of the randomized rule $\hat{\delta}_{\sigma_Z,\bm{w}}$ in (\ref{randomized_rule}) contains global minimax regret rules when the parameter space is a Lipschitz parameter space, $\mathcal{T} = \{ \bm{\tau}: |\tau_k - \tau_l | \leq C \|x_k - x_l\| \ \text{for any $k,l$} \}$. Let $\rho$ be the ratio of the minimax regret of the linear aggregation rules to that of the global minimum: 
\[
\rho \ \equiv \ \frac{\min_{\bm{w}} \max_{\bm{\tau} \in \mathcal{T}} R(\bm{\tau}, \hat{\delta}_{\bm{w}})}{\min_{\hat{\delta}} \max_{\bm{\tau} \in \mathcal{T}} R(\bm{\tau}, \hat{\delta})} %\ = \ \frac{\min_{\bm{w}} \max_{\bm{\tau} \in \mathcal{T}} R(\bm{\tau}, \hat{\delta}_{\bm{w}})}{\min_{\sigma_Z, \bm{w}} \max_{\bm{\tau} \in \mathcal{T}} R(\bm{\tau}, \hat{\delta}_{\sigma_Z, \bm{w}})}.
\]
The following theorem provides an upper bound of $\rho$ for the Lipschitz parameter space.

\bigskip

%% Theorem 5 %%
\begin{Theorem}\label{thm:global}
Suppose that $\|x_1-x_0\| < \min_{k=2, \ldots,K} \|x_k-x_0\| $ and $\mathcal{T} = \{ \bm{\tau}: |\tau_k - \tau_l | \leq C \|x_k - x_l\| \ \text{for any $k,l$} \}$. Let $b_1 \equiv C\|x_1-x_0\|$. Then,
\[
1 \ \leq \ \rho \ \leq \ \bar{\rho}(b_1/\sigma_1) \ \equiv \ \begin{cases}
1, & \text{if $b_1/\sigma_1 \leq \sqrt{\pi / 2}$} \\
\frac{2 \eta(b_1/\sigma_1)}{b_1 / \sigma_1}, & \text{if $b_1/\sigma_1 > \sqrt{\pi / 2}$}
\end{cases}.
\]
Furthermore, $\overline{\rho}(b_1/\sigma_1)$ is an increasing function in $b_1 / \sigma_1$ and is bounded above by $2$. 
\end{Theorem}

\bigskip

This theorem shows that under the Lipschitz parameter space, $\hat{\delta}_{\bm{w}_{\mathrm{minimax}}}$ is a global minimax regret rule when the maximal bias extrapolated from the closest study $b_1$ is small relative to the standard deviation of the estimate there by $b_1/\sigma_1 \leq \sqrt{\pi/2} \simeq 1.25$. This result is consistent with the global minimax regret rules shown by \cite{montiel2023decision} and \cite{yata2021optimal}, and clarifies that restricting the rules to the linear non-randomized aggregation rules does not miss a globally optimal rule if $b_1/\sigma_1 \leq \sqrt{\pi/2}$. 
In contrast, if the bias - noise ratio $b_1/\sigma_1$ is larger than $\sqrt{\pi/2}$, $\hat{\delta}_{\bm{w}_{\mathrm{minimax}}}$ is no longer a global minimax regret, while Theorem \ref{thm:global} shows that the ratio of the maximum regrets can be bounded above by $\frac{2 \eta(b_1/\sigma_1)}{b_1/\sigma_1}$. For example, if $b_1/\sigma_1 =2$, then $\rho$ is bounded above by about $1.05$, i.e., restricting to linear aggregation rules comes with the increase of the maximum regret by only $5\%$. When $b_1/\sigma_1 =10$, $\rho$'s upper bound increases to $1.59$, while Theorem \ref{thm:global} shows that $\rho$ never doubles for any $C$ and $\sigma_1, \ldots, \sigma_K$. %From the proof of Theorem \ref{thm:global}, if $\bm{w}_{\mathrm{minimax}} = (1,0, \ldots, 0)'$, then $\rho$ coincides with the upper bound $\overline{\rho}(b_1/\sigma_1)$ when $b_1/\sigma_1 > \sqrt{\pi/2}$.

To illustrate our result, we investigate $\rho$ numerically when the Lipschitz parameter space has $K=2$ and $(x_0,x_1,x_2) = (0,1,1.5)$. From the proof of Theorem \ref{thm:global}, we have $\min_{\hat{\delta}} \max_{\bm{\tau} \in \mathcal{T}} R(\bm{\tau}, \hat{\delta}) = b_1 / 2$ when $b_1/\sigma_1 > \sqrt{\pi/2}$. We compute $\rho$ numerically in the following two cases: (i) $(\sigma_1, \sigma_2) = (1,1)$, (ii) $(\sigma_1, \sigma_2) = (1,10)$. In both cases, we have $b_1/\sigma_1 = C$.
Figure \ref{fig:rho} shows the relationship between $\rho$ and $b_1/\sigma_1$ in cases (i) and (ii). We find that $\rho$ coincides with the upper bound $\overline{\rho}(b_1/\sigma_1)$ in case (i), whereas in case (ii) it is strictly below the upper bound when $b_1/\sigma_1$ is large. The maximum regret of $\hat{\delta}_{\bm{w}_{\mathrm{minimax}}}$ is no more than $1.163$ times that of the global minimax regret rule even when $b_1/\sigma_1 = 3$. When $\sigma_1$ is very small, the upper bound becomes large; however, if there exists a study with large standard error whose covariates are close to $x_0$, then $\rho$ can be smaller than the upper bound.

%%% Figure 2 %%%
\begin{figure}[H]
\centering
\includegraphics[width=9cm]{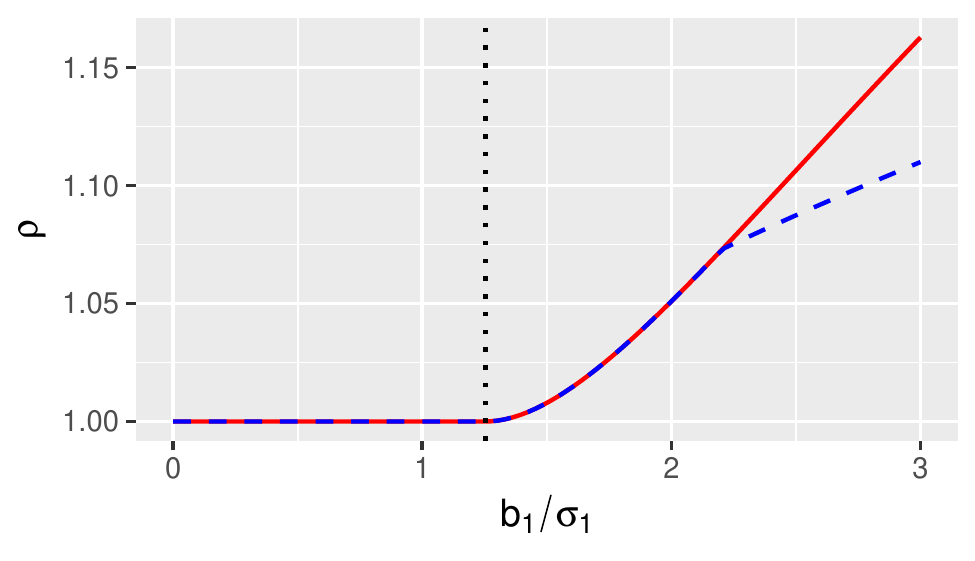}
\caption{{\footnotesize The red solid and blue dashed lines denote $\rho$ for $K=2$ in cases (i) $(\sigma_1, \sigma_2) = (1,1)$ and (ii) $(\sigma_1, \sigma_2) = (1,10)$, respectively. The red line coincides with $\overline{\rho}(b_1 / \sigma_1)$, where $b_1/\sigma_1 = C$ in both cases. The dotted line denotes $\sqrt{\pi/2}$.}}\label{fig:rho}
\end{figure}

%%%%%%%%%%%%%%%%%%%%%%%%%%%%%%%%%%%%%%%%%%%
\section{Numerical analysis}\label{sec:numerical}

To illustrate the results of the previous sections, we present some numerical analysis. Throughout this section, we set the study-specific covariates as equidistant grid points on $[0,1]$:
$$
x_k = (k-1)/(K-1) \ \ \text{and} \ \ \sigma_k = 1 \ \ \text{for $k=1, \cdots , K$.}
$$

We consider the following two parameter spaces:
\begin{eqnarray}
\mathcal{T}_1 &=& \{ \bm{\tau}: \tau_k = \beta_0 + \beta x_k,  \ \text{$\beta_0 \in \mathbb{R}$ and $\beta \in [-C,C]$}\}, \nonumber \\
\mathcal{T}_2 &=& \{ \bm{\tau}: |\tau_k - \tau_l | \leq C |x_k - x_l| \ \text{for any $k,l$} \}, \nonumber
\end{eqnarray}
where $C$ is a positive constant. For these two parameter spaces, we derive the minimax regret and minimax MSE rules and compare the maximum regrets of these two treatment rules.

In the case of the linear parameter space $\mathcal{T}_1$, one natural treatment rule is based on plugging in the OLS estimator,
$$
\hat{\delta}_{\bm{w}_{\text{OLS}}} \ \equiv \ \mathbf{1} \left\{ \tilde{x}_0'\hat{b} \geq 0 \right\},
$$
where $\tilde{x}_0 \equiv (1,x_0)'$ and $\hat{b}$ is the OLS estimator of $(\beta_0, \beta)'$. Because $\hat{b}$ is linear with respect to $(\hat{\tau}_1, \dots , \hat{\tau}_K)'$, this rule can be expressed as a linear aggregation rule, $\sum_{k=1}^K w_{\text{OLS},k} \hat{\tau}_k$.

Another natural treatment rule is based on plugging in the hierarchical Bayes (HB) estimator defined in (\ref{HB_rule}),
$$
\hat{\delta}_{\bm{w}_{\text{HB}}} \ \equiv \ \mathbf{1} \left\{ \sum_{k=1}^K w_{\text{HB},k} \hat{\tau}_k \geq 0 \right\},
$$
where $\bm{w}_{\text{HB}} \equiv (w_{\text{HB},1}, \dots, w_{\text{HB},K})' \equiv \tilde{\bm{w}}_{\text{HB}} / \sum_{k=1}^K \tilde{w}_{\text{HB},k}$ and $\tilde{\bm{w}}_{\text{HB}} \equiv (\tilde{w}_{\text{HB},1}, \dots, \tilde{w}_{\text{HB},K})' \equiv \bm{\Sigma}^{-1} \bm{X} \left(\bm{\Sigma}_{\tilde{\beta}}^{-1} + \bm{X}' \bm{\Sigma}^{-1} \bm{X} \right)^{-1} \tilde{x}_0$. We set the prior variance matrix as
\[
\bm{\Sigma}_{\tilde{\beta}} \ = \ \left( \begin{array}{cc}
10 & 0 \\
0 & 1.96 \cdot C 
\end{array} \right).
\]
Because the state space $\mathcal{T}_1$ does not restrict $\beta_0$, we specify a diffuse prior for $\beta_0$. In contrast, because $\mathcal{T}_1$ assumes that $\beta \in [-C,C]$, we assume that $\beta$ is contained in $[-C,C]$ with prior probability $0.95$.

We calculate $\bm{w}_{\text{OLS}}$, $\bm{w}_{\text{HB}}$, $\bm{w}_{\text{MSE}}$, and $\bm{w}_{\text{minimax}}$ for $K = 30$ and $x_0 = 0.1$. Table \ref{tab:numerical_T1} contains the results of this experiment for $C = 0.1$,  $1.0$, and $2.0$. Table \ref{tab:numerical_T1} shows the ratios of $b(\bm{w})$ and $s(\bm{w})$, and the ratios of the maximum regrets, that is,
$$
\cfrac{\max_{\bm{\tau} \in \mathcal{T}_1} R(\bm{\tau},\hat{\delta}_{\bm{w}_{\text{OLS}}})}{\max_{\bm{\tau} \in \mathcal{T}_1} R(\bm{\tau},\hat{\delta}_{\bm{w}_{\text{minimax}}})}, \ \ \cfrac{\max_{\bm{\tau} \in \mathcal{T}_1} R(\bm{\tau},\hat{\delta}_{\bm{w}_{\text{HB}}})}{\max_{\bm{\tau} \in \mathcal{T}_1} R(\bm{\tau},\hat{\delta}_{\bm{w}_{\text{minimax}}})}, \ \ \text{and} \ \ \cfrac{\max_{\bm{\tau} \in \mathcal{T}_1} R(\bm{\tau},\hat{\delta}_{\bm{w}_{\text{MSE}}})}{\max_{\bm{\tau} \in \mathcal{T}_1} R(\bm{\tau},\hat{\delta}_{\bm{w}_{\text{minimax}}})}.
$$
Because the OLS estimator is unbiased, the maximum bias of the OLS estimator is exactly zero. Hence, the ratio $b(\bm{w}_{\text{OLS}})/s(\bm{w}_{\text{OLS}})$ is exactly zero in all settings. For the HB rule, $b(\bm{w}_{\text{HB}})/s(\bm{w}_{\text{HB}})$ increases as $C$ increases. The ratio of $b(\bm{w}_{\text{minimax}})$ and $s(\bm{w}_{\text{minimax}})$ is smaller than that of $\bm{w}_{\text{MSE}}$ in all settings. This implies that the minimax regret criterion places more emphasis on the bias than the variance compared with the minimax MSE criterion. Table \ref{tab:numerical_T1} shows that the maximum regret of the minimax MSE rule is about 40 percent greater than the minimax regret when $C=1.0$, and the maximum regret of the hierarchical Bayes rule more than doubles the minimax regret. When $C=0.1$, the maximum regrets of $\bm{w}_{\text{MSE}}$ and $\bm{w}_{\text{HB}}$ are close to the minimax regret. If $C$ is sufficiently large, $\bm{w}_{\text{minimax}}$ is almost the same as $\bm{w}_{\text{OLS}}$. Hence, when $C = 1.0$ or $2.0$, the maximum regret of $\hat{\delta}_{\bm{w}_{\text{OLS}}}$ is almost identical to the minimax regret. In contrast, when $C$ is small, $\bm{w}_{\text{minimax}}$ is quite different from $\bm{w}_{\text{OLS}}$ and the maximum regret of $\hat{\delta}_{\bm{w}_{\text{OLS}}}$ is about 30 percent greater than the minimax regret.

% Table 1 %
\begin{table}[H]
\caption{Results for the linear parameter space $\mathcal{T}_1$.}
\begin{center}
  \begin{tabular}{l c c c c c c c} \hline 
      & $\frac{b(\bm{w}_{\text{OLS}})}{s(\bm{w}_{\text{OLS}})}$ & $\frac{b(\bm{w}_{\text{HB}})}{s(\bm{w}_{\text{HB}})}$ & $\frac{b(\bm{w}_{\text{MSE}})}{s(\bm{w}_{\text{MSE}})}$ & $\frac{b(\bm{w}_{\text{minimax}})}{s(\bm{w}_{\text{minimax}})}$ & ratio (OLS) & ratio (HB) & ratio (MSE) \rule[0mm]{0mm}{5mm} \vspace{1mm} \\ \hline \hline
$C=0.1$ & 0 & 0.22 & 0.21 & 0.17 & 1.30 & 1.03 & 1.02 \\ 
$C=1.0$ & 0 & 1.13 & 0.43 & 0.00 & 1.00 & 2.28 & 1.41 \\  
$C=2.0$ & 0 & 0.83 & 0.24 & 0.00 & 1.00 & 2.14 & 1.28 \\  \hline
  \end{tabular}
  \end{center}
  \footnotesize{Note: The ratios that are shown in the final three columns are the ratios of the maximum regrets, $\max_{\bm{\tau} \in \mathcal{T}} R(\bm{\tau},\hat{\delta}_{\bm{w}}) / \max_{\bm{\tau} \in \mathcal{T}} R(\bm{\tau},\hat{\delta}_{\bm{w}_{\text{minimax}}})$, for $\bm{w} = \bm{w}_{\text{OLS}}$, $\bm{w}_{\text{HB}}$, and $\bm{w}_{\text{MSE}}$.}
\begin{center}
\end{center}\label{tab:numerical_T1}
\end{table}

%\textcolor{red}{The left panel in Figure ?? shows a value of $\bm{\tau}$ maximizing the regret of $\hat{\delta}_{\bm{w}_{\text{minimax}}}$ for each $C=0.1,1,2$. A least favorable prior for $\bm{\tau}$ replicating $\hat{\delta}_{\bm{w}_{\text{minimax}}}$ as a Bayes rule can be constructed as a symmetric two-point prior allocating probability $1/2$ to the worst-case value of $\bm{\tau}$ and the other probability $1/2$ to its negative counterpart. At the worst case $\bm{\tau}$, the slope parameter $\beta_1$ agrees with the value of corresponding $C$ so that we can directly interpret the choice of $C$ as the value of $|\beta_1|$ at the worst-case $\bm{\tau}$. 
%The right panel in Figure ?? shows the regret functions profiled for $\tau_0$ for the minimax regret, minimax MSE, and hierarchical Bayes rules at $C=1$. For each rule, the maximum regret is attained at a similar value of $\tau_0$ around  $\tau_0 = 0.25$. We can observe that the profiled regret of the minimax regret rule dominates those of the minimax MSE and hierarchical Bayes rules uniformly over $\tau_0$. }

Figure \ref{fig:linear_least_regret} shows the regret functions profiled for $\tau_0$ for the minimax regret, minimax MSE, and hierarchical Bayes rules at $C=1$. For each rule, the maximum regret is attained at a similar value of $\tau_0$ around  $\tau_0 = 0.25$. We can observe that the profiled regret of the minimax regret rule dominates those of the minimax MSE and hierarchical Bayes rules almost uniformly over $\tau_0$.

\if0
% Figure 3 %
\begin{figure}[htbp]
    \begin{tabular}{ccc}
      \begin{minipage}[t]{0.45\hsize}
        \centering
        \includegraphics[width=7.8cm]{linear_least.pdf}
        \subcaption{{\footnotesize The least-favorable values of $\bm{\tau}$.}}
        \label{fig:linear_least}
      \end{minipage} &
      \begin{minipage}[t]{0.45\hsize}
        \centering
        \includegraphics[width=7.8cm]{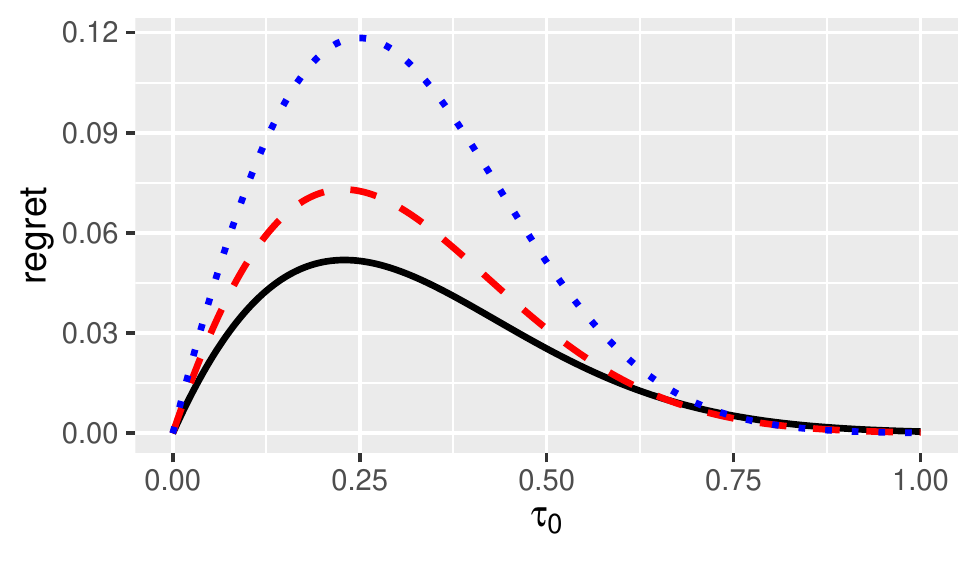}
        \subcaption{{\footnotesize The regret functions profiled for $\tau_0$.}}
        \label{fig:linear_regret_func}
      \end{minipage}
      
    \end{tabular}
    \caption{{\footnotesize The left panel denotes the least-favorable values of $\bm{\tau}$ for the minimax regret rule when $\mathcal{T}=\mathcal{T}_1$. The circles, triangles, and squares denote the least-favorable values of $\bm{\tau}$ for the minimax regret rule when $C=0.1$, $1$, and $2$, respectively. The blue dotted line denotes $X=x_0$. The right panel denotes the regret functions $R(t) \equiv \max_{\bm{\tau} \in \mathcal{T}_1: \tau_0 = t} R(\bm{\tau},\hat{\delta})$. The black solid, red dashed, and blue dotted lines denote the regret functions $R(t)$ for the minimax regret, minimax MSE, and hierarchical Bayes rules when $C=1$.}}\label{fig:linear_least_regret}
\end{figure}
\fi

% Figure 3 %
\begin{figure}[H]
\centering
\includegraphics[width=9cm]{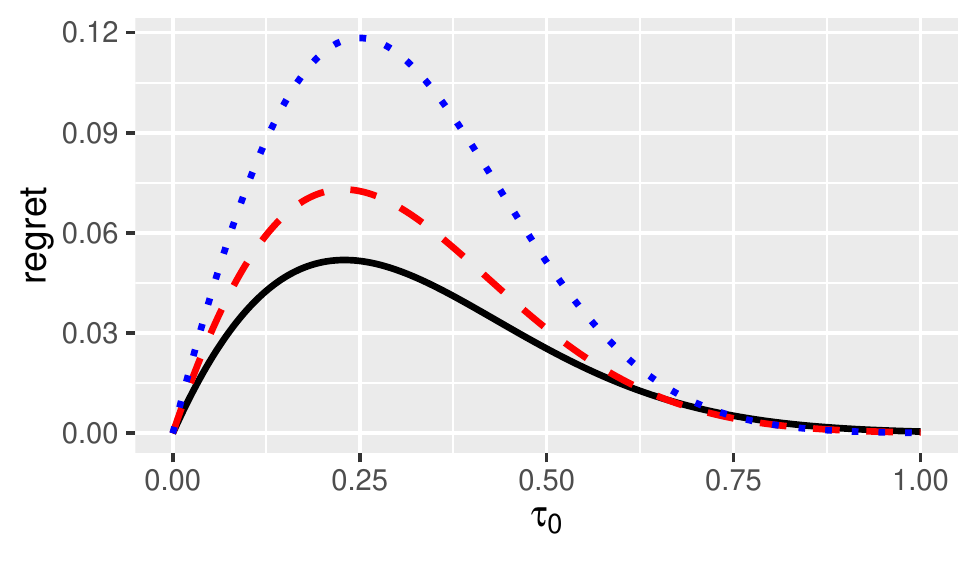}
\caption{{\footnotesize This figure plots the profiled regret functions $R(t) \equiv \max_{\bm{\tau} \in \mathcal{T}_1: \tau_0 = t} R(\bm{\tau},\hat{\delta})$. The black solid, red dashed, and blue dotted lines denote the profiled regret functions $R(t)$ for the minimax regret, minimax MSE, and hierarchical Bayes rules when $C=1$.}}\label{fig:linear_least_regret}
\end{figure}

Next, we consider the Lipschitz parameter space $\mathcal{T}_2$. We calculate $\bm{w}_{\text{minimax}}$ and $\bm{w}_{\text{MSE}}$ for $K=30$ and $x_0 = 0.5$. Similar to $\mathcal{T}_1$, we consider the hierarchical Bayesian model of (\ref{prior_dist}). We set the prior variance of $\tau_k$ as $10$ and the prior covariance of $\tau_k$ and $\tau_l$ as $10 \cdot \exp(-|x_k-x_l| / a)$ for some positive constant $a>0$. We choose a positive constant $a$ that satisfies $\frac{1}{K(K+1)/2} \sum_{k <l} P(|\tau_k - \tau_l| > C |x_k-x_l|) = 0.05$. Then, the posterior mean of $\tau_0$ is written as
\[
E[\tau_0 | \hat{\bm{\tau}}] \ = \ \bm{\Sigma}_{\bm{\tau}, 12} \bm{\Sigma}_{\bm{\tau}, 22}^{-1} \left(\bm{\Sigma}_{\bm{\tau},22}^{-1} + \bm{\Sigma}^{-1} \right)^{-1} \bm{\Sigma}^{-1} \hat{\bm{\tau}},
\]
which pins down the weights of the Bayes optimal decision rule $\bm{w}_{\text{HB}}$.

Table \ref{tab:numerical_T2} shows the ratios of $b(\bm{w})$ and $s(\bm{w})$ and the ratios of the maximum regrets for $C = 0.1$,  $1.0$, and $2.0$. The ratio of $b(\bm{w}_{\text{minimax}})$ and $s(\bm{w}_{\text{minimax}})$ is smaller than the ratios of $\bm{w}_{\text{HB}}$ and $\bm{w}_{\text{MSE}}$ in all settings. When $C$ is small, the maximum regret of the minimax MSE rule nearly attains the minimax regret. In contrast, when $C$ is large, the maximum regret of the minimax MSE rule is about 17 percent greater than the minimax regret. Similar to the minimax MSE rule, the maximum regret of the hierarchical Bayes rule nearly attains the minimax regret when $C$ is small. In addition, it is about 30 percent greater than the minimax regret when $C$ is large. Figure \ref{fig:weights_Lip} shows $\bm{w}_{\text{minimax}}$, $\bm{w}_{\text{MSE}}$, and $\bm{w}_{\text{HB}}$ for $C=1.0$. It shows that the minimax treatment rule is quite different from other treatment rules. The minimax MSE and Bayes criteria gives positive weights to most of the studies. In contrast, the minimax regret criterion yields weights that sharply concentrate around zero and rule out half of the studies.

% Table 2 %
\begin{table}[H]
\caption{Results of the Lipschitz parameter space $\mathcal{T}_2$.}
\begin{center}
  \begin{tabular}{c c c c c c} \hline 
      & $\frac{b(\bm{w}_{\text{HB}})}{s(\bm{w}_{\text{HB}})}$ & $\frac{b(\bm{w}_{\text{MSE}})}{s(\bm{w}_{\text{MSE}})}$ & $\frac{b(\bm{w}_{\text{minimax}})}{s(\bm{w}_{\text{minimax}})}$ & ratio (HB) & ratio (MSE) \rule[0mm]{0mm}{5mm} \vspace{1mm} \\ \hline \hline
$C=0.1$ & .107 & .141 & .131 & 1.02 & 1.01 \\ 
$C=1.0$ & .890 & .708 & .282 & 1.30 & 1.17 \\  
$C=2.0$ & .792 & .708 & .285 & 1.26 & 1.17 \\  \hline
  \end{tabular}
  \end{center}
  \footnotesize{Note: The ratios that are shown in the final two columns are the ratios of the maximum regrets, $\max_{\bm{\tau} \in \mathcal{T}} R(\bm{\tau},\hat{\delta}_{\bm{w}}) / \max_{\bm{\tau} \in \mathcal{T}} R(\bm{\tau},\hat{\delta}_{\bm{w}_{\text{minimax}}})$, for $\bm{w} = \bm{w}_{\text{HB}}$ and $\bm{w}_{\text{MSE}}$.}\label{tab:numerical_T2}
\end{table}

% Figure 4 %
\begin{figure}[H]
\centering
\includegraphics[width=9cm]{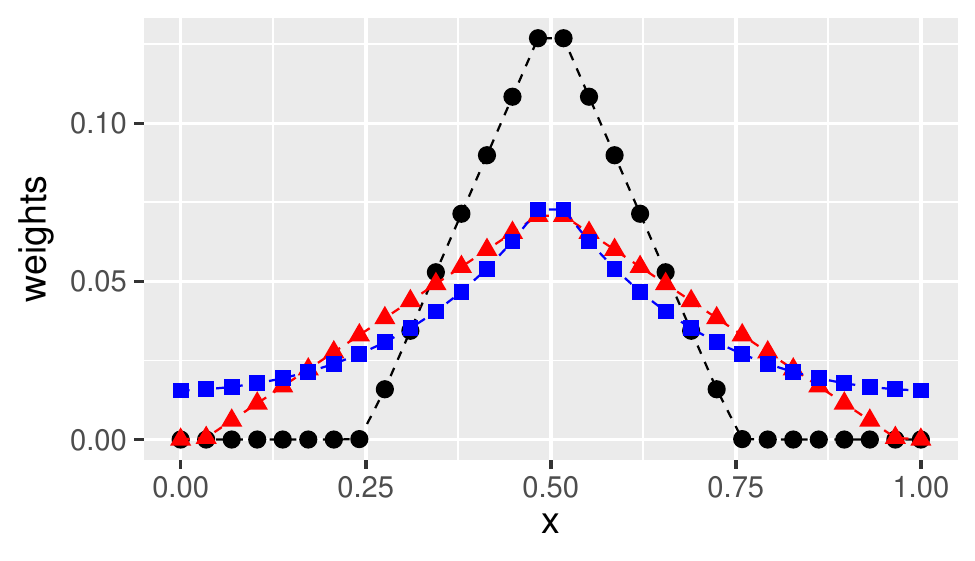}
\caption{{\footnotesize The black circles, red triangles, and blue squares denote $\bm{w}_{\text{minimax}}$, $\bm{w}_{\text{MSE}}$, and $\bm{w}_{\text{HB}}$ for $C=1.0$.}}\label{fig:weights_Lip}
\end{figure}

\if0
%%% Figure 2 %%%
\begin{figure}[H]
\centering
\includegraphics[width=9cm]{w_minimax_Lipschitz_C=1.pdf}
\caption{The minimax regret weights, $\bm{w}_{\text{minimax}}$, for $C=1.0$.}
\end{figure}
%%% Figure 3 %%%
\begin{figure}[H]
\centering
\includegraphics[width=9cm]{w_MSE_Lipschitz_C=1.pdf}
\caption{The minimax MSE weights, $\bm{w}_{\text{MSE}}$, for $C=1.0$.}
\end{figure}
%%% Figure 4 %%%
\begin{figure}[H]
\centering
\includegraphics[width=9cm]{w_HB_Lipschitz_C=1.pdf}
\caption{The HB weights, $\bm{w}_{\text{HB}}$, for $C=1.0$.}
\end{figure}
\fi

Figure \ref{fig:lipschitz_least} shows least favorable $\bm{\tau}$ maximizing the regret of $\hat{\delta}_{\bm{w}_{\text{minimax}}}$ for each $C=0.1,1,2$. For the studies that $\bm{w}_{\text{minimax}}$ assigns zero weights, the corresponding values of $\tau_k$ does not affect the regret, so we set them to zero in the plot. For the studies receiving the positive weights, the worst-case $\tau_k$'s vary linearly in $x_k$ and increasing toward the worst-case $\tau_0$ with the value of slope equal to $C$. 
The right panel, Figure \ref{fig:lipschitz_regret_func}, shows the profiled regret functions for $\tau_0$ for the minimax regret, minimax MSE, and hierarchical Bayes rules at $C=1$. Similarly to Figure \ref{fig:linear_least_regret}, these rules have maximum regret at a similar value of $\tau_0$, while a notable difference is that the minimax regret rule no longer dominates the other two rules in terms of the profiled regret.

% Figure 5 %
\begin{figure}[htbp]
    \begin{tabular}{ccc}
      \begin{minipage}[t]{0.45\hsize}
        \centering
        \includegraphics[width=7.8cm]{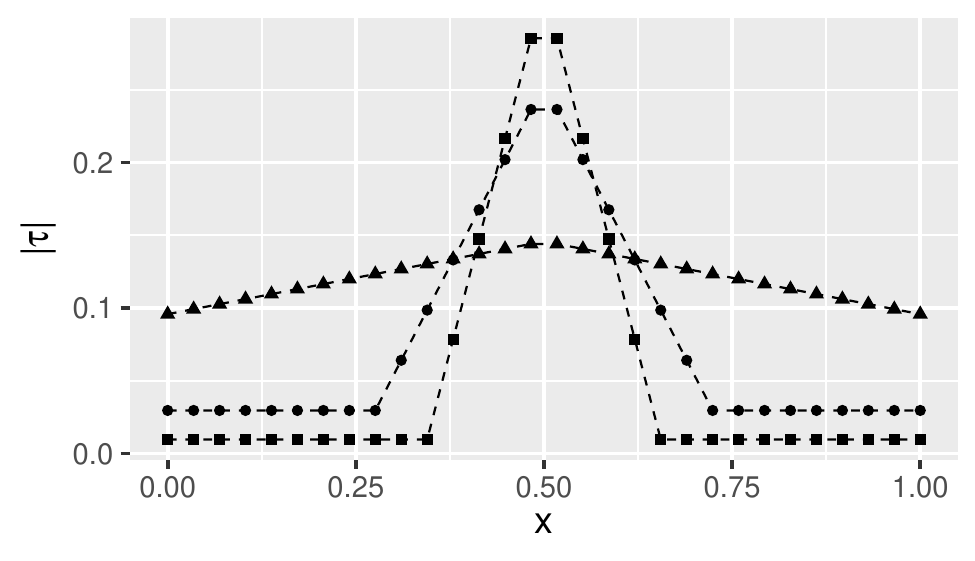}
        \subcaption{{\footnotesize The least-favorable values of $\bm{\tau}$.}}
        \label{fig:lipschitz_least}
      \end{minipage} &
      \begin{minipage}[t]{0.45\hsize}
        \centering
        \includegraphics[width=7.8cm]{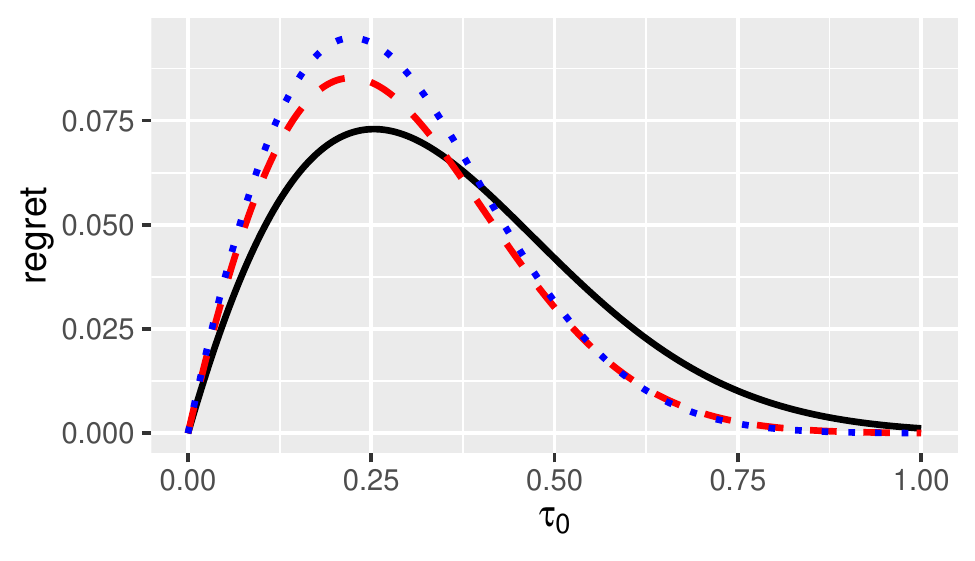}
        \subcaption{{\footnotesize The regret functions profiled for $\tau_0$.}}
        \label{fig:lipschitz_regret_func}
      \end{minipage}
      
    \end{tabular}
    \caption{{\footnotesize The left panel denotes the least-favorable values of $\bm{\tau}$ for the minimax regret rule when $\mathcal{T}=\mathcal{T}_2$. The circles, triangles, and squares denote the least-favorable values of $\bm{\tau}$ for the minimax regret rule when $C=0.1$, $1$, and $2$, respectively. The right panel denotes the profiled regret functions $R(t) \equiv \max_{\bm{\tau} \in \mathcal{T}_2: \tau_0 = t} R(\bm{\tau},\hat{\delta})$. The black solid, red dashed, and blue dotted lines denote the profiled regret functions $R(t)$ for the minimax regret, minimax MSE, and hierarchical Bayes rules when $C=1$.}}\label{fig:lipschitz_numerical}
\end{figure}

%%%%%%%%%%%%%%%%%%%%%%%%%%%%%%%%%%%%%%%%%%%
\section{Empirical application: Active labor market policies}\label{sec:empirical}

%We illustrate the use of our methods by means of two applications. The first application considers whether an active labor market policy should be adopted, and the second application considers whether a COVID-19 treatment should be approved.

%\subsection{Active Labor Market Policies}\label{subsec:active_labor}

We use the meta-database of \cite{card2017works}, which contains the estimates from over 200 recent studies of active labor market programs including training, subsidized employment, and job search assistance. We focus on papers that analyze RCT data to assess the impact of job training on the employment rate. This criterion reduces the meta-sample to 14 RCT estimates ($K=14$) collected from 8 different countries: Argentina, Brazil, Colombia, Dominica, Jordan, Nicaragua, Turkey, and the United States. Table \ref{tab:active_labor_list} lists the papers included in the meta-sample of this application. In this analysis, we denote the 14 RCT estimates by $\hat{\tau}_k$ and their standard errors by $\sigma_k$.

To form a vector of study characteristics $x_k$, $k=1,2,\dots,K$, we use five covariates that characterize the country and the sub-population on which the RCT study was performed. These are a gender dummy (male only $=$ 0, female only $=$ 1, mixed $=$ 0.5), an age dummy (age $<$ 25 only $=$ 1, age $\geq$ 25 only $=$ 0, both $=$ 0.5), an OECD dummy, the (standardized) GDP growth rate, and the (standardized) unemployment rate in the program period. Table \ref{tab:active_labor_estimates} shows the estimates, standard errors, and study characteristics in this meta-sample.

We consider whether the training program should be adopted in the following three target populations:
\begin{itemize}
\item Japan (in 2010): female, age $\geq 25$, OECD, $\Delta$GDP $=4.2$, unemp. $=5.1$
\item The UK (in 2010): female, age $\geq 25$, OECD, $\Delta$GDP $= 1.7$, unemp. $= 7.8$
\item Peru (in 2010): female, age $\geq 25$, not OECD, $\Delta$GDP $= 8.3$, unemp. $= 7.9$
\end{itemize}
Therefore, if the target population is Japan, the study characteristics of the target population are $x_0 = (1,0,1,4.2,5.1)'$.

%% Table 3 %%
\begin{table}[H]
\caption{The list of studies used in the job training program application.}
\begin{center}
\begin{tabular}{c l l}
  \hline
  No. & Paper & Country \\ \hline \hline
  1 & \cite{alzua2016long} & Argentina \\
  2 & \cite{attanasio2011subsidizing} & Colombia \\
  3 & \cite{calero2017can} & Brazil \\
  4 & \cite{fairlie2015behind} & United States \\
  5 & \cite{groh2012soft} & Jordan \\
  6 & \cite{hirshleifer2016impact} & Turkey \\
  7 & \cite{ibarraran2014life} & Dominica \\
  8 & \cite{macours2013demand} & Nicaragua \\
  \hline
\end{tabular}
\end{center}\label{tab:active_labor_list}
\end{table}

%% Table 4 %%
\begin{table}[H]
\caption{Estimates, standard errors, and study characteristics.}
\begin{center}
\begin{tabular}{c c c c c c c c c}
  \hline
  No. & Country & Estimates & S.E. & Gender & Age & OECD & GDP & Unemployment \\ \hline \hline
  1a & Argentina & 0.245 & 0.073 & male & both & no & 9.02 & 7.5 \\
  1b & Argentina & 0.005 & 0.039 & female & both & no & 9.02 & 7.5 \\
  2a & Colombia & -0.025 & 0.022 & male & $<$ 25 & no & 4.71 & 11.3 \\
  2b & Colombia & 0.061 & 0.024 & female & $<$ 25 & no & 4.71 & 11.3 \\
  3 & Brazil  & 0.129 & 0.069 & both & $<$ 25 & no & 0.87 & 6.9 \\
  4 & US & 0.071 & 0.046 & both & both & yes & 3.31 & 5.6 \\
  5 & Jordan & 0.015 & 0.032 & female & $<$ 25 & no & 2.31 & 12.5 \\
  6a & Turkey & 0.069 & 0.029 & male & $\geq$ 25 & yes & 6.72 & 10.3 \\
  6b & Turkey & 0.023 & 0.020 & female & $\geq$ 25 & yes & 6.72 & 10.3 \\
  6c & Turkey & 0.011 & 0.025 & female & $<$ 25 & yes & 6.72 & 10.3 \\
  6d & Turkey & -0.021 & 0.030 & male & $<$ 25 & yes & 6.72 & 10.3 \\
  7a & Dominica & 0.007 & 0.024 & female & $<$ 25 & no & 5.49 & 13.8 \\
  7b & Dominica & -0.024 & 0.023 & male & $<$ 25 & no & 5.49 & 13.8 \\
  8 & Nicaragua & 0.038 & 0.020 & both & both & no & 4.36 & 5.5 \\
  \hline
\end{tabular}
\end{center}\label{tab:active_labor_estimates}
\end{table}

We derive the minimax regret and minimax MSE rules with the following parameter space:
\begin{equation}
\mathcal{T}_C \ \equiv \ \left\{ \bm{\tau} : |\tau_k - \tau_l| \leq C \|x_k-x_l\| \ \text{for $k,l = 0, 1, \cdots, K$}  \right\}, \label{Tau_C}\nonumber
\end{equation}
with a prespecified Lipschitz constant $C \geq 0$. To compute our linear aggregation rule, we need to specify the value of $C$. If prior information about $\bm{\tau}$ is available and enables an appropriate choice of $C$, then $C$ can be selected using that information. However, if such information is not available, $C$ can be chosen in a data-driven manner, although this approach lacks theoretical justification.

In this application, we select $C$ using the following two data-driven methods; leave-one-out cross-validation and marginal likelihood maximization. In leave-one-out cross-validation, one first removes a single study, assigns the treatment based on the remaining studies, and computes the welfare of this decision using the estimated ATE from the omitted study. This procedure is repeated for all studies, and $C$ is chosen to maximize the resulting welfare. In marginal likelihood maximization, a prior distribution on $\bm{\tau}$ with support $\mathcal{T}_C$ is assumed, and $C$ is chosen to maximize the marginal likelihood of the data. For details on these two methods, see \ref{subsec:choice_C}. Leave-one-out cross-validation selects $C=0.025$, while marginal likelihood maximization selects $C=0.040$,  $0.045$, $0.040$ for the target populations of Japan, the UK, and Peru, respectively. For $C=0.025$, leave-one-out cross-validation correctly predict the sign of $\hat{\tau}_k$ in 11 out of the 14 studies. In what follows, we primarily base our analysis on the value selected by cross-validation.

% Figure 6 %
\begin{figure}[htbp]
    \begin{tabular}{ccc}
      \begin{minipage}[t]{0.45\hsize}
        \centering
        \includegraphics[width=8.0cm]{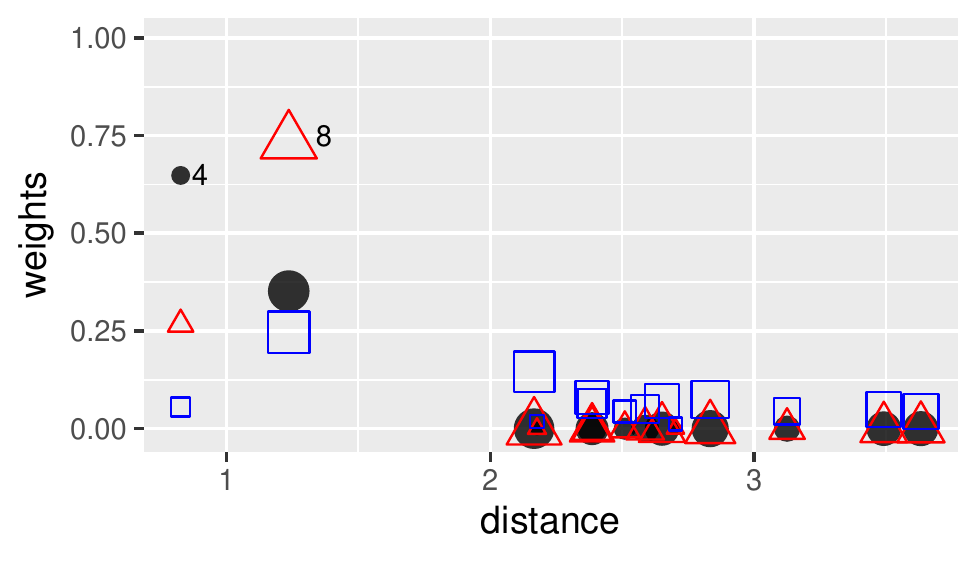}
        \subcaption{{\footnotesize Japan.}}
        \label{fig:JPN}
      \end{minipage} &
      \begin{minipage}[t]{0.45\hsize}
        \centering
        \includegraphics[width=8.0cm]{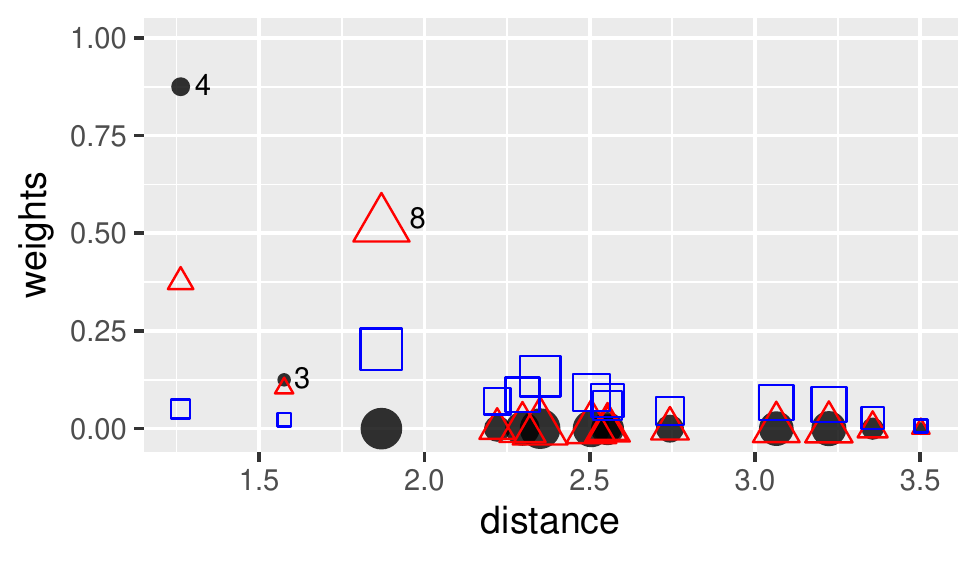}
        \subcaption{{\footnotesize The UK.}}
        \label{fig:UK}
      \end{minipage} \\
   
      \begin{minipage}[t]{0.45\hsize}
        \centering
        \includegraphics[width=8.0cm]{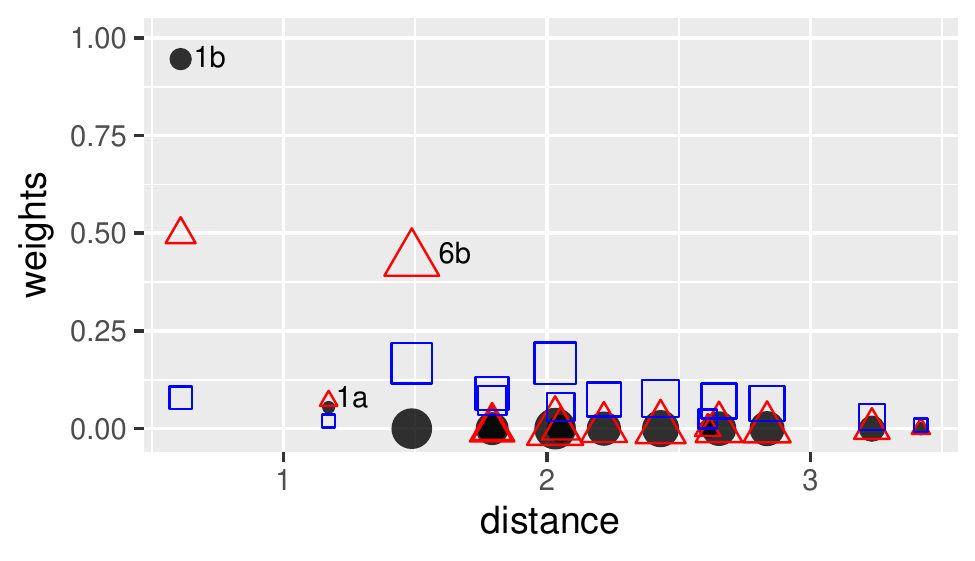}
        \subcaption{{\footnotesize Peru.}}
        \label{fig:PER}
      \end{minipage} &
      
    \end{tabular}
    \caption{{\footnotesize The minimax regret, minimax MSE, and hierarchical Bayes weights when $C=0.025$. The black circles, red triangles, and blue squares denote $\bm{w}_{\text{minimax}}$, $\bm{w}_{\text{MSE}}$, and $\bm{w}_{\text{HB}}$, respectively. The horizontal axis measures the Euclidean distance between $x_k$ and $x_0$, $\|x_k - x_0\|$. The size of the plotted marker is proportional to the precision of the estimates, i.e.,  a smaller $\hat{\sigma}_k$ corresponds to a larger marker.}}\label{fig:weights_ALMP}
\end{figure}

%% Table 5 %%
\begin{table}[H] 
\caption{The minimax regret, minimax MSE, and hierarchical Bayes rules when $C=0.025$.}
\begin{center}
\begin{tabular}{l c c c}
  \hline
  & Japan & UK & Peru \\ \hline \hline
  $\hat{\tau}_0$ with $\bm{w}_{\text{minmax}}$ & 0.059 & 0.078 & 0.018 \\
  $\hat{\tau}_0$ with $\bm{w}_{\text{MSE}}$ & 0.046 & 0.059 & 0.031 \\
  $\hat{\tau}_0$ with $\bm{w}_{\text{HB}}$ & 0.029 & 0.026 & 0.025 \\
  ratio (MSE) & 1.107 & 1.215 & 1.180 \\
  ratio (HB) & 3.131 & 2.580 & 3.356 \\
  $w_{\text{minmax},k} \geq 1/K$ & Nicaragua, US  & Brazil, US & Argentina \\ \hline
\end{tabular}
\label{tab:emp_application_regret}
\end{center}
\end{table}

Using $\mathcal{T}_C$ as the parameter space, we derive the minimax regret, minimax MSE, and hierarchical Bayes rules for each target population. Figures \ref{fig:JPN}--\ref{fig:PER} plot $\bm{w}_{\text{minimax}}$, $\bm{w}_{\text{MSE}}$, and $\bm{w}_{\text{HB}}$ for the three different target populations when $C=0.025$. Similar to Section 4, the hierarchical Bayes rule uses the prior $\bm{\tau} \ \sim \ \mathcal{N}\left( \bm{0}, \bm{\Sigma}_{\bm{\tau}} \right)$ with $\bm{\Sigma}_{\bm{\tau}}[k,l] = \exp (-\|x_{k-1} - x_{l-1}\| / a)$ and we choose $a$ satisfying $\frac{1}{K(K+1)/2} \sum_{k < l} P(|\tau_k - \tau_l| > C \|x_{k} - x_{l}\|) = 0.05$. The horizontal axis measures the Euclidean distance between $x_k$ and $x_0$. The size of the plotted symbol is proportional to the precision of the estimates, i.e., a smaller $\hat{\sigma}_k$ corresponds to a larger symbol. The figures show that, overall, both $\bm{w}_{\text{minimax}}$ and $\bm{w}_{\text{MSE}}$ tend to put greater weight on those studies that are in similar in terms of their population characteristics. This tendency is more evident for the minimax regret weights $\bm{w}_{\text{minimax}}$ than for the minimax MSE weights $\bm{w}_{\text{MSE}}$. 

We note that $\bm{w}_{\text{minimax}}$ differs from $\bm{w}_{\text{MSE}}$ for every target population. In all cases, the minimax regret criterion puts the most weight on the closest study. In contrast, the minimax MSE criterion can put the largest weight on a study that is not closest provided that it has a small standard error. For instance, in the case of Japan, the minimax regret weight of the closest study is more than 0.6 but the minimax MSE weight is about 0.3. These results reflect the different degrees of bias variance trade-off that the minimax regret and minimax-MSE weights aim to balance out, as discussed in Section 3.1. When we use the values of $C$ selected by marginal likelihood maximization, the weight assigned to the closest study is approximately one in all cases. In contrast, as in the case of $C=0.025$, the other criteria assign substantial weights to studies that are not the closest.

Table \ref{tab:emp_application_regret} lists $\hat{\tau}_0(\bm{w}_{\text{minimax}})$, $\hat{\tau}_0(\bm{w}_{\text{MSE}})$, $\hat{\tau}_0(\bm{w}_{\text{HB}})$, the ratio of the maximum regrets, and the countries that were awarded minimax regret weights larger than $1/K = 1/14$. The table also shows that the minimax regret and minimax MSE rules select different decisions in some cases. For example, the average annual salary amongst Japanese women aged 25--29 years is approximately \$30,000; if the cost per person of adopting the policy is \$1,500 and individuals that start a new job work for one year, we could set $c_0 = 0.05$.\footnote{According to \cite{fairlie2015behind}, the cost per person of implementing the program in the United States is \$1,321.} Then, the recommendation of the minimax regret criterion is to introduce the policy in Japan. However, the minimax MSE criterion does not recommend the introduction of the policy in Japan.

For all of the target populations, the maximum regret of the minimax MSE rule is more than 10 percent greater than that of the minimax regret. For Japan and the UK, the minimax regret aggregation rule puts the most weight on the estimates of the US. In contrast, for Peru, the minimax regret criterion puts most of the weight on one estimate obtained from Argentina.

To examine the sensitivity of the choice of $C$, we compute $\hat{\tau}_0(\bm{w}_{\text{minimax}})$ for different values of $C$. Figure \ref{fig:sensitivity} plots $\hat{\tau}_0(\bm{w}_{\text{minimax}})$ when the target populations are Japan, the UK, and Peru. This figure shows that all values of $\hat{\tau}_0(\bm{w}_{\text{minimax}})$ are below $0.05$ for Peru. This implies that, when the cost is set to $c_0=0.05$, the choice of $C$ does not alter the minimax regret policy recommendation for Peru. For Japan and the UK, the decision $\hat{\delta}_{\bm{w}_{\text{minimax}}}$ remains the same for $C \geq 0.02$. Thus, for all the three target populations, the policy recommendations based on Table \ref{tab:emp_application_regret} are not sensitive to the choice of the constant for $C \geq 0.02$.

% Figure 7 %
\begin{figure}[H]
\centering
\includegraphics[width=9cm]{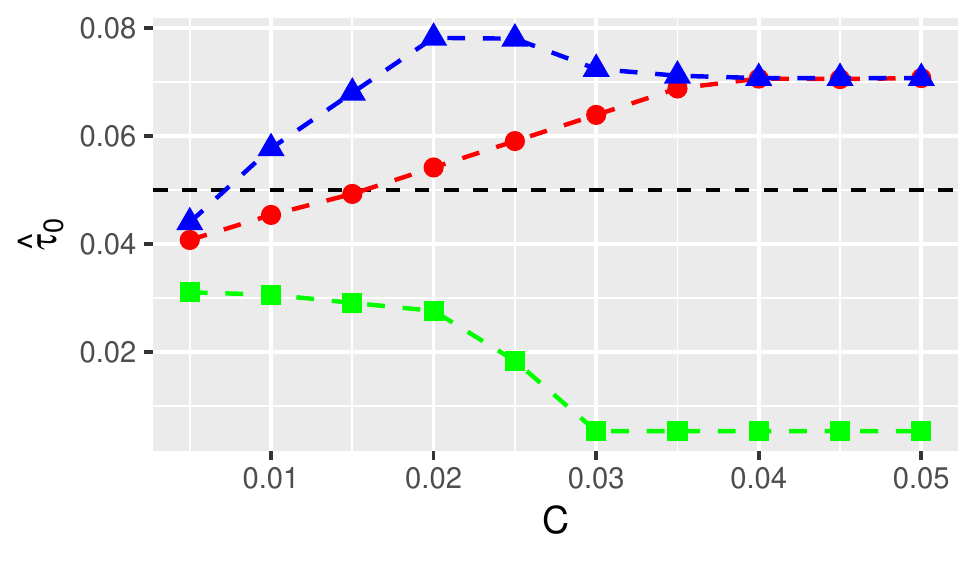}
\caption{{\footnotesize The red circles, blue triangles, and green squares denote $\hat{\tau}_0(\bm{w}_{\text{minimax}})$ when the target populations are Japan, the UK, and Peru, respectively. The dashed line denotes $0.05$.}}\label{fig:sensitivity}
\end{figure}

\Copy{global_optimal_empirical}{
We compare the proposed linear aggregation rule with the global minimax regret rule. Let $k^{\ast}$ be the closest study, defined as the one that minimizes the distance $\|x_k-x_0\|$. Table \ref{tab:active_labor_rho} reports the values of $C\|x_{k^{\ast}}-x_0\| / \sigma_{k^{\ast}}$ when $C$ is selected by cross-validation or marginal likelihood maximization. Table \ref{tab:active_labor_rho} shows that the values of $C\|x_{k^{\ast}}-x_0\| / \sigma_{k^{\ast}}$ are below $\sqrt{\pi / 2} \simeq 1.25$ in all cases. It follows from Theorem \ref{thm:global} that the proposed linear aggregation rule attains the global minimax regret. Even if $C$ is doubled from its current value, Theorem \ref{thm:global} implies that the maximum regret of the proposed linear aggregation rule is only 10 percent worse than that of the global minimax regret rule. Therefore, in this application, restricting the rule to non-randomized linear aggregation rules does not result in a substantial increase in the maximum regret.}

%% Table 6 %%
\begin{table}[H]
\caption{The values of $C\|x_{k^{\ast}}-x_0\| / \sigma_{k^{\ast}}$.}\label{tab:active_labor_rho}
\begin{center}
\begin{tabular}{l c c }
  \hline
  & Cross-validation & Marginal likelihood \\ \hline \hline
  Japan & 0.446 & 0.713 \\
  UK & 0.680 & 1.224 \\
  Peru & 0.400 & 0.633 \\ \hline
\end{tabular}
\end{center}
\end{table}

%\textcolor{red}{Because the minimax regret rule minimizes the worst-case regret, it performs well when the true $\bm{\tau}$ is close to the worst-case configuration. However, when the true $\bm{\tau}$ is far from the worst case, the regret of the proposed rule may be larger than that of other rules.}

%%%%%%%%%%%%%%%%%%%%%%%%%%%%%%%%%%%%%%%%%%%
\section{Conclusion}\label{sec:conclusion}
Motivated by the recently proposed paradigm of `patient-centered meta-analysis' (\citet{manski2020towards}), this paper develops a method to aggregate available evidence and inform optimal treatment choice for a target population that is of interest to the planner. Building upon the framework of statistical decision theory and adopting the minimax regret criterion, we obtain a minimax regret treatment choice rule that is simple to implement in practice. The key steps of our analysis that deliver analytical and computational tractability are to
constrain decision rules to the class of linear aggregation rules and to restrict the
parameter space to a symmetric and invariant one (Assumption \ref{parameter_space}). These conditions for the parameter space are mild and hold in numerous contexts.  

Several questions remain unanswered. First, when $\bm{\tau}$ is constrained to Lipschitz vectors while the Lipschitz constant $C$ is unknown, we do not know what is a theoretically justifiable data-driven way to select $C$. \Copy{data-driven_C}{In the presented empirical applications, we selected $C$ by leave-one-out cross-validation and marginal likelihood maximization without any analytical justification for this choice. As an alternative approach, an estimation method adaptive to the bias bound as in \cite{armstrong2025adapting} might be applicable to the current setting. analyze the related settings in which the bound on the bias is unknown and propose an adaptive estimation method under the worst-case adaptation regret.} Second, our framework assumes away any publication bias of published estimates despite a growing concern in the scientific community about it and increasing interest in how to detect and correct for any such bias (see, e.g., \citet{AK19}). Third, other than the standard errors of the estimates, our framework does not offer any way to incorporate a measure of the credibility of reported estimates. Depending on how the data were sampled and what identifying assumptions the estimate relies on, the credibility of studies can vary greatly. How to incorporate a measure of the credibility of reported estimates beyond their standard errors remains an interesting open question. 
\Copy{distributional_aggregation}{Fourth, this paper considers the standard setting of meta-analysis where the researcher can access only to summary statistics of the outcomes and covariates. We leave for future research a setting where the researcher can access to the micro-data of each study and perform a treatment choice based on aggregation of distributional information or individualized treatment assignments exploiting individual-level observable characteristics.}
%\clearpage

\appendix

\renewcommand{\thesection}{Appendix \Alph{section}}

\renewcommand{\theequation}{A.\arabic{equation}}
\setcounter{equation}{0}

%%%%%%%%%%%%%%%%%%%%%%%%%%%%%%%%%%%%%%%%%%%
\section{Proofs and Lemmas}\label{sec:proof}

\begin{proof}[Proof of Theorem \ref{thm:finite}]
From Theorem \ref{thm:minimax}, if $b(\bm{w})$ is bounded for some $\bm{w}$, then the minimax regret is bounded. Without loss of generality, we assume that $|\tau_k-\tau_0|$ is bounded. Setting $w_k = 1$, we have $b(\bm{w}) = \max_{\tau \in \mathcal{T}, \tau_0 = 0} \left\{ \sum_{k=1}^K w_k (\tau_k-\tau_0) \right\} = \tau_k - \tau_0 < M$. Hence, the minimax regret is finite.
\end{proof} \vspace{0.15in}

%% Lemma 1 %%
\begin{Lemma}\label{lemma:1}
For any $v, a \geq 0$, we have
\begin{equation}
v\cdot\Phi(-v) + \Phi(-v)\cdot a \ \ \leq \ \ \eta(a) \ \ \leq \ \ \eta(0) + a. \label{Lemma 1}
\end{equation}
\end{Lemma}
\begin{proof}[Proof of Lemma \ref{lemma:1}]
We observe that
\begin{eqnarray}
\eta(a) &=& \max_{t \geq 0} \left\{ (t-a)\Phi \left(-(t-a) \right) + a \cdot \Phi \left(-(t-a) \right) \right\}. \nonumber
\end{eqnarray}
Because $(t-a)\Phi \left(-(t-a) \right) \leq \max_{t'\geq 0} \{t' \cdot \Phi(-t') \} = \eta(0)$, we have
$$
\eta(a) \ \leq \ \eta(0) + a.
$$
The lower bound is obtained by substituting $t = a + v$ for $t \cdot \Phi(-t+a)$.
\end{proof} \vspace{0.15in}

%% Lemma 2 %%
\begin{Lemma} \label{lemma:2}
For any $a \geq 0$, we have
\begin{equation}
\eta(0) \sqrt{1+a^2} \ \ \leq \ \ \eta(a) \ \ \leq \ \ \sqrt{1+a^2}. \label{eq:Lemma2}
\end{equation}
\end{Lemma}
\begin{proof}[Proof of Lemma \ref{lemma:2}]
First, we show that the derivative of $\eta(a)$ is bounded from below by 0 and from above by 1. For $a \geq 0$, we define
$$
t^*(a) \equiv \text{arg} \max_{t \geq 0} \left\{ t \cdot \Phi(-t+a) \right\}.
$$
Then, by the first order condition, we have
\begin{equation}
t^*(a)\cdot \phi\left( -t^*(a)+a \right) \ = \ \Phi\left( -t^*(a)+a \right). \label{lem_2_1}
\end{equation}
It follows from the envelope theorem that
\begin{equation}
\eta'(a) \ = \ t^*(a)\cdot \phi\left( -t^*(a)+a \right). \label{lem_2_2}
\end{equation}
Hence, from (\ref{lem_2_1}) and (\ref{lem_2_2}), we obtain $\eta'(a) \ = \ \Phi\left( -t^*(a)+a \right)$, which implies $0 \leq \eta'(a) \leq 1$.

Next, we show the right-most inequality of (\ref{eq:Lemma2}). For $0 \leq a \leq 2$, we have $\eta(a) \leq 0.17 + a \leq \sqrt{1+a^2}$.  Because $\eta'(a)$ is bounded from below by 0 and from above by 1, we have $\eta(a) \leq (a-2)+\eta(2)$ for $a \geq 2$. From numerical evaluation, we obtain $\eta(2) \risingdotseq 1.051$, and hence $\eta(a) \leq a-0.5 \leq  \sqrt{1+a^2}$ for $a \geq 2$.

Finally, we show the left-most inequality of (\ref{eq:Lemma2}). From (\ref{lem_2_1}), $t^*(a)$ is a solution to the following equation:
$$
t \ = \ \cfrac{\Phi(-(t-a))}{\phi(t-a)},
$$
where $\Phi(-x)/\phi(x)$ is the Mills ratio of the standard normal distribution.  Because we know that $\Phi(-x)/\phi(x)$ is a strictly decreasing function, we find that $t \mapsto t  \cdot \Phi(-t+a)$ is a uni-modal function. For any $d>0$, we observe that
\begin{eqnarray}
\left. \frac{\partial}{\partial t} \left\{ t  \cdot \Phi(-t+(a+d)) \right\} \right|_{t=t^*(a)+d} &=& \Phi(-t^*(a)+a) - (t^*(a)+d)\cdot \phi(-t^*(a)+a) \nonumber \\
&=& -d \cdot \phi(-t^*(a)+a) \ < \ 0, \nonumber
\end{eqnarray}
where the second equality follows from (\ref{lem_2_1}). Because $t \mapsto t  \cdot \Phi(-t+a)$ is uni-modal, we obtain $t^*(a+d) < t^*(a)+d$ for any $d>0$. This implies that $-t^*(a)+a$ is strictly increasing in $a$. Moreover, since $\eta'(a) \ = \ \Phi\left( -t^*(a)+a \right)$ is strictly increasing, $\eta(a)$ is convex. Because we have $\frac{d}{dt} \{ t \cdot \Phi(-t) \} \big|_{t=0.8} = \Phi(-0.8) - 0.8 \phi(-0.8) < 0$, we obtain $t^{\ast}(0) < 0.8$. This implies that $\eta'(0) = \Phi \left( - t^{\ast}(0) \right) > \Phi(-0.8) \risingdotseq 0.212 > \eta(0)$. Because we have $\frac{d}{da} \left\{ \eta(0)\sqrt{1+a^2} \right\} \leq \eta(0)$, we obtain the left inequality of (\ref{eq:Lemma2}).
\end{proof} \vspace{0.15in}

\begin{proof}[Proof of Theorem \ref{thm:regret_bound}]
The first lower and upper bounds follow directly from Theorem \ref{thm:minimax} and Lemma \ref{lemma:1}. From Theorem \ref{thm:minimax} and Lemma \ref{lemma:2}, we have
\begin{eqnarray}
\max_{\bm{\tau} \in \mathcal{T}} R(\bm{\tau},\hat{\delta}_{\bm{w}}) &=& s(\bm{w}) \eta \left( b(\bm{w})/s(\bm{w}) \right) \nonumber \\
& \leq & s(\bm{w}) \times \sqrt{1+b^2(\bm{w})/s^2(\bm{w})} \nonumber \\
&=& \sqrt{b^2(\bm{w}) + s^2(\bm{w})}. \nonumber
\end{eqnarray}
Similarly, using the lower bound of Lemma \ref{lemma:2}, we obtain 
$$
\eta(0) \cdot \sqrt{b^2(\bm{w}) + s^2(\bm{w})} \ \leq \ \max_{\bm{\tau} \in \mathcal{T}} R(\bm{\tau},\hat{\delta}_{\bm{w}}).
$$
This completes the proof.
\end{proof} \vspace{0.15in}

\begin{proof}[Proof of Theorem \ref{thm:MSE}]
Because $\bm{w}_{\text{minimax}}$ minimizes the maximum regret, we have
$$
\cfrac{\max_{\bm{\tau} \in \mathcal{T}} R(\bm{\tau},\hat{\delta}_{\bm{w}_{\text{minimax}}})}{\max_{\bm{\tau} \in \mathcal{T}} R(\bm{\tau},\hat{\delta}_{\bm{w}_{\text{MSE}}})} \ \leq \ 1.
$$

Next, we show the right-most inequality of (\ref{Theorem 4}). From Theorem \ref{thm:regret_bound}, we obtain
\begin{eqnarray}
\max_{\bm{\tau} \in \mathcal{T}} R(\bm{\tau},\hat{\delta}_{\bm{w}_{\text{MSE}}}) \ \leq \  \sqrt{b^2(\bm{w}_{\text{MSE}}) + s^2(\bm{w}_{\text{MSE}})}. \nonumber 
\end{eqnarray}
Because $\bm{w}_{\text{MSE}}$ minimizes the maximum mean squared error, we obtain
\begin{eqnarray}
\sqrt{b^2(\bm{w}_{\text{MSE}}) + s^2(\bm{w}_{\text{MSE}})} & \leq & \sqrt{b^2(\bm{w}_{\text{minimax}}) + s^2(\bm{w}_{\text{minimax}})}. \nonumber 
\end{eqnarray}
Using Theorem \ref{thm:regret_bound} again, we have 
\begin{eqnarray}
\sqrt{b^2(\bm{w}_{\text{minimax}}) + s^2(\bm{w}_{\text{minimax}})}
& \leq & \frac{1}{\eta(0)} \times \max_{\bm{\tau} \in \mathcal{T}} R(\bm{\tau},\hat{\delta}_{\bm{w}_{\text{minimax}}}). \nonumber
\end{eqnarray}
Hence, we obtain the right-most inequality of (\ref{Theorem 4}).
\end{proof} \vspace{0.15in}

\begin{proof}[Proof of Theorem \ref{thm:global}]
Since it follows from the definition of $\rho$ that we have $\rho \geq 1$, we consider the upper bound of $\rho$. Since \cite{montiel2023decision} and \cite{yata2021optimal} show that the minimax regret rule takes the form $\hat{\delta}_{\bm{w}}$ when $b_1 / \sigma_1 \leq \sqrt{\pi / 2}$, we have $\rho = 1$ when $b_1 / \sigma_1 \leq \sqrt{\pi / 2}$. Henceforth, we consider the case where $b_1 / \sigma_1 > \sqrt{\pi / 2}$. Define 
$$
(\sigma_Z^{\ast},\bm{w}^{\ast}) \in \mathrm{arg} \min_{\sigma_Z, \bm{w}} \max_{\bm{\tau} \in \mathcal{T}} R(\bm{\tau}, \hat{\delta}_{\sigma_Z, \bm{w}}).
$$
We observe that
\begin{eqnarray*}
\rho &=& \frac{\min_{\bm{w}} \max_{\bm{\tau} \in \mathcal{T}} R(\bm{\tau}, \hat{\delta}_{\bm{w}})}{\min_{v, \bm{w}} \max_{\bm{\tau} \in \mathcal{T}} R(\bm{\tau}, \hat{\delta}_{v, \bm{w}})} \\
& \leq & \frac{\max_{\bm{\tau} \in \mathcal{T}} R(\bm{\tau}, \hat{\delta}_{\bm{w}^{\ast}})}{\max_{\bm{\tau} \in \mathcal{T}} R(\bm{\tau}, \hat{\delta}_{\sigma_Z^{\ast}, \bm{w}^{\ast}})} \ = \ \frac{s(0,\bm{w}^{\ast}) \cdot \eta \left( \frac{b(\bm{w}^{\ast})}{s(0,\bm{w}^{\ast})} \right)}{s(\sigma_Z^{\ast},\bm{w}^{\ast}) \cdot \eta \left( \frac{b(\bm{w}^{\ast})}{s(\sigma_Z^{\ast},\bm{w}^{\ast})} \right)}.
\end{eqnarray*}
It follows from \cite{montiel2023decision} and \cite{yata2021optimal} that when $b_1 / \sigma_1 > \sqrt{\pi / 2}$, we have
\[
\sigma_Z^{\ast} = \sqrt{2b_1^2/\pi - \sigma_1^2}, \ \ \bm{w}^{\ast} = (1, 0, \ldots, 0)'.
\]
This implies that
\begin{eqnarray}
\rho & \leq & \frac{\sigma_1 \cdot \eta\left( b_1 / \sigma_1 \right)}{\sqrt{2/\pi} \cdot b_1 \cdot \eta\left( \sqrt{\pi / 2} \right)} \ = \ \frac{\sqrt{\pi / 2}}{\eta\left( \sqrt{\pi / 2} \right)} \cdot \frac{\eta\left( b_1 / \sigma_1 \right)}{b_1 / \sigma_1}. \label{rho_ineq1}
\end{eqnarray}

We investigate the properties of the function $a \mapsto \eta(a) / a$. Since we have $\frac{d}{da}\{\eta(a)/a\} = \frac{a\eta'(a)-\eta(a)}{a^2}$, the sign of $\frac{d}{da}\{\eta(a)/a\}$ is determined by the sign of $a\eta'(a)-\eta(a)$. From the proof of Lemma \ref{lemma:2}, we obtain
\begin{eqnarray*}
a\eta'(a)-\eta(a) &=& a \cdot \Phi\left( -t^{\ast}(a)+a \right) - t^{\ast}(a) \cdot \Phi\left( -t^{\ast}(a)+a \right) \\
&=& \Phi\left( -t^{\ast}(a)+a \right) \cdot \left\{ - t^{\ast}(a) + a \right\},
\end{eqnarray*}
where the function $-t^{\ast}(a)+a$ is strictly increasing in $a$. It follows from (\ref{lem_2_1}) that $t^{\ast}(a)$ solves the equation $t - m(t-a) = 0$, where $m(x) \equiv \frac{\Phi(-x)}{\phi(x)}$ is the Mills ratio. Since $m'(x) < 0$ for all $x$, we obtain $\frac{\partial}{\partial t} \left\{ t - m(t-a) \right\} > 0$ for all $t$ and $a$. Hence, by the implicit function theorem, $t^{\ast}(a)$ is continuous in $a$. This implies that the function $a\eta'(a)-\eta(a)$ must have at most one intersection with the horizontal axis. Since $t \mapsto t \cdot \Phi(-t+a)$ is a unimodal function and we have $\frac{d}{dt} \{ t \cdot \Phi(-t+2) \} \big|_{t=2} = \Phi(0) - 2 \cdot \phi(0) < 0$, we have $-t^{\ast}(2)+2 > 0$. Furthermore, because $t^{\ast}(0)$ is strictly positive, $-t^{\ast}(a)+a$ is strictly negative when $a$ is sufficiently close to zero. Therefore, the function $a\eta'(a)-\eta(a)$ has exactly one zero. Let $a^{\ast}$ denote the unique value satisfying $a\eta'(a)-\eta(a)=0$. It follows that $\frac{\eta(a)}{a}$ is increasing on $(a^{\ast}, \infty)$.

Next, we show $a^{\ast} = \sqrt{\pi/2}$ and $\eta(a^{\ast})/ a^{\ast} = \frac{1}{2}$. From the above argument, $a^{\ast}$ satisfies $\frac{d}{da}\{\eta(a)/a\} \big|_{a=a^{\ast}} = 0$, which implies $t^{\ast}(a^{\ast})=a^{\ast}$ from the above argument. Furthermore, it follows from the first-order condition that $t^{\ast}(a)$ satisfies
\[
\Phi \left( -t^{\ast}(a) + a \right) - t^{\ast}(a) \cdot \phi\left( -t^{\ast}(a) + a \right) \ = \ 0.
\]
This implies that $\Phi(0) - a^{\ast} \phi(0) = 0$, which implies $a^{\ast} = \Phi(0)/\phi(0) =\sqrt{\pi / 2}$. Hence, we obtain 
$$
\frac{\eta(a^{\ast})}{a^{\ast}} \ = \ \frac{t^{\ast}(a^{\ast}) \cdot \Phi\left( -t^{\ast}(a) + a \right)}{a^{\ast}} \ = \ \Phi(0) \ = \ \frac{1}{2}. 
$$
As a result, it follows from (\ref{rho_ineq1}) that when $b_1 / \sigma_1 > \sqrt{\pi/2}$, we have $\rho \leq \frac{2 \eta(b_1/\sigma_1)}{b_1/\sigma_1}$ and the upper bound $\overline{\rho}(b_1/\sigma_1)$ is increasing in $b_1/\sigma_1$. Since Lemma \ref{Lemma 1} implies $\eta(a)/a \leq \frac{\eta(0)+a}{a}$ and $\eta(a)/a$ is increasing on $(a^{\ast},\infty)$, $\sup_{a \in (a^{\ast},\infty)} \{ \eta(a) / a \} \leq \lim_{a \to \infty} \left\{ \frac{\eta(0)+a}{a} \right\} = 1$. Therefore, $\overline{\rho}(b_1/\sigma_1)$ is bounded above by $2$.
\end{proof}

\renewcommand{\theequation}{B.\arabic{equation}}
\setcounter{equation}{0}
%%%%%%%%%%%%%%%%%%%%%%%%%%%%%%%%%%%%%%%%%%%
\section{Technical derivations}\label{sec:derivation}

\subsection{Derivations in Section \ref{subsec:MSE}}\label{subsec:derivation_MSE}

In this section, we show that $Q_{\bm{\theta}}\left( \bm{w}_{\mathrm{MSE}} \right) < 0$ when $b_{\bm{\theta}}'(\bm{w}_{\mathrm{MSE}}) < 0$ and $s_{\bm{\theta}}'(\bm{w}_{\mathrm{MSE}}) > 0$. Let $t^*(a)$ be the maximizer of $t \cdot \Phi(-t+a)$. Then, by the proof of Lemma \ref{lemma:2}, we have
\begin{eqnarray}
Q_{\bm{\theta}}(\bm{w}) &=& \Phi\left( -t^*\left( \frac{b(\bm{w})}{s(\bm{w})} \right) + \frac{b(\bm{w})}{s(\bm{w})} \right) \times \left\{ s_{\bm{\theta}}'(\bm{w}) \left( t^*\left( \frac{b(\bm{w})}{s(\bm{w})} \right) - \frac{b(\bm{w})}{s(\bm{w})} \right) + b_{\bm{\theta}}'(\bm{w}) \right\}. \nonumber
\end{eqnarray}
Here, the sign of $Q_{\bm{\theta}}(\bm{w})$ is determined by
$$
s_{\bm{\theta}}'(\bm{w}) \left( t^*\left( \frac{b(\bm{w})}{s(\bm{w})} \right) - \frac{b(\bm{w})}{s(\bm{w})} \right) + b_{\bm{\theta}}'(\bm{w}),
$$
where $t^*(a) - a$ is decreasing in $a$ as shown in the proof of Lemma \ref{lemma:2}. Similarly, the sign of the directional derivative of the maximum MSE, $b^2(\bm{w})+s^2(\bm{w})$, is determined by
$$
s_{\bm{\theta}}'(\bm{w}) \left( \frac{b(\bm{w})}{s(\bm{w})} \right)^{-1} + b_{\bm{\theta}}'(\bm{w}).
$$
Suppose that $b_{\bm{\theta}}'(\bm{w}) < 0$ and $s_{\bm{\theta}}'(\bm{w}) > 0$, that is, we face the bias-variance tradeoff. Then, because numerical evaluation implies $t^*(a)-a < a^{-1}$ for $a \geq 0$, we obtain
$$
s_{\bm{\theta}}'(\bm{w}) \left( t^*\left( \frac{b(\bm{w})}{s(\bm{w})} \right) - \frac{b(\bm{w})}{s(\bm{w})} \right) + b_{\bm{\theta}}'(\bm{w}) \ < \ s_{\bm{\theta}}'(\bm{w}) \left( \frac{b(\bm{w})}{s(\bm{w})} \right)^{-1} + b_{\bm{\theta}}'(\bm{w}).
$$
When $\bm{w} = \bm{w}_{\mathrm{MSE}}$, the right-hand side must be zero. Hence, if $b_{\bm{\theta}}'(\bm{w}_{\mathrm{MSE}}) < 0$ and $s_{\bm{\theta}}'(\bm{w}_{\mathrm{MSE}}) > 0$, we conclude $Q_{\bm{\theta}}\left( \bm{w}_{\mathrm{MSE}} \right) < 0$.

\subsection{Derivations in Section \ref{subsec:K=2}}\label{subsec:derivation_K=2}

In this section, we derive the minimax regret and hierarchical Bayes rules for the simple setting described in Section \ref{subsec:K=2}. First, we derive the analytical expression of $b(w)$. When $0 \leq w \leq 1$, we obtain $ b(w) = \max_{\bm{\tau} \in \mathcal{T}, \tau_0 = 0} \left\{ w \tau_1 + (1-w) \tau_2 \right\} = C$. For $w > 1$, we obtain $b(w) =  \max_{\bm{\tau} \in \mathcal{T}, \tau_0 = 0} \left\{ w \tau_1 + (1-w) \tau_2 \right\} = C w$. Similarly, we have $b(w) = C (1-w)$ when $w < 0$. Hence, the maximum bias can be written as follows:
\begin{equation*}
b(w) \ = \ C \cdot \max\{ w, 1-w, 1 \}.
\end{equation*}

Next, we derive the minimax weight $w^{\ast} \in \text{arg} \min_{w} \left\{ s(w) \cdot \eta \left( b(w)/s(w) \right) \right\}$. From the proof of Lemma 2, let $t^*(a) \equiv \text{arg} \max_{t \geq 0} \left\{ t \cdot \Phi(-t+a) \right\}$ and obtain
\begin{eqnarray*}
\frac{\partial}{\partial s} \left\{ s \eta(C/s) \right\} &=& \eta\left( \frac{C}{s} \right) - \left( \frac{C}{s} \right) \eta'\left( \frac{C}{s} \right) \\
&=& \left\{ t^{\ast} \left( \frac{C}{s} \right) - \left( \frac{C}{s} \right) \right\} \cdot \Phi\left( - t^{\ast} \left( \frac{C}{s} \right) + \left( \frac{C}{s} \right)  \right),
\end{eqnarray*}
where $t^{\ast}(a) - a$ is a strictly decreasing function. 
From the proof of Theorem \ref{thm:global}, we have $t^{\ast}(a) - a = 0$ when $a = a^{\ast} \equiv \sqrt{\pi/2} \simeq 1.253$. 
Hence, $s \eta(C/s)$ is increasing when $s > C / a^{\ast} \simeq 0.798 C$ and decreasing when $s < C / a^{\ast}$. 
In addition, $\sqrt{\frac{\sigma_1^2 \sigma_2^2}{\sigma_1^2 + \sigma_2^2}} \leq s(w) \leq \max\{\sigma_1, \sigma_2\}$ holds with the inequalities hold with equalities at some $0 \leq w \leq 1$. Because $b(w) > C$ for $w < 0$ or $w > 1$ and $\eta(a)$ is a strictly increasing function, the minimax weight $w^{\ast}$ satisfies the following conditions:
\begin{eqnarray*}
\text{if} \ C / a^{\ast} < \sqrt{\frac{\sigma_1^2 \sigma_2^2}{\sigma_1^2 + \sigma_2^2}} & \Rightarrow & s(w^{\ast}) = \sqrt{\frac{\sigma_1^2 \sigma_2^2}{\sigma_1^2 + \sigma_2^2}}, \\
\text{if} \ \sqrt{\frac{\sigma_1^2 \sigma_2^2}{\sigma_1^2 + \sigma_2^2}} \leq C / a^{\ast} \leq \max\{\sigma_1, \sigma_2\} & \Rightarrow & s(w^{\ast}) = C / a^{\ast}, \\
\text{if} \ C / a^{\ast} >  \max\{\sigma_1, \sigma_2\} & \Rightarrow & w^{\ast} \leq 0 \ \text{or} \ w^{\ast} \geq 1.
\end{eqnarray*}
Therefore, if the dispersion of parameters $C$ is small compared to the standard deviations $\sigma_1$ and $\sigma_2$, the minimax weight attains the smallest variance, that is, $w^{\ast} = \frac{\sigma_2^2}{\sigma_1^2 + \sigma_2^2}$.

We consider the following hierarchical Bayes model:
\begin{eqnarray*}
& & \tau_k = \tau_0 + \epsilon_k, \ \ k = 1,2, \\
& & \tau_0 \sim N(0, \sigma_{\tau}^2), \ \ \epsilon_k \sim N(0, \sigma_{\epsilon}^2),
\end{eqnarray*}
where $\tau_0$, $\epsilon_1$, and $\epsilon_2$ are mutually independent. Then the posterior distribution of $\tau_0$ can be written as follows:
\begin{eqnarray*}
\pi(\tau_0 | \hat{\tau}_1, \hat{\tau}_2) %& \propto & \exp \left\{ - \frac{(\hat{\tau}_1 - \tau_0)^2}{2 (\sigma_1^2 + \sigma_{\epsilon}^2)} - \frac{(\hat{\tau}_2 - \tau_0)^2}{2 (\sigma_2^2 + \sigma_{\epsilon}^2)} - \frac{\tau_0^2}{2 \sigma_{\tau}^2}\right\} \\
& \propto & \exp \left[ - \frac{1}{2} \left\{ (\sigma_1^2 + \sigma_{\epsilon}^2)^{-1} + (\sigma_2^2 + \sigma_{\epsilon}^2)^{-1} + \sigma_{\tau}^{-2} \right\} \right. \\
& & \hspace{0.8in} \left. \times \left\{ \tau_0 - \frac{(\sigma_1^2 + \sigma_{\epsilon}^2)^{-1} \hat{\tau}_1 + (\sigma_2^2 + \sigma_{\epsilon}^2)^{-1} \hat{\tau}_2}{(\sigma_1^2 + \sigma_{\epsilon}^2)^{-1} + (\sigma_2^2 + \sigma_{\epsilon}^2)^{-1} + \sigma_{\tau}^{-2}} \right\} \right].
\end{eqnarray*}
As discussed in Section \ref{subsec:minimax}, the Bayes optimal rule becomes
\begin{eqnarray*}
\bm{1}\{ E_{\pi}( \tau_0 | \mathbf{D}) \geq 0 \} &=& \bm{1} \left\{ (\sigma_1^2 + \sigma_{\epsilon}^2)^{-1} \hat{\tau}_1 + (\sigma_2^2 + \sigma_{\epsilon}^2)^{-1} \hat{\tau}_2 \geq 0 \right\} \\
&=& \bm{1} \left\{ (\sigma_2^2 + \sigma_{\epsilon}^2) \hat{\tau}_1 + (\sigma_1^2 + \sigma_{\epsilon}^2) \hat{\tau}_2 \geq 0 \right\} \\
&=& \bm{1} \left\{ \left( \frac{\sigma_2^2 + \sigma_{\epsilon}^2}{\sigma_1^2 + \sigma_2^2 + 2 \sigma_{\epsilon}^2} \right) \hat{\tau}_1 + \left( \frac{\sigma_1^2 + \sigma_{\epsilon}^2}{\sigma_1^2 + \sigma_2^2 + 2 \sigma_{\epsilon}^2} \right)  \hat{\tau}_2 \geq 0 \right\},
\end{eqnarray*}
which implies that $w_{\text{HB}} = \frac{ \sigma_2^2 + \sigma_{\epsilon}^2}{\sigma_1^2 + \sigma_2^2 + 2 \sigma_{\epsilon}^2}$. We observe that
\begin{eqnarray*}
w_{\text{HB}} &=& \frac{ \sigma_2^2 + \sigma_{\epsilon}^2}{\sigma_1^2 + \sigma_2^2 + 2 \sigma_{\epsilon}^2} \\ 
&=& \frac{ \sigma_2^2}{\sigma_1^2 + \sigma_2^2} \cdot \frac{\sigma_1^2 + \sigma_2^2}{\sigma_1^2 + \sigma_2^2+2\sigma_{\epsilon}^2} + \frac{1}{2} \cdot \frac{2\sigma_{\epsilon}^2}{\sigma_1^2 + \sigma_2^2+2\sigma_{\epsilon}^2}.
\end{eqnarray*}
Hence, $w_{\text{HB}}$ shrinks $\frac{ \sigma_2^2}{\sigma_1^2 + \sigma_2^2}$ toward $1/2$ and the degree of shrinkage is increasing in $\sigma_{\epsilon}^2$.

\renewcommand{\theequation}{C.\arabic{equation}}
\setcounter{equation}{0}
%%%%%%%%%%%%%%%%%%%%%%%%%%%%%%%%%%%%%%%%%%%

\bibliographystyle{ecta}
\bibliography{meta-analysis}

\pagebreak

\section*{Online Appendix}
\section{Additional extensions and discussion}\label{sec:additional_extensions}

\subsection{The minimax regret rule under a bounded parameter space}\label{subsec:bounded}

We consider the following bounded parameter space:
\begin{equation}
\mathcal{T}_b \ \equiv \ \mathcal{T} \cap \left[ -\underline{t}, \overline{t} \right]^{K+1}, \label{bounded_PS}\nonumber
\end{equation}
where $\underline{t}, \overline{t} \in (0,+\infty]$ and $\mathcal{T} = \{ \bm{\tau} : |\tau_k - \tau_l| \leq C \|x_k -x_l\| \ \text{for all $k,l$} \}$. For example, if the policy effects are known to be nonnegative and the cost is $c_0$, the cost-adjusted policy effects can be bounded below by $-c_0$. Then we can set $\underline{t}=c_0$ and $\overline{t}=+\infty$.

Because the parameter space is not symmetric, we consider the following treatment rule:
\[
\hat{\delta}_{v, \bm{w}}(\mathbf{D}) \ \equiv \ \mathbf{1} \left\{ v + \sum_{k=1}^K w_k \hat{\tau}_k \geq 0 \right\}.
\]
Instead of introducing an intercept term $v$, we impose that all elements of $\bm{w}$ are nonnegative. In addition, we assume that $\underline{t}$ and $\overline{t}$ satisfy
\begin{equation}
\min \left\{ \underline{t}, \overline{t} \right\} \ \geq \ C \cdot \max_{1 \leq k \leq K} \|x_k - x_0\|. \label{condition1_bounded_PS}
\end{equation}
This condition implies that $\underline{t}$ and $\overline{t}$ are sufficiently large relative to the Lipschitz constraint. Hence, we consider the following maximization problem:
\[
\max_{\bm{\tau} \in \mathcal{T}_b} R(\bm{\tau}, \hat{\delta}_{v, \bm{w}}),
\]
where $\bm{w}$ satisfies $\sum_{k=1}^K w_k = 1$ and $w_k \geq 0$ for all $k$.

For $0 \leq t \leq \overline{t}$, we have
\begin{eqnarray*}
\max_{\bm{\tau} \in \mathcal{T}_b, \tau_0 = t} R(\bm{\tau}, \hat{\delta}_{v, \bm{w}}) &=& \max_{\bm{\tau} \in \mathcal{T}_b, \tau_0 = t} \left\{ t \cdot \Phi \left( -  \frac{t+v+\sum_{k=1}^K w_k (\tau_k-t)}{s(\bm{w})}\right) \right\} \\
&=& t \cdot \Phi \left( - \frac{t}{s(\bm{w})} - \frac{v + \min_{\bm{\tau} \in \mathcal{T}_b, \tau_0 = t}\left\{ \sum_{k=1}^K w_k (\tau_k-t) \right\}}{s(\bm{w})} \right).
\end{eqnarray*}
Because $w_k$ is nonnegative, we have
\begin{eqnarray*}
\min_{\bm{\tau} \in \mathcal{T}_b, \tau_0 = t}\left\{ \sum_{k=1}^K w_k (\tau_k-t) \right\} &=& \sum_{k=1}^K w_k \left\{ \max \{t-C\|x_k-x_0\|,-\underline{t}\} - t  \right\} \\
&=& -C \sum_{k=1}^K w_k \|x_k-x_0\| \ = \ - b(\bm{w}),
\end{eqnarray*}
where the second equality follows from $t \geq 0$ and $\underline{t} \geq C \cdot \max_{1 \leq k \leq K} \|x_k - x_0\|$. Hence, $0 \leq t \leq \overline{t}$, we obtain
\[
\max_{\bm{\tau} \in \mathcal{T}_b, \tau_0 = t} R(\bm{\tau}, \hat{\delta}_{v, \bm{w}}) \ = \ t \cdot \Phi \left( - \frac{t}{s(\bm{w})} + \frac{-v + b(\bm{w})}{s(\bm{w})} \right).
\]
Similarly, for $-\underline{t} \leq t \leq 0$, we have
\[
\max_{\bm{\tau} \in \mathcal{T}_b, \tau_0 = t} R(\bm{\tau}, \hat{\delta}_{v,\bm{w}}) \ = \ -t \cdot \Phi \left( \frac{t}{s(\bm{w})} + \frac{v+b(\bm{w})}{s(\bm{w})} \right).
\]
Therefore, we obtain
\begin{eqnarray}
\max_{\bm{\tau} \in \mathcal{T}_b} R(\bm{\tau}, \hat{\delta}_{v,\bm{w}}) &=& \max_{t \in [ -\underline{t}, \overline{t} ]} \left\{ |t| \cdot \Phi \left( - \frac{|t|}{s(\bm{w})} + \frac{- sgn(t) \cdot v + b(\bm{w})}{s(\bm{w})} \right) \right\} \nonumber \\
&=& s(\bm{w}) \cdot \max \left[ \max_{0 \leq s \leq \overline{t}/s(\bm{w})} \left\{ s \cdot \Phi \left( -s + \frac{-v+b(\bm{w})}{s(\bm{w})} \right) \right\} , \right. \nonumber \\
& & \hspace{1.5in} \left. \max_{0 \leq s \leq \underline{t}/s(\bm{w})} \left\{ s \cdot \Phi \left( -s + \frac{v+b(\bm{w})}{s(\bm{w})} \right) \right\} \right] \nonumber \\
&=& s(\bm{w}) \cdot \max \left\{ \max_{0 \leq s \leq \overline{t}/s(\bm{w})} \psi \left( s, \frac{-v+b(\bm{w})}{s(\bm{w})} \right), \right. \nonumber \\
& & \hspace{2.2in} \left. \max_{0 \leq s \leq \underline{t}/s(\bm{w})} \psi \left( s, \frac{v+b(\bm{w})}{s(\bm{w})} \right) \right\}, \nonumber
\end{eqnarray}
where $\psi(s,a) \equiv s \cdot \Phi(-s+a)$. Note that $\psi(s,a)$ is unimodal in $s$ and maximized at $s=t^{\ast}(a)$. As a result, we obtain
\begin{eqnarray*}
\max_{0 \leq s \leq \overline{t}/s(\bm{w})} \psi \left( s, \frac{-v+b(\bm{w})}{s(\bm{w})} \right) &=& \begin{cases}
\eta \left( \frac{-v+b(\bm{w})}{s(\bm{w})} \right), & \text{if $t^{\ast} \left( \frac{-v+b(\bm{w})}{s(\bm{w})} \right) \leq \frac{\overline{t}}{s(\bm{w})}$} \\
\psi \left( \frac{\overline{t}}{s(\bm{w})}, \frac{-v+b(\bm{w})}{s(\bm{w})} \right), & \text{if $t^{\ast} \left( \frac{-v+b(\bm{w})}{s(\bm{w})} \right) > \frac{\overline{t}}{s(\bm{w})}$}
\end{cases}, \\
\max_{0 \leq s \leq \underline{t}/s(\bm{w})} \psi \left( s, \frac{v+b(\bm{w})}{s(\bm{w})} \right) &=& \begin{cases}
\eta \left( \frac{v+b(\bm{w})}{s(\bm{w})} \right), & \text{if $t^{\ast} \left( \frac{v+b(\bm{w})}{s(\bm{w})} \right) \leq \frac{\underline{t}}{s(\bm{w})}$} \\
\psi \left( \frac{\underline{t}}{s(\bm{w})}, \frac{v+b(\bm{w})}{s(\bm{w})} \right), & \text{if $t^{\ast} \left( \frac{v+b(\bm{w})}{s(\bm{w})} \right) > \frac{\underline{t}}{s(\bm{w})}$}
\end{cases}.
\end{eqnarray*}

We can show that calculating the maximum regret becomes even simpler when $v=0$ and $\underline{t}$ and $\overline{t}$ are sufficiently large relative to the standard errors. We assume that
\begin{equation}
\min \left\{ \underline{t}, \overline{t} \right\} \ \geq \ \sqrt{\pi /2} \cdot \max_{1 \leq k \leq K} \sigma_k. \label{condition2_bounded_PS}
\end{equation}
It follows from (\ref{condition1_bounded_PS}) that we have $b(\bm{w}) \leq \min \left\{ \underline{t}, \overline{t} \right\}$. In the proof of Theorem \ref{thm:global}, $a-t^{\ast}(a)$ is a strictly increasing function and $a-t^{\ast}(a) \geq 0$ for $a \geq \sqrt{\pi / 2}$. Hence, we obtain
\begin{eqnarray*}
\frac{\overline{t}}{s(\bm{w})} - t^{\ast} \left( \frac{b(\bm{w})}{s(\bm{w})} \right) & \geq & \frac{\overline{t}}{s(\bm{w})} - t^{\ast} \left( \frac{\overline{t}}{s(\bm{w})} \right) \ \geq \ 0,
\end{eqnarray*}
where the second inequality follows from (\ref{condition2_bounded_PS}) and $s(\bm{w}) \leq \max_{1 \leq k \leq K} \sigma_k$. Similarly, we have $t^{\ast} \left( \frac{b(\bm{w})}{s(\bm{w})} \right) \leq \frac{\underline{t}}{s(\bm{w})}$. Therefore, if $w_k \geq 0$ for all $k$ and $\underline{t}$ and $\overline{t}$ satisfies (\ref{condition1_bounded_PS}) and (\ref{condition2_bounded_PS}), then we obtain
\[
\max_{\bm{\tau} \in \mathcal{T}_b} R(\bm{\tau}, \hat{\delta}_{\bm{w}}) \ = \ s(\bm{w}) \cdot \eta \left( \frac{b(\bm{w})}{s(\bm{w})} \right).
\]
That is, if the weights are restricted to be nonnegative and $\underline{t}$ and $\overline{t}$ are sufficiently large relative to both the Lipschitz constraint and the standard errors, then the maximum regret rule under the bounded parameter space $\mathcal{T}_b$ coincides with that under the unbounded parameter space $\mathcal{T}$.

\subsection{The minimax regret rule under unkown $\sigma_k$}\label{subsec:unknown}

In this section, we consider the case in which $\sigma_k$ is unknown. Suppose that instead of $\sigma_k$, an estimator $\hat{\sigma}_k$ is available. We assume that $\hat{\sigma}_k$ is independent of $\hat{\tau}_k$ and the distribution of $\hat{\sigma}_k$ does not depends on $\bm{\tau}$. This assumption can be justified when the data are normally distributed. By replacing $\sigma_k$ with $\hat{\sigma}_k$, we obtain the following weights:
\[
\hat{\bm{w}}_{\mathrm{minimax}} \ \in \ \text{arg} \min_{\bm{w}} \left\{ \hat{s}(\bm{w}) \cdot \eta\left( \frac{b(\bm{w})}{\hat{s}(\bm{w})} \right) \right\},
\]
where $\hat{s}(\bm{w}) = \sqrt{\sum_{k=1}^K w_k^2 \hat{\sigma}_k^2}$. We compare the maximum regrets of $\hat{\delta}_{\bm{w}_{\mathrm{minimax}}}$ and $\hat{\delta}_{\hat{\bm{w}}_{\mathrm{minimax}}}$.

To compare the maximum regrets of $\hat{\delta}_{\bm{w}_{\mathrm{minimax}}}$ and $\hat{\delta}_{\hat{\bm{w}}_{\mathrm{minimax}}}$, we evaluate the ratio of the maximum regrets
\[
\frac{\max_{\bm{\tau} \in \mathcal{T}} R(\bm{\tau}, \hat{\delta}_{\hat{\bm{w}}_{\mathrm{minimax}}})}{\max_{\bm{\tau} \in \mathcal{T}} R(\bm{\tau}, \hat{\delta}_{\bm{w}_{\mathrm{minimax}}})}.
\]
Since $\hat{\sigma}_k$ is independent of $\hat{\tau}_k$ and $\hat{\bm{w}}_{\mathrm{minimax}}$ does not depend on $\hat{\tau}_1, \ldots, \hat{\tau}_K$, we obtain
\begin{eqnarray*}
\max_{\bm{\tau} \in \mathcal{T}} R(\bm{\tau}, \hat{\delta}_{\hat{\bm{w}}_{\mathrm{minimax}}}) &=& \max_{\bm{\tau} \in \mathcal{T}} E \left[ E_{\bm{\tau}}\left[ \left. W(\delta^{\ast}) - W(\hat{\delta}_{\hat{\bm{w}}_{\mathrm{minimax}}}(\mathbf{D})) \right|  \hat{\sigma}_1, \ldots , \hat{\sigma}_K \right] \right] \\
&\leq &  E \left[ \max_{\bm{\tau} \in \mathcal{T}} E_{\bm{\tau}}\left[ \left. W(\delta^{\ast}) - W(\hat{\delta}_{\hat{\bm{w}}_{\mathrm{minimax}}}(\mathbf{D})) \right|  \hat{\sigma}_1, \ldots , \hat{\sigma}_K \right] \right] \\
&=& E\left[ s(\hat{\bm{w}}_{\mathrm{minimax}}) \cdot \eta\left( \frac{b(\hat{\bm{w}}_{\mathrm{minimax}})}{s(\hat{\bm{w}}_{\mathrm{minimax}})} \right) \right],
\end{eqnarray*}
where the last expectation is taken with respect to $\hat{\sigma}_1, \ldots , \hat{\sigma}_K$. Hence, we obtain
\begin{eqnarray}
\frac{\max_{\bm{\tau} \in \mathcal{T}} R(\bm{\tau}, \hat{\delta}_{\hat{\bm{w}}_{\mathrm{minimax}}})}{\max_{\bm{\tau} \in \mathcal{T}} R(\bm{\tau}, \hat{\delta}_{\bm{w}_{\mathrm{minimax}}})} &\leq & E \left[ \frac{s(\hat{\bm{w}}_{\mathrm{minimax}}) \cdot \eta\left( \frac{b(\hat{\bm{w}}_{\mathrm{minimax}})}{s(\hat{\bm{w}}_{\mathrm{minimax}})} \right)}{s(\bm{w}_{\mathrm{minimax}}) \cdot \eta\left( \frac{b(\bm{w}_{\mathrm{minimax}})}{s(\bm{w}_{\mathrm{minimax}})} \right)} \right] \label{ratio_unknown_sigma}.
\end{eqnarray}

To evaluate the right-hand side of (\ref{ratio_unknown_sigma}), we show that
\begin{equation}
\gamma_0^{-1} \leq \hat{s}/s \leq \gamma_0  \ \ \Rightarrow \ \ \max \left\{ \frac{ \hat{s} \cdot \eta\left( b/\hat{s} \right)}{s \cdot \eta\left( b/s \right)}, \, \frac{s \cdot \eta\left( b/s \right)}{\hat{s} \cdot \eta\left( b/\hat{s} \right)} \right\} \ \leq \ 2\gamma_0-1, \label{bound_unknown_sigma}
\end{equation}
where $\gamma_0 > 1$. If $1 \leq \hat{s}/s \leq \gamma_0$, we have
\begin{eqnarray*}
\frac{ \hat{s} \cdot \eta\left( b/\hat{s} \right)}{s \cdot \eta\left( b/s \right)} & \leq & \left( \frac{\hat{s}}{s} \right) \cdot \frac{\eta(b/s)}{\eta(b/s)} \ \leq \ \gamma_0 \ \leq \ 2 \gamma_0 -1.
\end{eqnarray*}
Furthermore, if $1 \leq s/\hat{s} \leq \gamma_0$, we have
\begin{eqnarray*}
\frac{ \hat{s} \cdot \eta\left( b/\hat{s} \right)}{s \cdot \eta\left( b/s \right)} & \leq & \left( \frac{\hat{s}}{s} \right) \times \left( \frac{\eta(b/s) + (b/\hat{s}) - (b/s)}{\eta(b/s)} \right) \\
& \leq & \left( \frac{\hat{s}}{s} \right) \times \left( 1 + \frac{(\gamma_0-1) \cdot (b/s)}{\eta(b/s)} \right) \ \leq \ 1 + (\gamma_0-1)\cdot H(b/s),
\end{eqnarray*}
where $H(a) \equiv a/\eta(a)$ and the first inequality follows because $a \geq a' \Rightarrow \eta(a) \leq \eta(a') + (a-a')$. It follows from Lemma \ref{lemma:1} that $H(a) = a / \eta(a) \leq a / (a/2) = 2$. Therefore, we obtain (\ref{bound_unknown_sigma}).

Letting $\hat{\gamma}_0 \equiv \max_{k \in \{1, \ldots, K\}} \left[ \max \{ \hat{\sigma}_k / \sigma_k , \sigma_k / \hat{\sigma}_k \} \right]$, we have $\hat{\gamma}_0^{-1} \leq \frac{\hat{s}(\bm{w})}{s(\bm{w})} \leq \hat{\gamma}_0$ for any $\bm{w}$. From (\ref{ratio_unknown_sigma}) and (\ref{bound_unknown_sigma}), we obtain
\begin{eqnarray*}
\frac{\max_{\bm{\tau} \in \mathcal{T}} R(\bm{\tau}, \hat{\delta}_{\hat{\bm{w}}_{\mathrm{minimax}}})}{\max_{\bm{\tau} \in \mathcal{T}} R(\bm{\tau}, \hat{\delta}_{\bm{w}_{\mathrm{minimax}}})} &\leq & E \left[ (2\hat{\gamma}_0-1)^2 \times \frac{\hat{s}(\hat{\bm{w}}_{\mathrm{minimax}}) \cdot \eta\left( \frac{b(\hat{\bm{w}}_{\mathrm{minimax}})}{\hat{s}(\hat{\bm{w}}_{\mathrm{minimax}})} \right)}{\hat{s}(\bm{w}_{\mathrm{minimax}}) \cdot \eta\left( \frac{b(\bm{w}_{\mathrm{minimax}})}{\hat{s}(\bm{w}_{\mathrm{minimax}})} \right)} \right] \\
&\leq & E \left[ (2\hat{\gamma}_0-1)^2 \right],
\end{eqnarray*}
where the last inequality holds since the weight vector $\hat{\bm{w}}_{\mathrm{minimax}}$ minimizes $\hat{s}(\bm{w}) \cdot \eta\left( \frac{b(\bm{w})}{\hat{s}(\bm{w})} \right)$. Therefore, if $\hat{\sigma}_k / \sigma_k$ is close to 1 with high probability, the maximum regret of $\hat{\delta}_{\hat{\bm{w}}_{\mathrm{minimax}}}$ is not much larger than that of $\hat{\delta}_{\bm{w}_{\mathrm{minimax}}}$.

\subsection{Consistency of the minimax treatment rule}\label{subsec:consistency}

If we have perfect knowledge of $(\tau_1, \dots, \tau_K)$, i.e., $\sigma_k = 0$ for all $1 \leq k \leq K$, we can obtain the (true) identified set of $\tau_0$ based on the constraints on the parameter space $\mathcal{T}$. For instance, we construct the identified set of $\tau_0$ by intersecting multiple bounds for $\tau_0$, each of which is constructed by extrapolating from $\tau_k$, as considered in \citet{manski2020towards}. We can then consider finding the minimax regret treatment rule given the true identified set of $\tau_0$ without any sampling uncertainty. We denote by $\delta^*_{IS}$ such a (non-randomized) minimax regret rule. As the sample size of each study increases, that is, as $\sigma_k \rightarrow 0$, should we expect the minimax regret rule $\hat{\delta}_{\bm{w}_{\text{minimax}}}$ we constructed in the previous section to converge to $\delta^*_{IS}$? In what follows, we compare $\delta^*_{IS}$ with the limiting version of $\hat{\delta}_{\bm{w}_{\text{minimax}}}$, and show that $\hat{\delta}_{\bm{w}_{\text{minimax}}}$ does \textit{not} necessarily converge to $\delta^*_{IS}$ as $\sigma_k \rightarrow 0$.

First, we consider the minimax regret rule when $\sigma_k = 0$ for all $k = 1, \dots , K$. Then, $\tau_k = \hat{\tau}_k$ for all $k = 1, \dots , K$ and the identified set of $\tau_0$ is
\begin{equation}
IS_0 \ \equiv \ \left\{ \tau_0 : \text{$\bm{\tau} \in \mathcal{T}$ and $\tau_k = \hat{\tau}_k$ for all $k = 1, \cdots , K$} \right\}. \label{identified_set}
\end{equation}
In this case, the parameter space $\mathcal{T}$ projected for $\tau_0$ yields the identified set $IS_0$. Because there is no randomness in this problem, for a treatment rule $\delta$, the welfare regret of $\delta$ becomes
$$
W(\delta^*)-W(\delta) \ = \ \begin{cases}
    -\tau_0 \cdot \delta & \text{if $\tau_0 < 0$} \\
    \tau_0 \cdot (1-\delta) & \text{if $\tau_0 \geq 0$} 
  \end{cases}.
$$
Hence, the minimax regret rule over $IS_0$ can be written as
\begin{equation}
\delta_{IS}^*(\mathbf{D}) \ = \ \begin{cases}
    1 & \ \ \ \ \ \text{if $-\underline{\tau}_0 \leq \overline{\tau}_0$} \\
    0 & \ \ \ \ \ \text{if $-\underline{\tau}_0 > \overline{\tau}_0$}
  \end{cases}, \label{minimax_IS}
\end{equation}
where $\delta_{IS}^*(\mathbf{D}) \in \mathrm{arg} \min_{\delta \in \{0,1\}} \max_{\tau_0 \in IS_0} \left\{ W(\delta^{\ast}) - W(\delta)  \right\}$, $\underline{\tau}_0 \equiv \inf\{\tau_0:\tau_0 \in IS_0\}$, and $\overline{\tau}_0 \equiv \sup\{\tau_0:\tau_0 \in IS_0\}$ are the smallest and largest values of the identified set of $\tau_0$, respectively. The rule $\delta_{IS}^*$ becomes 1 (or 0) when we have $\underline{\tau}_0 > 0$ (or $\overline{\tau}_0 < 0$), that is, all values of the identified set of $\tau_0$ are positive (or negative). When the identified set of $\tau_0$ contains both of positive and negative values, $\delta_{IS}^*$ becomes 1 (or 0) if the absolute value of $\overline{\tau}_0$ is larger (or smaller) than that of $\underline{\tau}_0$.\footnote{In this paper, we consider only non-randomized treatment rules. \cite{manski2011choosing} shows that the minimax regret criterion always yields a randomized treatment rule when $\underline{\tau}_0 < 0 < \overline{\tau}_0$. He shows that the minimax randomized treatment rule randomly assigns a fraction $|\overline{\tau}_0|/(|\underline{\tau}_0|+|\overline{\tau}_0|)$ of the population to treatment 1 and the remaining $|\underline{\tau}_0|/(|\underline{\tau}_0|+|\overline{\tau}_0|)$ to treatment 0.}

\bigskip

Next, we consider the large sample properties of our minimax regret rule $\hat{\delta}_{\bm{w}_{\text{minimax}}}$, i.e., $\sigma_k \rightarrow 0$. In this case, we can show that the minimax regret criterion yields the treatment rule that minimizes the maximum bias.
We focus on the linear aggregation rule that minimizes the limit of the maximum regret
\[  \lim_{\sigma_k \to 0 \, \forall k} \max_{\bm{\tau} \in \mathcal{T}} R(\bm{\tau}, \hat{\delta}_{\bm{w}}).  \]
From Lemma \ref{lemma:1}, we observe that
\begin{eqnarray*}
\lim_{\sigma_k \to 0 \, \forall k} \max_{\bm{\tau} \in \mathcal{T}} R(\bm{\tau}, \hat{\delta}_{\bm{w}}) &=& \lim_{s \to 0} s \cdot \eta \left( b(\bm{w}) / s \right) \\
& \leq & \lim_{s \to 0} s \cdot \left\{ \eta(0) + \frac{b(\bm{w})}{s} \right\} \ = \  b(\bm{w}).
\end{eqnarray*}
Furthermore, since $\eta(a)$ is strictly increasing and convex from the proof of Lemma \ref{lemma:2}, for any $a > 0$ we have
\begin{eqnarray*}
\lim_{s \to 0} s \cdot \eta \left( b(\bm{w}) / s \right) & \geq & \lim_{s \to 0} s \cdot \left\{ \eta(a) + \eta'(a) \left( \frac{b(\bm{w})}{s} - a \right)  \right\} \\
&=& \eta'(a) \cdot b(\bm{w}) \ = \ \Phi (-t^{\ast}(a)+a) \cdot b(\bm{w}).
\end{eqnarray*}
We now show that $-t^{\ast}(a)+a \to \infty$ as $a \to \infty$. Suppose, to the contrary, that $-t^{\ast}(a)+a$ converges to a finite constant. Then, since $t^{\ast}(a)$ must diverge to infinity as $a \to \infty$, we obtain $t^{\ast}(a) \cdot \phi (-t^{\ast}(a)+a) \to \infty$, which contradicts the identity
$$
t^{\ast}(a) \cdot \phi (-t^{\ast}(a)+a) \ = \ \Phi (-t^{\ast}(a)+a) \ \leq \ 1.
$$
Hence we obtain $-t^{\ast}(a)+a \to \infty$ and $\Phi (-t^{\ast}(a)+a) \to 1$. Therefore, we obtain
\[
\lim_{\sigma_k \to 0 \, \forall k} \max_{\bm{\tau} \in \mathcal{T}} R(\bm{\tau}, \hat{\delta}_{\bm{w}}) \ = \ b(\bm{w}).
\]
As a result, in large samples, the limiting version of $\hat{\delta}_{\bm{w}_{\text{minimax}}}$ becomes a treatment rule that minimizes the maximum bias $b(\bm{w})$.

%The proof of Lemma \ref{lemma:2} shows that $\eta(a)$ is strictly increasing and convex with its slope bounded from above by one. Hence, the slope of $\eta(a)$ converges to a positive constant $c \in (0,1]$ as $a \rightarrow +\infty$. This implies that when $a$ is large, $\eta(a)$ can be approximated by $d+c \cdot a$ for some $d$. As $\sigma_k \rightarrow 0$ for all $k$, we have $s(\bm{w}) \rightarrow 0$ for any $\bm{w}$. From Theorem \ref{thm:minimax}, as $s(\bm{w}) \rightarrow 0$, we can approximate the maximum regret of $\hat{\delta}_{\bm{w}}$ by $c \cdot b(\bm{w})$. This implies that in large samples, the minimax regret rule becomes a treatment rule that minimizes the maximum bias $b(\bm{w})$.

We examine whether the following consistency property holds:
\[
\text{arg} \min_{\hat{\delta}_{\bm{w}}} \lim_{\sigma_k \to 0 \, \forall k}  \max_{\tau \in \mathcal{T}} R(\tau, \hat{\delta}_{\bm{w}}) \ = \ \text{arg} \min_{\hat{\delta}}  \max_{\tau_0 \in IS_0} \left\{ W(\delta^{\ast}) - W(\hat{\delta}(\mathbf{D}))  \right\},
\]
where $W(\delta^{\ast}) - W(\hat{\delta}(\mathbf{D})) = \lim_{\sigma_k \to 0 \, \forall k} R(\tau, \hat{\delta})$ holds when $\hat{\delta}$ is continuous at $\mathbf{D}$. To be specific, consider the case in which $K=2$, $\mathcal{T} = \{ \bm{\tau} : |\tau_k - \tau_l| \leq |x_k-x_l| \ \text{for all $k,l$} \}$, and $(x_0,x_1,x_2) = (0,1,-2)$. Then, the maximum bias $b(\bm{w})$ is minimized at $\bm{w}=(1,0)$. Hence, the limiting version of $\hat{\delta}_{\bm{w}_{\text{minimax}}}$ becomes
\[ 1\{ \hat{\tau}_1 \geq 0 \}. \]
On the contrary, we can show that this treatment rule is not equivalent to $\delta^*_{IS}$. In this case, the identified set of $\tau_0$ can be written as
\[ IS_0 \ = \ [\hat{\tau}_1-1,\hat{\tau}_1+1] \cap  [\hat{\tau}_2-2,\hat{\tau}_2+2].  \]
For example, when $\hat{\tau}_1 = -0.9$ and $\hat{\tau}_2 = 2$, we obtain $|\hat{\tau}_1 -\hat{\tau}_2| \leq |x_1-x_2|$ and $\left( \underline{\tau}_0,\overline{\tau}_0 \right) = (0,0.1)$. This implies that $\delta^*_{IS}(\mathbf{D}) = 1$ holds even when $\hat{\tau}_1 < 0$. Therefore, in this case, the limiting version of $\hat{\delta}_{\bm{w}_{\text{minimax}}}$ does not agree with $\delta^*_{IS}$.

In contrast to the partially identified setting, we can show that the limiting version of $\hat{\delta}_{\bm{w}_{\text{minimax}}}$ coincides with $\delta^*_{IS}$ when there exists an unbiased estimator of $\tau_0$. We assume that there exists $\bm{w}^{\ast}$ such that $b(\bm{w}^{\ast}) = 0$. Then, from the above arguments, the limiting version of $\hat{\delta}_{\bm{w}_{\text{minimax}}}$ becomes $\hat{\delta}_{\bm{w}^{\ast}}$ since $\bm{w}^{\ast}$ minimizes the maximum bias. Furthermore, when $\sigma_k \to 0$ for all $k$, it follows from $b(\bm{w}^{\ast}) = 0$ that $IS_0$ must be a singleton $\{ \tau_0 \}$ and $\tau_0 = \sum_{k=1}^K w_k^{\ast} \hat{\tau}_k$ holds. Hence, we obtain $\delta_{IS}^{\ast}(\mathbf{D}) = 1\{ \sum_{k=1}^K w_k^{\ast} \hat{\tau}_k \geq 0 \}$. Therefore, if we consider the point-identified setting such as $\mathcal{T}_{\mathrm{meta}}$, then the minimax regret rule $\hat{\delta}_{\bm{w}_{\text{minimax}}}$ converges to $\delta^{\ast}_{IS}$ as $\sigma_k \to 0$ for all $k$.

%% Remark 5 %%
\begin{Remark} \label{rem:IS_randomized}
We can show that some minimax regret rule does not converges to $\delta^{\ast}_{IS}$ even if we allow randomized rules. For example, consider $\mathcal{T} = \{ (\tau_0,\tau_1) \in \mathbb{R}^2 : |\tau_0 - \tau_1| \leq 1 \}$. Then, the maximum welfare regret with the sample analogue of the identified set plugged in is
\begin{eqnarray*}
\max_{\tau_0 \in IS_0} r(\tau_0,\delta) \ = \ \max\{ -(\hat{\tau}_1-1)\cdot \delta, (\hat{\tau}_1+1)\cdot (1-\delta) \},
\end{eqnarray*}
where $IS_0 = [\hat{\tau}_1 - 1, \hat{\tau}_1 + 1]$. Hence, the maximum welfare regret is minimized by the following treatment rule:
\[
\delta^{\ast}_{IS}(\hat{\tau}_1) \ = \ \begin{cases}
    1 & \text{if $\hat{\tau}_1 > 1$} \\
    \frac{\hat{\tau}_1 + 1}{2} & \text{if $-1 \leq \hat{\tau}_1 \leq 1$} \\
    0 & \text{if $\hat{\tau}_1 < -1$}
\end{cases}.
\]
On the other hand, \cite{stoye2012minimax} shows that when $\mathcal{T} = \{ (\tau_0,\tau_1) \in \mathbb{R}^2 : |\tau_0 - \tau_1| \leq 1 \}$, the minimax regret treatment rule is given by
\[
\hat{\delta}(\hat{\tau}_1) \ = \ \begin{cases}
    \bm{1}\{\hat{\tau}_1 \geq 0\} & \text{if $\sigma_1 \geq 2 \phi(0)$} \\
    \Phi\left( \hat{\tau}_1 / \sqrt{4 \phi(0)^2 - \sigma_1^2} \right) & \text{if $\sigma_1 < 2 \phi(0)$}
\end{cases}.
\]
\cite{yata2021optimal} also show the same result in more general settings. Therefore, as $\sigma_1 \to 0$, the minimax regret treatment rule converges to $\Phi\left( \hat{\tau}_1 / \sqrt{4 \phi(0)^2} \right)$. These results imply that the minimax regret treatment rule does not converge to $\delta^{\ast}_{IS}$ in this setting.\footnote{\label{footnote:consistency} \cite{montiel2023decision} shows that the set of minimax regret optimal rules contains a piecewise linear rule and it converges to $\delta^{\ast}_{IS}$ as the variance of the signal approaches zero. Hence, there is a sequence of minimax regret treatment rules that converges to $\delta^{\ast}_{IS}$, but there is also a sequence of minimax regret treatment rules that does not converge to $\delta^{\ast}_{IS}$.}
\end{Remark}

\subsection{The choice of the Lipschitz constant $C$}\label{subsec:choice_C}

In this section, we describe the two methods used in Section \ref{sec:empirical} to select $C$; leave-one-out cross-validation and marginal likelihood maximization. We begin with leave-one-out cross-validation. For each $k$, we treat study $k$ as the target population and let $d_{k, C} \in \{0,1\}$ denote the decision based on the data excluding study $k$ under the parameter space $\mathcal{T}_C$. Letting $C_1, \ldots, C_M$ be candidate values for the Lipschitz constant $C$, leave-one-out cross-validation selects $C_m$ that maximizes
\[
\sum_{k=1}^K \left\{ \hat{\tau}_k \cdot d_{k,C_m} + \hat{\mu}_{k}\right\},
\]
where $\hat{\mu}_{k}$ denotes the average outcome in study $k$ in the absence of the policy. Since $\hat{\mu}_k$ does not depend on $C_m$, this criterion is equivalent to maximizing
$$
\sum_{k=1}^K \hat{\tau}_k \cdot d_{k,C_m},
$$
and therefore can be computed without using $\hat{\mu}_k$. In Section \ref{sec:empirical}, we evaluate this criterion for $C = 0.005, \, 0.010, \ldots, 0.040, \, 0.045$. Among these candidate values, the smallest $C$ that maximizes the criterion is $0.025$. Moreover, when $C=0.025$, the decision correctly matches the sign of $\hat{\tau}_k$ in 11 out of the 14 studies; that is, $d_{k,C} = 1\{\hat{\tau}_k \geq 0\}$ holds for 11 out of the 14 studies.

Next, we describe marginal likelihood maximization. Let $\pi_{a,C}$ denote a prior distribution for $\bm{\tau}$ with support $\mathcal{T}_C$, where $a$ is a hyperparameter of the prior distribution. Then the marginal likelihood of $\mathbf{D}$ is given by
\[
p(\mathbf{D};a,C) \ \equiv \ \int \left\{ \prod_{k=1}^K \phi_{\sigma_k} \left( \hat{\tau}_k-\tau_k\right) \right\} d \pi_{a,C}(\bm{\tau}),
\]
where $\phi_{\sigma}(\cdot)$ is the density function of $N(0,\sigma^2)$. Marginal likelihood maximization selects $C_m$ that maximizes $\max_{a } p(\mathbf{D};a,C_m)$. In Section \ref{sec:empirical}, we use the following prior distribution of $\bm{\tau}$:
\[
\tau_0 \sim N(0,a^2) \ \ \text{and} \ \ (\xi_1, \ldots, \xi_K)|\tau_0 \sim U(\mathcal{T}_{-0,C}),
\]
where $\xi_k \equiv \tau_k - \tau_0$, $\mathcal{T}_{-0,C} \equiv \{ (\tau_1, \ldots, \tau_K) : (0, \tau_1, \ldots, \tau_K) \in \mathcal{T}_C \}$, and $U(\mathcal{T}_{-0,C})$ denotes the uniform distribution on $\mathcal{T}_{-0,C}$. For $C = 0.005, \, 0.010, \ldots, 0.095, \, 0.100$, we evaluate $\max_{a } p(\mathbf{D};a,C)$ and find that marginal likelihood maximization selects $C=0.040$,  $0.045$, $0.040$ for the target populations of Japan, the UK, and Peru, respectively.

\renewcommand{\theequation}{D.\arabic{equation}}
\setcounter{equation}{0}
%%%%%%%%%%%%%%%%%%%%%%%%%%%%%%%%%%%%%%%%%%%
\section{Treatments for COVID-19}\label{sec:covid-19}

We consider a drug approval decision for a COVID-19 treatment using the meta-database of randomized clinical trials provided by \citet{juul2020interventions}. There is an urgent gloabl need for evidence-based treatment of COVID-19. To search for effective treatments, numerous randomized clinical trials have been conducted in different countries across different demographic groups. At the time of writing, evidence as to the efficacy of various proposed treatments is mixed with limited precision of trial estimates.

We focus on Remdesivir, an antiviral medication known to be effective against viruses in the coronavirus family, such as Middle East Respiratory Syndrome (MERS) and Severe Acute Respiratory Syndrome (SARS). The effectiveness of Remdesivir in fighting Covid-19, however, remains undetermined and is a source of much controversy due to the conflicting nature of the existing evidence; the U.S. Food and Drug Administration approved emergency use of Remdesivir for Covid-19 patients, while the World Health Organization recommends against its use.

The meta-database of \cite{juul2020interventions} collates data from 33 RCT studies enrolling a total of 13,312 participants. Each study provides an estimate of the treatment effect compared with standard care or a placebo. Because we focus on the effects of Remdesivir on mortality rate, we use 6 estimates from 4 RCTs, \cite{beigel2020remdesivir}, \cite{who2021repurposed}, \cite{spinner2020effect}, and \cite{wang2020remdesivir}.

To form a vector of study characteristics, we include two covariates summarizing the average patients' characteristics in each study. They are the (standardized) mean (or median) age and the (standardized) proportion of female patients. Table \ref{tab:covid_estimates} lists the estimates, standard errors, and study characteristics.\footnote{\label{footnote:covid_female_proportion}Studies 2a--2c report the subgroup treatment effect estimates for three age subgroups ($<$50, 50--69, $\leq$70). Because we do not have detailed information about the age of these subgroups, we suppose that the mean age of these subgroups are 45, 60, and 75, respectively. \Copy{covid_study2}{In addition, we also assume that the female proportion is the same within these subgroups.}}

Similar to the previous section, we derive the minimax regret and minimax MSE rules with the following parameter space:
$$
\mathcal{T}_C \ \equiv \ \left\{ \bm{\tau} : |\tau_k - \tau_l| \leq C \|x_k-x_l\| \ \text{for $k,l = 0, 1, \cdots, K$}  \right\},
$$
with a prespecified Lipschitz constant $C \geq 0$. Because $K$ is small, leave-one-out cross-validation does not seem sensible. Hence, in this application, we set $C=0.01$ based on the WLS estimates. We consider hypothetical populations of interest whose characteristics range over 40--80 in terms of average age and 0.34--0.41 for the fraction of female patients.

%% Table 7 %%
\begin{table}[H]
\caption{Estimates, standard errors, and study characteristics.}
\begin{center}
\begin{tabular}{c c c c c c}
  \hline
  No. & Paper & Estimates & S.E. & Age & Female (\%) \\ \hline \hline
  1  & \cite{beigel2020remdesivir} &  0.038 & 0.021 & 58.9 & 0.356 \\
  2a & \cite{who2021repurposed} & -0.002 & 0.011 & 45.0 & 0.371 \\
  2b & \cite{who2021repurposed} &  0.005 & 0.013 & 60.0 & 0.371 \\
  2c & \cite{who2021repurposed} &  0.005 & 0.024 & 75.0 & 0.371 \\
  3  & \cite{spinner2020effect} &  0.030 & 0.021 & 57.0 & 0.389 \\
  4  & \cite{wang2020remdesivir} & -0.011 & 0.047 & 65.3 & 0.407 \\
  \hline
\end{tabular}
\end{center}\label{tab:covid_estimates}
\end{table}

Figure \ref{fig:covid} plots $\hat{\tau}_0(\bm{w}_{\text{minimax}})$ at each specified grid point in the space of characteristics of the target populations. In this figure, dark red means the value of $\hat{\tau}_0(\bm{w}_{\text{minimax}})$ is large and white means the value is small. Figure \ref{fig:covid} also plots covariate values of the meta-sample and the size of the plotted circle is proportional to the precision of the estimates, i.e.,  a smaller $\hat{\sigma}_k$ corresponds to a larger circle. This result implies that the minimax regret criterion recommends to treat Remdesivir for the age group greater than 50 years-old. In contrast, for some local populations below age 50 years-old, the minimax regret criterion recommends not to treat them with Remdesivir.

%%% Figure 7 %%%
\begin{figure}[H]
\centering
\includegraphics[width=9.5cm]{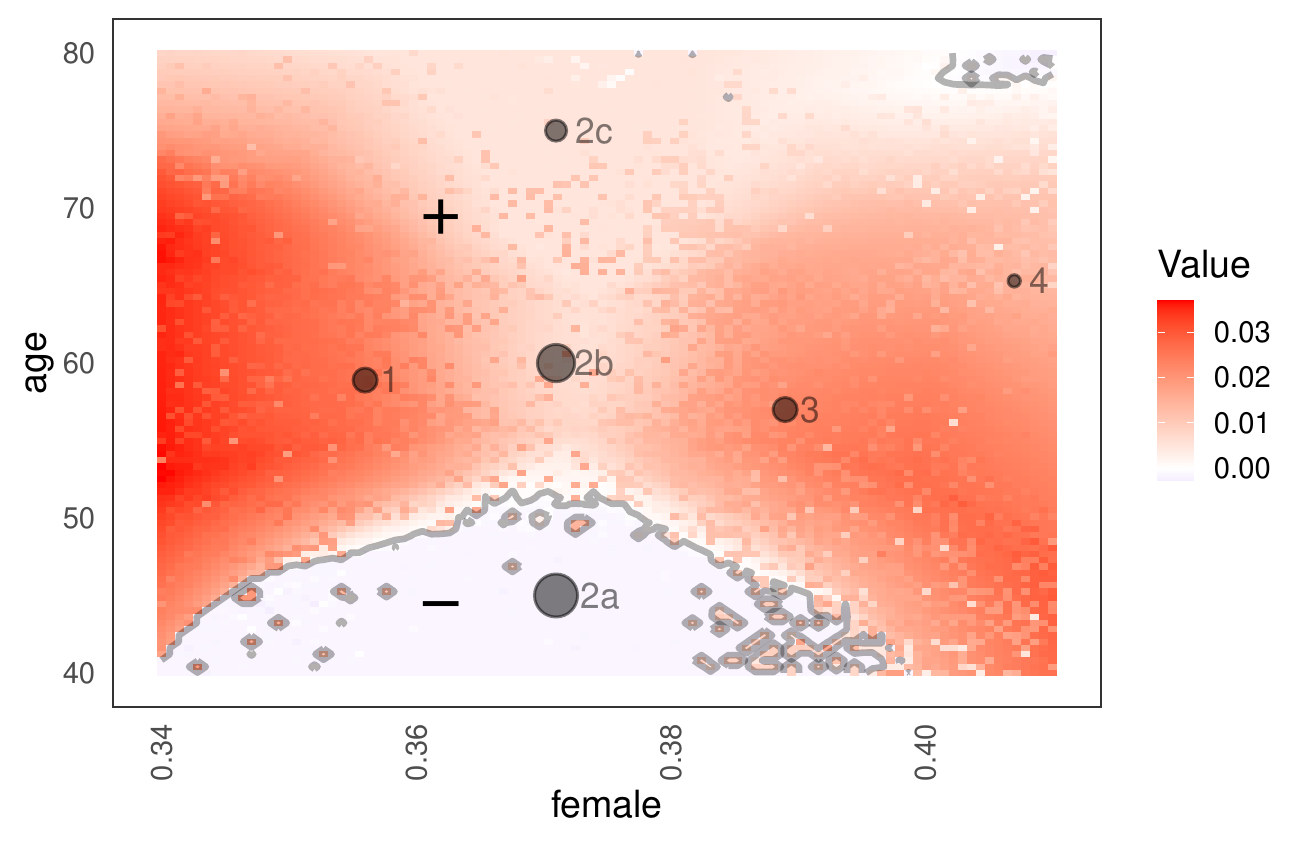}
\caption{{\footnotesize Heatmap of $\hat{\tau}_0(\bm{w}_{\text{minimax}})$ based on the values of $\hat{\tau}_0(\bm{w}_{\text{minimax}})$ computed at each $x_0$. Dark red areas correspond to regions of $x_0$ such that $\hat{\tau}_0(\bm{w}_{\text{minimax}})$ is positive and large. The white area corresponds to the region of $x_0$ such that $\hat{\tau}_0(\bm{w}_{\text{minimax}})$ is negative or near zero. The grey line is the boundary that separates the regions of positive and negative $\hat{\tau}_0(\bm{w}_{\text{minimax}})$. We also plot the covariate values of the meta-sample with the sizes of the plotted circles being proportional to the precision of the estimates.}}\label{fig:covid}
\end{figure}

\clearpage

\end{document}